\documentclass[final,onefignum,onetabnum]{siamonline220329}

\usepackage{lipsum}
\usepackage{amsfonts}
\usepackage{graphicx}
\usepackage{epstopdf}
\usepackage{algorithmic}
\ifpdf
  \DeclareGraphicsExtensions{.eps,.pdf,.png,.jpg}
\else
  \DeclareGraphicsExtensions{.eps}
\fi

\usepackage{cite}
\usepackage{dsfont}
\usepackage{graphicx}
\usepackage{subcaption}
\usepackage{amssymb}
\newtheorem{assumption}[theorem]{Assumption}

\usepackage{tikz}
\usepackage[utf8]{inputenc}
\usepackage{pgfplots}
\DeclareUnicodeCharacter{2212}{−}
\usepgfplotslibrary{groupplots,dateplot}
\usetikzlibrary{patterns,shapes.arrows}
\pgfplotsset{compat=newest}
\usepackage{standalone} % allow separate tikz figures into standalone tex files

\def\<{\langle}
\def\>{\rangle}

\def\Lout{L_{\text{out}}}
\newcommand{\normal}{\mathcal{N}}
\newcommand{\reals}{\mathbb{R}}
\newcommand{\iid}{{\rm i.i.d.}}

\newcommand{\beq}{\begin{equation}}
\newcommand{\eeq}{\end{equation}}

\newcommand{\sign}{\textrm{\sign}}
\newcommand{\E}{\mathbb{E}}
\newcommand{\Var}{\mathrm{Var}}
\newcommand{\Cov}{\mathrm{Cov}}

\newcommand{\bs}{\boldsymbol}
\newcommand{\mb}{\mathbb}
\newcommand\Tr{\mathrm{Tr}}

\newcommand\bX{\breve{X}}
\newcommand\bY{\breve{Y}}
\newcommand\tX{\widetilde{X}}
\newcommand\tY{\widetilde{Y}}

\newcommand{\hB}{\widehat{B}}
\newcommand{\hbeta}{\hat{\beta}}
\newcommand{\hR}{\widehat{R}}
\newcommand{\hz}{\hat{z}}

\newcommand{\tPsi}{\widetilde{\Psi}}
\newcommand{\bPsi}{\bar{\Psi}}

\newcommand{\tg}{\tilde{g}}
\newcommand\Tau{\mathrm{T}}
\newcommand\Mu{\mathrm{M}}
\newcommand\bu{\bar{u}}
\newcommand\bv{\bar{v}}
\newcommand\bc{\bar{c}}
\newcommand\bd{\bar{d}}
\newcommand\bh{\bar{h}}
\newcommand\be{\bar{e}}

\newcommand\tw{\tilde{w}}
\newcommand\tilz{\tilde{z}}

\newcommand{\tpartial}{\tilde{\partial}}
\newcommand\opt{\text{opt}}

\newcommand\AlaouiAMP{\text{ERKZJ-AMP}\;}

\usepackage{enumitem}
\setlist[enumerate]{leftmargin=.5in}
\setlist[itemize]{leftmargin=.5in}

\newsiamremark{remark}{Remark}
\newsiamremark{hypothesis}{Hypothesis}
\crefname{hypothesis}{Hypothesis}{Hypotheses}
\newsiamthm{claim}{Claim}

\headers{AMP for Pooled Data and Quantitative Group Testing}{N.~Tan, P.~Pascual~Cobo, J.~Scarlett, and R.~Venkataramanan}

\title{Approximate Message Passing with Rigorous Guarantees for \\ Pooled Data and Quantitative Group Testing\thanks{N.~Tan was supported by the Cambridge Trust and the Harding Distinguished Postgraduate Scholars Programme Leverage Scheme. P.~Pascual~Cobo was supported by a Engineering and Physical Sciences Research Council Doctoral Training Partnership. J.~Scarlett was supported by the Singapore National Research Foundation (NRF) under grant number A-0008064-00-00.}}

\author{Nelvin Tan\thanks{Department of Engineering, University of Cambridge (\email{tcnt2@cam.ac.uk}, \email{pp423@cam.ac.uk}, \email{rv285@cam.ac.uk}).}
\and Pablo Pascual Cobo\footnotemark[2]
\and Jonathan Scarlett\thanks{Department of Computer Science, Department of Mathematics, and Institute of Data Science, National University of Singapore (\email{scarlett@comp.nus.edu.sg}).}
\and Ramji Venkataramanan\footnotemark[2]}

\usepackage{amsopn}
\DeclareMathOperator{\diag}{diag}

\ifpdf
\hypersetup{
  pdftitle={AMP for pooled data},
  pdfauthor={N. Tan, P. {Pascual Cobo}, J. Scarlett, and R. Venkataramanan}
}
\fi

\begin{document}

\maketitle

% REQUIRED
\begin{abstract}
    In the \textit{pooled data} problem, the goal is to identify the categories associated with a large collection of items via a sequence of pooled tests. Each pooled test reveals the number of items of each category within the pool. We study an approximate message passing (AMP) algorithm for estimating the categories and rigorously characterize its performance, in both the noiseless and noisy settings.
    For the noiseless setting, we show that the  AMP algorithm is equivalent to one recently proposed by El Alaoui et~al. Our results provide a rigorous version of their performance guarantees, previously obtained via non-rigorous techniques.
    For the case of pooled data with two categories, known as \textit{quantitative group testing} (QGT),  we use the AMP guarantees to compute precise limiting values of the false positive rate  and the false negative rate. Though the pooled data problem and QGT are both instances of estimation in a linear model, existing AMP theory cannot be directly applied since the design matrices  are binary valued.  The key technical ingredient in our analysis is a rigorous asymptotic characterization of AMP for generalized linear models defined via generalized white noise design matrices. This result, established using a recent universality result of Wang et al., is of independent interest.  Our theoretical results are validated by numerical simulations. 
    For comparison, we propose estimators based on convex relaxation and iterative thresholding, without providing theoretical guarantees. The simulations indicate that AMP consistently outperforms these estimators.
    %The simulations indicate that AMP  outperforms the convex estimator for noiseless pooled data and QGT, but the convex estimator performs slightly better for noisy pooled data with three categories when the number of observations is small.
\end{abstract}

% REQUIRED

% REQUIRED
% \begin{MSCcodes}
%  62F12, 62F15, 65K05, 68P30 
% \end{MSCcodes}

\section{Introduction}

Consider a large population of items, each of which has an associated category. The associated categories are initially unknown, and are to be estimated based on \textit{pooled tests}. Each pool consists of some subset of the population, and the test outcome reveals the \textit{total number of items} corresponding to each category that are present in the pool (but not the individual categories of the items). This problem, which we refer to as the \textit{pooled data} problem, is of interest in applications such as machine learning \cite{Ach20}, computational biology \cite{Cao14}, multi-access communication \cite{Mar21}, traffic monitoring \cite{Wan15b}, and network tomography \cite{Che12}.

\subsection{Problem Setup} \label{subsec:pooled_setup}

We begin by introducing the noiseless setting of the problem. Writing $[n]$ for the set $\{1,\dots,n\}$, let $\tau:[p]\rightarrow[L]$ be an assignment of $p$ variables to $L$ categories. For example, $\tau(1)=L$ means item 1 is assigned to category $L$. The output of each test tells us the number of items from each category present in the test, which can be viewed as a histogram. For each test $i\in[n]$,  the queried sub-population is denoted by $S_i\subset[p]$, and the test outcome from the pooled sub-population $S_i$ is denoted by
\begin{align*}
    Y_{i,:}=\big(|\tau^{-1}(1)\cap S_i|,\dots,|\tau^{-1}(L)\cap S_i|\big),
\end{align*}
where $\tau^{-1}(1)$ represents the set of items in category 1, $|\tau^{-1}(1)\cap S_i|$ represents the number of items in $S_i$ of category 1, and so on.  (The notation $Y_{i,:}$ indicates that the vector is viewed as the $i$th row of a matrix $Y \in \reals^{n \times L}$.)

We denote the vector of proportions of assigned values (i.e., the empirical distribution of the categories) by
\begin{align}
    \hat{\pi}=\frac{1}{p}\big(|\tau^{-1}(1)|,\dots,|\tau^{-1}(L)|\big),
    \label{eq:true_proportion_vec}
\end{align}
The signal to be estimated is denoted as $B\in\mb{R}^{p\times L}$, with $j$th row $B_{j,:}=e_{\tau(j)}\in\mb{R}^L$ for $j\in[p]$. Here $e_{\tau(j)}$ is the vector with a $1$ in position $\tau(j)$ and $0$ elsewhere.   For example, if $\tau(j)=2$, then $B_{j,:}=[0,1,0,\dots,0]^\top$.

An equivalent linear-algebraic formulation will turn out to be more convenient to work with.  In this formulation, we have a design matrix $X\in\mb{R}^{n\times p}$, where $X_{ij}=\mathds{1}\{j\in S_i\}$, for $i \in [n]$, $j \in [p]$. Then our histogram query can be succinctly rewritten as
\begin{align}
    Y_{i,:}
    =\sum_{j=1}^pX_{ij}B_j
    =B^\top X_{i,:}
    \, \in\mb{R}^L, \quad i \in [n], \label{eq:noiseless_pooled_data}
\end{align}
where $Y_{i,:}\in\mb{R}^L$ is the $i$th row of $Y$ represented as a column vector, and $X_{i,:}\in\mb{R}^p$ is the $i$th row of $X$ represented as a column vector. More generally, the model can be written as $Y=XB\in\mb{R}^{n\times L}$.

For the noisy setting, we study the situation where the observed pooled measurements are corrupted by additive independent noise. Specifically, we have
\begin{align}
    Y_{i,:}&=B^\top X_{i,:}+\Psi_{i,:} \, \in\mb{R}^L, \quad i \in [n],\label{eq:noisy_pooled_data}
\end{align}
where $\Psi_{i,:}$ is  the zero-mean noise (e.g., Gaussian or uniform) for the $i$th test. From \eqref{eq:noiseless_pooled_data} and \eqref{eq:noisy_pooled_data},  we observe that the pooled data problem can be viewed as an instance of compressed sensing, with additional constraints on the design matrix and the signal. Specifically, in the pooled data model $Y = X B + \Psi$, the $n \times p$ design matrix $X$ is binary-valued, as is the  $p \times L$ signal matrix $B$  which has exactly one non -zero value in each row.  

The special case of pooled data with two categories, called \emph{quantitative group testing} (QGT), also known as the coin weighing problem \cite{Bsh09}, has been studied in a number of recent works \cite{Kar19a,Kar19b,Sol23,Mas23, Hah23}. QGT is typically represented using a binary signal vector $\beta \in \reals^{p}$, with ones in positions (items) corresponding to the first category, and zeros in positions corresponding to the second category. Therefore, the goal is to recover $\beta$ from the observed vector 
$Y=X\beta + \Psi$.

In this paper, we study the natural `linear category' regime \cite{Ala18} where no category is dominant nor negligible, i.e., as $p$ grows, the fraction of items assigned to each of the $L$ categories is $\Theta(1)$. We consider random design matrices, specifically,  the \textit{random dense} setting, where the  $X_{ij}\stackrel{\iid}{\sim}\text{Bernoulli}(\alpha)$ for some fixed $\alpha\in(0,1)$, for $i \in [n], j \in [p]$. We consider the high dimensional regime where both $n,p \to \infty$ with $n/p\rightarrow\delta$. Note that the noiseless version of the problem becomes trivial if $n=p$ and $\alpha\in(0,1/2]$, since the random binary square matrix $X$ will be invertible with high probability \cite[Theorem 4.8]{Gui21}. Hence, we assume $\delta<1$ for the noiseless setting.

\subsection{Approximate Message Passing}
In this paper, we consider Approximate Message Passing (AMP) techniques to recover the categories of each item (or equivalently, the signal matrix $B \in \reals^{p \times L}$) from $Y \in \reals^{n \times L}$ produced according to \eqref{eq:noiseless_pooled_data} or \eqref{eq:noisy_pooled_data}. AMP is a family of iterative algorithms that can be tailored to take advantage of structural information about the signals and the model, e.g., a known prior on the signal matrix, or on the proportion of observations that come from each signal. AMP algorithms were first proposed for the standard linear model (compressed sensing) \cite{Kab03,Bay11,Don09,Krz12,Ran19}, but have since been applied to a range of statistical  problems, including estimation in generalized linear models and their variants \cite{Ran11,Sch14,Bar19,Ma19,Sur19,Mai20,Pan20,Ven22},  and low-rank matrix estimation \cite{Des14,Fle18,Les17,Mont21,Fan22}. In all these settings, under suitable model assumptions, the performance of AMP in the high-dimensional limit is characterized by a succinct deterministic recursion called \emph{state evolution}. Furthermore, the state evolution characterization has been used to show that AMP achieves Bayes-optimal performance for some models \cite{Des14,Don13,Mont21,Cob23}. The monograph \cite{Fen21} contains a survey of AMP algorithms for various models defined via Gaussian matrices.

Though AMP for compressed sensing has been widely studied, including for matrix-valued signals \cite{Zin12, Jav13}, these results assume a Gaussian design matrix, and hence cannot be applied to the pooled data problem which has a binary-valued design matrix.  We address this issue in this paper, by establishing a universality result for an AMP algorithm and using it to obtain rigorous guarantees for the pooled data problem. The pooled data problem can be cast as a special case of a matrix generalized linear model (\emph{matrix GLM}) by centering and rescaling the data.  In a matrix GLM,  the goal is to estimate a signal matrix $B \in \reals^{p \times L}$ from an observed matrix $\tY := (\tY_1,\dots,\tY_n)^\top\in\mb{R}^{n\times \Lout}$, whose $i$th row $\tY_{i,:}$ is generated as:
\begin{align}
    \tY_{i,:}&=q(B^\top \tX_{i,:} \, , \, \tPsi_{i,:})\in\mb{R}^{\Lout},
    \quad i\in[n]. \label{eq:matrix-GLM}
\end{align}
Here $\tPsi \in \reals^{n \times L_\Psi}$ is a matrix of unobserved auxiliary variables (with $i$th row $\tPsi_{i,:}$), $\tX$ is a design matrix (with $i$th row $\tX_{i,:}$), and $q:\mb{R}^L \times\mb{R}^{L_\Psi}\rightarrow\mb{R}^{\Lout}$ is a known function. 

An AMP algorithm for estimation in the matrix GLM was recently studied in \cite{Tan23c}, for i.i.d.~Gaussian design matrices. In the pooled data setting, the modified design matrix $\tX$ is not Gaussian, but is an instance of a  \emph{generalized white noise} matrix \cite{Wan22}. Wang et al. \cite{Wan22} analyzed an abstract AMP recursion for generalized white noise matrices and proved that the state evolution remains valid in this setting.
We show that the AMP algorithm in \cite{Tan23c} for the matrix GLM can be reduced to the abstract AMP in \cite{Wan22}, using which we establish a rigorous state evolution for the pooled data problem.

\subsection{Main Contributions}

\paragraph{AMP universality} We establish a rigorous state evolution result for the AMP algorithm applied to the matrix GLM in \eqref{eq:matrix-GLM}, where $\tX$ is a  generalized white noise design matrix (see Definition \ref{def:gen_white_noise_matrix}). Theorem \ref{thm:GAMP} gives a rigorous characterization of the joint empirical distribution of the AMP iterates in the high-dimensional limit as $n,p \to \infty$ with  $n/p \to \delta$, for a constant $\delta >0$. This allows us to compute exact asymptotic formulas for performance measures such as the mean-squared error and overlap between the signals and their estimates. Theorem \ref{thm:GAMP} generalizes the state evolution result in \cite[Theorem 1]{Tan23c} for i.i.d.~Gaussian designs, guaranteeing that the same AMP algorithm and state evolution remain valid for a much broader class of GLMs. 

\paragraph{Pooled data} We show that after centering and suitable rescaling, the pooled data problem is a special case of the matrix-GLM with a generalized white noise design matrix. Therefore, rigorous performance guarantees for AMP can be readily obtained from  Theorem \ref{thm:GAMP}. Furthermore, we show that in the noiseless setting, a special case of our AMP state evolution is equivalent to the one given by El Alaoui et al. \cite{Ala18},  thereby making their performance guarantees rigorous; see Proposition \ref{prop:eqv_of_AMPs} and Appendix \ref{sec:eqv_of_AMP_and_SE}.

In Section \ref{sec:sim_pool_data}, we provide numerical simulations to validate the theoretical results. For comparison, we propose alternative estimators based on convex relaxation and iterative thresholding, without providing theoretical guarantees. Our simulations indicate that for both noiseless and noisy setting, the AMP algorithm outperforms the other estimators.
% Our simulations indicate that the optimization based method outperforms the AMP algorithm for lower sampling ratios $\delta$, and the thresholding method slightly outperforms the AMP algorithm for higher sampling ratios $\delta$ when the noise level is high.

\paragraph{Quantitative group testing} In Corollary \ref{cor:FPR_FNR}, we provide rigorous guarantees for the limiting false positive rate (FPR) and false negative rate (FNR) achieved by AMP. We also show that, with a simple modification to the AMP algorithm, it can be used to vary the trade-off between the FPR and the FNR. Numerical simulations show that AMP outperforms the other estimators for a wide range of QGT scenarios.

\subsection{Related Work} \label{subsec:related_work}

\paragraph{Information-theoretic results for pooled data} We first summarize works that also study the linear category regime, where each category constitutes a $\Theta(1)$ proportion of items.
\begin{itemize}
    \item For the noiseless setting, Grebinski and Kucherov \cite{Gre00} proved that with exponential time algorithms, the pooled data problem can be solved with $\Theta(\frac{p}{\log p})$ tests. More recently, under the condition that $\pi$ is the uniform distribution and $X_{ij}\sim\text{Bernoulli}(\alpha)$ for $\alpha\in(0,1)$, Wang et al. \cite{Wan16} presented a lower bound on the minimum number of tests $n$ for reliable recovery; they showed that if $n<\frac{\log L}{L-1}\frac{p}{\log p}$, then the signal $B$ cannot be uniquely determined. These results were later generalized and sharpened in \cite{Ala19},     where it was shown that  $\gamma_{\text{low}}\frac{p}{\log p}<n<\gamma_{\text{up}}\frac{p}{\log p}$ tests are necessary and sufficient for $B$ to be uniquely determined, where $\gamma_{\text{low}}$ and $\gamma_{\text{up}}$ are constants that depend on $\pi$ and $L$.      This gap was closed by Scarlett and Cevher \cite{Sca17}, who showed that $\gamma_{\text{up}}\frac{p}{\log p}$ tests are also necessary for $B$ to be unique.    %
    \item For the noisy setting, Scarlett and Cevher \cite{Sca17} developed a general framework for understanding variations of the pooled data problem with random noise. Specifically, under the noise model \eqref{eq:noisy_pooled_data} with Gaussian noise $\Psi_{i,:}\stackrel{\iid}{\sim}\normal_L(0,p\sigma^2 I_L)$, where $I_L$ is an $L\times L$ identity matrix, they showed that we require $n=\Omega(p\log p)$ for exact recovery, which is super-linear in the number of items $p$, in contrast with the sub-linear $\Theta\big(\frac{p}{\log p}\big)$ behaviour observed in the noiseless case. 
\end{itemize}

\paragraph{Algorithmic results for pooled data} Wang et al.~\cite{Wan16} proposed a deterministic design matrix and an algorithm that recovers $B$ with $n=\Omega\big(\frac{p}{\log p}\big)$ tests. In the sublinear category regime, there is one dominant category with $p-o(p)$ items, and the remaining categories have $k=o(p)$ items. For random designs, an efficient algorithm for the sublinear category regime was recently developed in \cite{Han22}, and shown to achieve exact recovery with $O(k)$ tests when $k=\Theta\big(p^{\kappa}\big)$, for a constant $\kappa \in (0,1)$. For the linear category regime, El Alaoui et al.~\cite{Ala18} proposed an AMP algorithm for a random dense design, and characterized the asymptotic behaviour of the algorithm in the limit as $n,p \to \infty$ with $n/p\rightarrow\delta$; As mentioned above, rigorous performance guarantees were not provided.

\paragraph{Algorithmic results for QGT} For the case of $L=2$, the existing work has largely focused on the sublinear category regime. Algorithms inspired from coding theory \cite{Kar19a,Kar19b,Sol23,Mas23} and thresholding \cite{Hah22b} require $\Omega(k\log p)$ tests for exact recovery. The number of tests can be reduced to $O(k)$ by specializing the algorithm from \cite{Han22} (for general $L$) to the case of $L=2$ when $k=\Theta\big(p^{\Omega(1)}\big)$. Recently, Li and Wang \cite{Li21} and Hahn-Klimroth et al.~\cite{Hah23} studied noisy versions of QGT.

\paragraph{Boolean group testing} In Boolean Group Testing (BGT), there are two categories of items, usually referred to as defectives and non-defectives, and the outcome of the pooled test is $1$ if it contains at least one defective, and zero otherwise. With the defectives corresponding ones in the binary signal vector $\beta$, the testing model is:
\begin{align}
    Y_i = \mathds{1}\{X_{i,:}^\top\beta>0\}, \quad i \in [n]. \label{eq:BGT}
\end{align} 
BGT can be viewed as a less informative version of QGT, where test yields at most one bit of information. BGT is of interest in a range of applications including medical testing, DNA sequencing, and communication protocols \cite[Section 1.7]{Ald19}, and more recently,  in testing for COVID-19 \cite{Ald21,Wan23}. BGT has been widely studied, including variants of the model \eqref{eq:BGT} under practical constraints \cite{Oli22, Tan21, Tan23b, Pri23, Gan19, Tan20}, and with noise \cite{Sca18a,Sca18b,Geb21,Pri23}.  We refer the interested reader to the survey by Aldridge et al. \cite{Ald19}. Belief propagation and AMP algorithms for noisy BGT with side information were recently proposed in \cite{Kar22} and \cite{Cao23}, without theoretical guarantees. In Section \ref{sec:disc} we discuss the challenges in extending the  AMP guarantees for pooled data and QGT to the BGT setting. 

\paragraph{AMP Universality} In addition to AMP universality results for generalized white noise matrices, Wang et al. \cite{Wan22} also gave similar results for generalized invariant ensembles. Other AMP universality results, for sub-Gaussian matrices and semi-random matrices,  were recently established in \cite{CheL21} and \cite{Dud22b}, respectively.

\section{Preliminaries} \label{sec:prelim}

\paragraph{Notation} We write $[n:m]$ for $[n,n+1,\dots,m]$ where $n<m$. All vectors (including those corresponding to rows of matrices) are assumed to be column vectors unless otherwise stated. For vectors $a,b\in\mb{R}^n$, $\langle a,b\rangle=a^\top b\in\mb{R}$ is the inner product, $a\odot b\in\mb{R}^n$ is their entry-wise product, and  $\langle a\rangle=\frac{1}{n}\sum_{i=1}^n a_i$ denotes the empirical average of the entries of $a$. Matrices are denoted by upper case letters, and given a matrix $A$, we write $A_{i,:}$ for its $i$th row and $A_{:,j}$ for its $j$th column. We write $A_{[r_1: r_2], [c_1 :c_2]}$ for the submatrix consisting of rows $r_1$ to $r_2$ and columns $c_1$ to $c_2$. The Frobenius norm is denoted by $\|A\|_F$ and the operator norm is denoted by  $\|A\|_{\text{op}}$. For $r\in[1,\infty)$ and a vector $a=(a_1,\dots,a_n)\in\mb{R}^n$, we write $\|a\|_r$ for the $\ell_r$-norm, so that $\|a\|_r=\big(\sum_{i=1}^n|a_i|^r\big)^{1/r}$.  We use $e_l$ to denote the one-hot vector with a 1 in position $l$, $1_p$ for the vector of $p$ ones, $0_p$ for the vector of $p$ zeros, and $I_p$ for the $p\times p$ identity matrix. Given random variables $U,V$, we write $U \stackrel{d}{=} V$ to denote equality in distribution. Throughout the paper, the function $\log(\cdot)$ has base $e$, and we use of Bachmann-Landau asymptotic notation (i.e., $O$, $o$, $\Omega$, $\omega$, $\Theta$).

\paragraph{Almost sure convergence} This is denoted  using the symbol $\stackrel{a.s.}{\rightarrow}$. Let $\{A^n\}$ be a sequence of random elements taking values in a Euclidean space $E$. We say that $A^n$ converges almost surely to a deterministic limit $a\in E$, and write $A^n\stackrel{a.s.}{\rightarrow}a$, if $\mb{P}[\lim_{n\rightarrow\infty}A^n=a]=1$.

\paragraph{Wasserstein convergence} We review the definition of \cite{Wan22}. For a vector $a\in\mb{R}^n$ and a random variable $A\in\mb{R}$, we write $a\stackrel{W_r}{\rightarrow}A$ as $n\rightarrow\infty$, for the Wasserstein-$r$ convergence of the empirical distribution of the entries of $a$ to the law of $A$. More generally, for vectors $a^1,\dots,a^k\in\mb{R}^n$ and a random vector $(A^1,\dots,A^k)\in\mb{R}^k$, we write
\begin{align*}
    (a^1,\dots,a^k) \, \stackrel{W_r}{\rightarrow} \, (A^1,\dots,A^k)
    \text{ as $n\rightarrow\infty$},
\end{align*}
for the Wasserstein-$r$ convergence of the empirical distribution of rows of $(a^1,\dots,a^k)\in\mb{R}^{n\times k}$ to the joint law of $(A^1,\dots,A^k)$. This means, for any continuous function $\phi:\mb{R}^k\rightarrow\mb{R}$ and input vector $(a_i^1,\dots,a_i^k)\in\mb{R}^k$ satisfying the \textit{polynomial growth} condition
\begin{align}
    |\phi(a_i^1,\dots,a_i^k)|
    \leq C\big(1+\|(a_i^1,\dots,a_i^k)\|_2^r\big)
    \text{ for a constant $C>0$},
    \label{eq:poly_growth_cond}
\end{align}
we have as $n\rightarrow\infty$
\begin{align}
    \frac{1}{n}\sum_{i=1}^n\phi(a^1_i,\dots,a^k_i)
    \rightarrow\E\big[\phi(A^1,\dots,A^k)\big].
    \label{eq:sum_to_exp}
\end{align}
We write $(a^1,\dots,a^k) \, \stackrel{W}{\rightarrow} \, (A_1,\dots,A_k)
    \text{ as $n\rightarrow\infty$}$
to mean that the above Wasserstein-$r$ convergences hold for every order $r\geq1$.

%%%%%%
\section{AMP for Matrix GLM with Generalized White Noise Design}\label{sec:AMP_GWN}

We begin with the definition of a generalized white noise matrix.

\begin{definition} \label{def:gen_white_noise_matrix} \textup{\cite[Definition 2.15]{Wan22}}
A generalized white noise matrix $\tX\in\mb{R}^{n\times p}$ with a (deterministic) variance profile $S\in\mb{R}^{n\times p}$ is one satisfying the following conditions:
\begin{itemize}
    \item All entries $\tX_{ij}$ are independent.
    \item Each entry $\tX_{ij}$ has mean 0, variance $n^{-1}S_{ij}$, and higher moments satisfying, for each integer $m\geq3$,
    \begin{align*}
        \lim_{n,p\rightarrow\infty} p\cdot\max_{i\in[n]}\max_{j\in[p]}\E\Big[|\tX_{ij}|^m\Big]=0.
    \end{align*}
    \item For a constant $C>0$,
    \begin{align*}
        \max_{i\in[n]}\max_{j\in[p]}S_{ij}\leq C,\quad
        \lim_{n,p\rightarrow\infty}\max_{i\in[n]}\Big|\frac{1}{p}\sum_{j=1}^pS_{ij}-1\Big|=0,\quad
        \lim_{n,p\rightarrow\infty}\max_{j\in[p]}\Big|\frac{1}{n}\sum_{i=1}^nS_{ij}-1\Big|=0.
    \end{align*}
\end{itemize}
\end{definition}
Note that Definition \ref{def:gen_white_noise_matrix} simplifies greatly for the case where $S_{ij}=1$, for all $(i,j)\in[n]\times[p]$. In this case, the entries are all i.i.d.~with variance $1/n$, the third condition in the definition  is trivially satisfied, and  the second condition requires  moments of order 3 and higher to decay faster than $1/p$.

\paragraph{Model assumptions}  Consider the matrix GLM model \eqref{eq:matrix-GLM} defined via a generalized white noise design matrix $\tX$. The signal matrix $B \in \mb{R}^{p \times L}$ and the auxiliary variable matrix $\tPsi \in \mb{R}^{n \times L_{\Psi}}$ are both independent of $\tX$. As $p\to\infty$, we assume that $n/p \to \delta$, for some positive constant $\delta$. As $p \to \infty$, the empirical distributions of the  rows of the signal matrix and the auxiliary variable matrix both converge in Wasserstein distance to well-defined limits. More precisely, there exist random variables $\bar{B}\sim P_{\bar{B}}$ (where $\bar{B}\in\mb{R}^{L}$) and $\bar{\Psi}\sim P_{\bar{\Psi}}$ (where $\bar{\Psi}\in\mb{R}^{L_\Psi}$) with $B\stackrel{W}{\rightarrow}\bar{B}$ and $\tPsi\stackrel{W}{\rightarrow}\bar{\Psi}$, respectively. \label{page:model_assump}

In this section, we allow general priors $P_{\bar{B}}$ and $P_{\bar{\Psi}} $ for the signal and auxiliary matrices, before specializing to the pooled data setting in the following section.

\subsection{Algorithm} \label{sec:AMP_algo} 

We present the algorithm below, before stating the main result in the next subsection. For the rest of this paper, we refer to it as the \textit{matrix-AMP} algorithm.

In each iteration $k \ge 1$, the matrix-AMP algorithm iteratively produces estimates $\hB^k$ and $\Theta^k$ of $B\in\mb{R}^{p\times L}$ and $\Theta:=\tX B\in\mb{R}^{n\times L}$, respectively. Starting with an initialization $\hB^0\in\mb{R}^{p\times L}$ and defining $\hR^{-1} := 0\in\mb{R}^{n\times L}$, for iteration $k \ge 0$ the algorithm computes:
\begin{align}
\begin{split}
        & \Theta^k=\tX\hB^k-\sum_{m=0}^{k-1}\hR^{m}(F^{k,m+1})^\top, \quad \hR^k=g_k\big(\Theta^0,\dots,\Theta^k,\tY\big), \\ 
        & B^{k+1}=\tX^\top\hR^k-\sum_{m=0}^k\hB^m (C^{k,m+1})^\top, \quad \hB^{k+1}=f_{k+1}\big(B^1,\dots,B^{k+1}\big).
\end{split}
\label{eq:GAMP}
\end{align}
Here the functions $g_k:\mb{R}^{L(k+1)}\times\mb{R}^{\Lout}\rightarrow\mb{R}^{L}$ and  $f_{k+1}:\mb{R}^{L(k+1)}\rightarrow\mb{R}^{L}$ act row-wise on their matrix inputs, and the matrices $C^{k,m+1}, F^{k,m+1} \in \reals^{L \times L}$ are defined as
\begin{align*}
    C^{k,m+1}= \frac{1}{n}\sum_{i=1}^n\partial_{m+1}g_k(\Theta_{i,:}^0 \, ,\dots ,\Theta_{i,:}^k \, ,\tY_{i,:}), \quad
    F^{k,m+1}=\frac{1}{n}\sum_{j=1}^p\partial_{m+1} f_{k}(B_{j,:}^1 \, ,\dots,B_{j,:}^{k}), 
    \label{eq:CkF_k1_def}
\end{align*}
where $\partial_{m+1}g_k, \partial_{m+1}f_{k}$ denote the Jacobians of $g_k, f_{k}$ with respect to their $(m+1)$th arguments. We note that the time complexity of each iteration of \eqref{eq:GAMP} is $O(npL)$.

\begin{remark}
The matrix-AMP algorithm in \eqref{eq:GAMP} is a general form of the following matrix-AMP algorithm presented in \cite{Tan23c,Fen21} with one step of memory:
\begin{align}
\begin{split}
   &  \Theta^k
    =\tX\hB^k-\hR^{k-1}(F^k)^\top,
    \quad
    \hR^k=g_k(\Theta^k,\tY),
    \quad
    C^k=\frac{1}{n}\sum_{i=1}^ng_k'(\Theta_{i,:}^k,\tY_{i,:}),
    \\
   &  B^{k+1}
    =\tX^\top\hR^k-\hB^k(C^k)^\top,
    \quad
    \hB^{k+1}
    =f_{k+1}(B^{k+1}),
    \quad
    F^{k+1}
    =\frac{1}{n}\sum_{j=1}^pf_{k+1}'(B_{j,:}^{k+1}).
\end{split}
\label{eq:memoryless_GAMP}    
\end{align}
In \eqref{eq:memoryless_GAMP}, the functions $g_k$ and $f_{k+1}$ only act on the current iterates $\Theta^k$ and $B^{k+1}$, respectively, whereas in the more general version \eqref{eq:GAMP},   these functions may use all the past iterates.
\end{remark}
\paragraph{State evolution} The `memory' terms $-\sum_{m=0}^{k-1}\hR^{m}(F^{k,m+1})^\top$ and $-\sum_{m=0}^k\hB^m(C^{k,m+1})^\top$ in \eqref{eq:GAMP} debias the iterates $\Theta^k$ and $B^{k+1}$, ensuring that their joint empirical distributions are accurately captured by state evolution (SE) in the high-dimensional limit. Theorem \ref{thm:GAMP} below shows that for each $k \ge 1$,  the joint empirical distribution of the rows of $B^1,\dots,B^k$ converges to the joint distribution of $\Mu^{1}_B\bar{B}+G^1_B, \dots ,\Mu^{k}_B\bar{B}+G^k_B \in \reals^L$, where $\bar{B}$ is the random variable representing the limiting distribution of the rows of the signal matrix $B$, and 
$(G^1_B,\dots,G^k_B)\sim \normal(0,\Tau^k_B)$ are jointly Gaussian and independent of $\bar{B}$. The deterministic matrices $\Mu^{k}_B \in \reals^{L\times L}$ and $\Tau^k_B\in\mb{R}^{Lk\times Lk}$ are defined below. The result implies that the joint  empirical distribution of the rows of $\hB^1,\dots,\hB^k$ converges to the joint distribution of 
$$
f_1\big(\Mu^{1}_B\bar{B}+G^1_B\big),
\dots,
f_k\big(\Mu^{1}_B\bar{B}+G^1_B,\dots,\Mu^{k}_B\bar{B}+G^k_B\big).
$$
Thus, $f_k$ can be viewed as a denoising function that can be tailored to take advantage of the prior on $\bar{B}$. Theorem \ref{thm:GAMP} also shows that the joint empirical distribution of the rows of $\Theta, \Theta^0,\dots,\Theta^k$ converges to $\normal(0, \Sigma^k)$, where $\Sigma^{k} \in \reals^{L(k+2) \times L(k+2)}$ is defined below.

We now describe the state evolution recursion defining the matrices $\Mu^{k}_B\in\reals^{L\times L}$, $\Tau^k_B\in\mb{R}^{Lk\times Lk}$, and $\Sigma^{k} \in \reals^{L(k+2) \times L(k+2)}$. Recalling  that the observation $\tY$ is generated via the function $q$ according to \eqref{eq:matrix-GLM}, it is convenient to rewrite $g_k$ in \eqref{eq:GAMP} in terms of another function $\tg_k: \reals^L \times \reals^{L(k+1)} \times \reals^{L_{\Psi}} \to \reals^L$ defined as:
\begin{align}
    \tg_k(z, z^0, \dots, z^k, v) := g_k(z^0, \dots, z^k, q(z, v)).
    \label{eq:g_tilde_k_def}
\end{align}
We write $\partial_1 \tg_k$ for the partial derivative (Jacobian) of $\tg_k$ with respect to its first argument $Z\in\mb{R}^{L}$, so it is an $L\times L$ matrix.  State evolution is initialized with  $\Sigma^0 \in \reals^{2L \times 2L}$, defined below in \eqref{eq:Sig0_def}. Then, for $k \ge 0$, given $\Sigma^k \in \reals^{L(k+2) \times L(k+2)}$, we take $(Z,Z^0,\dots,Z^k)^\top \sim \normal(0,\Sigma^k)$ to be independent of $\bar{\Psi}\sim P_{\bar{\Psi}}$ and compute:
\begin{align}
    &\Mu^{k+1}_{B} =\E[\partial_1 \tg_k(Z,Z^0,\dots,Z^k,\bar{\Psi})], 
    \label{eq:SE_Mk1B} \\ 
    & \Tau^{k+1}_{B}=
    \begin{bmatrix}
        \E[\tg_0\tg_0^\top]  & \dots & \E[\tg_0\tg_k^\top] \\
        \vdots & \dots & \vdots \\
        \E[\tg_k\tg_0^\top]  & \dots & \E[\tg_k\tg_k^\top] 
    \end{bmatrix},
    \label{eq:SE_Tk1B}
\end{align}
where $\tg_0:=\tg_0(Z,Z^0,\bar{\Psi})$ and $\tg_r:=\tg_r(Z,Z^0,\dots,Z^r,\bar{\Psi})$ for $r\in\{1,\dots,k\}$, and 
\begin{align}
    &\Sigma^{k+1} =
    \begin{bmatrix}
    \Sigma_{(1,1)}^{k+1} & \dots & \Sigma_{(1,k+3)}^{k+1} \\
    \vdots & \dots & \vdots \\
    \Sigma_{(k+3,1)}^{k+1} & \dots & \Sigma_{(k+3,k+3)}^{k+1}
    \end{bmatrix},
    \label{eq:SE_Sigk1}
\end{align}
where the $L \times L$ submatrices constituting $\Sigma^{k+1} \in \reals^{L(k+3) \times L(k+3)}$ are given by
\begin{align*}
    \begin{bmatrix}
        \Sigma_{(1,1)}^{k+1} & \Sigma_{(1,2)}^{k+1} \\
        \Sigma_{(2,1)}^{k+1} & \Sigma_{(2,2)}^{k+1}
    \end{bmatrix}
    =\Sigma^0,
\end{align*}
where $\Sigma^0$ is defined in \eqref{eq:Sig0_def}, and the following where $r,s\in\{3,\dots,k+3\}$:
\begin{align}
        \Sigma_{(1,r)}^{k+1} 
        = \big( \Sigma_{(r,1)}^{k+1} \big)^{\top} 
        &= \frac{1}{\delta} \E\Big[\bar{B}f_{r-2}(\Mu^{1}_{B}\bar{B}+G^{1}_B,\dots,\Mu^{r-2}_{B}\bar{B}+G^{r-2}_B)^\top\Big],
        \label{eq:Sigma_12_21_def} \\
        \Sigma_{(r,s)}^{k+1}
        =\big(\Sigma_{(s,r)}^{k+1}\big)^\top
        &=\frac{1}{\delta}  \E\Big[f_{r-2}(\Mu^{1}_{B}\bar{B}+G^{1}_B,\dots,\Mu^{r-2}_{B}\bar{B}+G^{r-2}_B) \nonumber \\
        &\qquad\qquad f_{s-2}(\Mu^{1}_{B}\bar{B}+G^{1}_B,\dots,\Mu^{s-2}_{B}\bar{B}+G^{s-2}_B)^\top\Big]. \label{eq:Sigma_22_def}
\end{align}
Here $G^{1}_B,\dots,G^{k+1}_B\sim \normal(0,\Tau^{k+1}_{B})$ is independent of $\bar{B}\sim P_{\bar{B}}$. 

\begin{remark}
For $(Z,Z^0,\dots,Z^k)^\top \sim \normal(0,\Sigma^k)$, using standard properties of Gaussian random vectors, we have 
\begin{equation}
    (Z,Z^0,\dots,Z^k,\bar{\Psi})\stackrel{d}{=}(Z,\Mu^{0}_{\Theta}Z+G^0_\Theta,\dots,\Mu^{k}_{\Theta}Z+G^k_\Theta,\bar{\Psi}),
    \label{eq:ZZk_joint}
\end{equation}
where $G_\Theta^0,\dots,G_\Theta^k\sim \normal(0,\Tau_\Theta^k)$. Here, for $t\in\{0,\dots,k\}$, 
\begin{align}
    \Mu^{t}_{\Theta} 
    =\Sigma_{(t+2,1)}^k \big(\Sigma_{(1,1)}^k \big)^{-1}, 
\end{align}
and the covariance matrix $\Tau^{k}_{\Theta} \in \reals^{L(k+1) \times L(k+1)}$ has the following block-wise representation: 
\begin{align}
    &\Tau^{k}_{\Theta}  =
    \begin{bmatrix}
    \{ \Tau_{\Theta}^{k} \}_{(1,1)} & \dots &  \{ \Tau_{\Theta}^{k} \}_{(1,k+1)}\\
    \vdots & \dots & \vdots \\
    \{ \Tau_{\Theta}^{k} \}_{(k+1,1)} & \dots & \{ \Tau_{\Theta}^{k} \}_{(k+1,k+1)}
    \end{bmatrix},
    \label{eq:SE_Tauk_Theta}
\end{align}
where 
\begin{align}
    & \big\{\Tau^{k}_{\Theta}\big\}_{(t+1, t+1)} %_{[Lt+1:L(t+1)], \, [Lt+1:L(t+1)]} 
    = \Sigma_{(t+2,t+2)}^k \, - \, \Sigma_{(t+2,1)}^k \big(\Sigma_{(1,1)}^k \big)^{-1} \,  \Sigma_{(1,t+2)}^k. \label{eq:Tk_theta_def}
\end{align}
\label{rem:alt_dist_rep}
\end{remark}

\subsection{Main result} 
We begin with the assumptions required for the main result. The first is on the matrix-AMP initializer $\hB^0 \in \reals^{p \times L}$, the second is on the functions $g_k, f_{k+1}$ used to define the matrix-AMP algorithm in \eqref{eq:GAMP}, and the third is on the design matrix $\tX \in \reals^{n \times p}$. 

    \textbf{(A1)} Almost surely as $n,p\rightarrow\infty$, $(B,\hB^0)\stackrel{W}{\rightarrow}(\bar{B},\bar{B}^0)$, with the 
   joint law of $[\bar{B}^{\top},(\bar{B}^0)^{\top}]^{\top} \in \reals^{2L}$ having finite moments of all orders. 
     Multivariate polynomials are dense in the real $L^2$-spaces of functions $f:\mb{R}^{L_\Psi}\rightarrow\mb{R}$, $g:\mb{R}^{2L}\rightarrow\mb{R}$ with the inner-products
    \begin{align*}
        \big\langle f,\tilde{f}\big \rangle
        :=\E\big[f\big(\bar{\Psi}\big)\tilde{f}\big(\bar{\Psi}\big)\big]
        \text{ and }
        \big\langle g,\tilde{g}\big\rangle
        :=\E\big[g\big(\bar{B},\bar{B}^0\big)\tilde{g}\big(\bar{B},\bar{B}^0\big)\big].
    \end{align*}
    \label{assump:A1}
    
     \textbf{(A2)} For $k \ge 0$, the functions $f_{k+1}$ and $\tg_k$ (defined in \eqref{eq:g_tilde_k_def}) are continuous and Lipschitz, with respect to the all the arguments for $f_{k+1}$ and the first $(k+2)$ arguments for  $\tg_k$ (i.e., $Z,Z^0,\dots,Z^k$). Furthermore, these functions  satisfy the polynomial growth condition in \eqref{eq:poly_growth_cond} for some order $r\geq 1$.
     
      \textbf{(A3)} $\|\tX\|_{\text{op}}<C$ almost surely  for  sufficiently large $n,p$, for some constant $C$. For any fixed polynomial functions $f^\dag:\mb{R}^{L\rightarrow\mb{R}}$ and $f^\ddag:\mb{R}^{2L}\rightarrow\mb{R}$, as $n,p\rightarrow\infty$,
    \begin{align*}
        \max_{i\in[n]}\left|\langle f^\dag(\tPsi)\odot S_{i,:}\rangle-\langle f^\dag(\tPsi)\rangle\cdot\langle S_{i,:}\rangle\right|
        \stackrel{a.s.}{\rightarrow}0 \\
        \max_{j\in[p]}\left|\langle f^\ddag(B,\hB^0)\odot S_{:,j}\rangle-\langle f^\ddag(B,\hB^0)\rangle\cdot\langle S_{:,j}\rangle\right|
        \stackrel{a.s.}{\rightarrow}0,
    \end{align*}
    where $S$ is the variance profile of $\tX$ (see Definition \ref{def:gen_white_noise_matrix}).
%\end{itemize}

Assumption \textbf{(A1)}  implies that we have
    \begin{align}
             &  \frac{1}{n}
            \begin{bmatrix}
                B^\top B & B^\top\widehat{B}^0 \\
                (\widehat{B}^0)^\top B & (\widehat{B}^0)^\top\widehat{B}^0
            \end{bmatrix}
            \stackrel{a.s.}{\rightarrow}
            \Sigma^0:=            \frac{1}{\delta} \begin{bmatrix}
                \E[\bar{B} \bar{B}^\top ] & \E[\bar{B} (\bar{B}^0)^\top] \\
                \E[ (\bar{B}^0)^\top \bar{B}] & \E[ (\bar{B}^0)^\top \bar{B}^0]
            \end{bmatrix} \ \in \reals^{2L \times 2L}
            .  \label{eq:Sig0_def}
    \end{align}

\begin{theorem} \label{thm:GAMP}
Consider the matrix-AMP algorithm in \eqref{eq:GAMP} for the matrix GLM model in \eqref{eq:matrix-GLM}, defined via a generalized white noise design matrix (Definition \ref{def:gen_white_noise_matrix}). Suppose that the model assumptions on p.\pageref{page:model_assump} and assumptions \textbf{(A1)}, \textbf{(A2)}, and \textbf{(A3)} are satisfied, and that $\Tau^{1}_{B}$ is positive definite. Then for each $k \ge 0$, we have 
\begin{align}
     \big(B,\hB^0,B^1,\dots,B^{k+1}\big)  &\stackrel{W_2}{\rightarrow}
     \big(\bar{B},\bar{B}^0,\Mu^{1}_{B}\bar{B}+G^{1}_B,\dots,\Mu^{k+1}_{B}\bar{B}+G^{k+1}_B\big),
    \label{eq:matrix-GAMP_SE_result1} \\
    \big(\tPsi,\Theta,\Theta^0,\dots,\Theta^k\big)
    &\stackrel{W_2}{\rightarrow}
     \big(\bar{\Psi},Z,\Mu^{0}_{\Theta} \, Z+G^{0}_\Theta,\dots,\Mu^{k}_{\Theta} \, Z+G^{k}_\Theta\big), 
    \label{eq:matrix-GAMP_SE_result2}
\end{align}
almost surely as $n,p\rightarrow\infty$ with $n/p\rightarrow\delta$. In the above, $G^{1}_B,\dots,G^{k+1}_B\sim \normal(0,\Tau^{k+1}_{B})$ is independent of $\bar{B}$, and $G^{0}_\Theta,\dots,G^{k}_\Theta\sim \normal(0,\Tau^{k}_{\Theta})$ is independent of $(Z,\bar{\Psi})$.
\end{theorem}
The proof of the theorem is given in Appendix \ref{sec:GAMP_thm_proof}. We provide a brief outline here.
We start with an abstract AMP iteration (with vector-valued iterates) for generalized white noise matrices. A state evolution result for this  iteration was recently established by Wang et al. \cite{Wan22}. Using induction, we show that the matrix-AMP algorithm in \eqref{eq:GAMP} can be reduced to the abstract AMP iteration via a suitable change of variables and iteration indices. This allows us to translate the state evolution of the abstract AMP iteration to the matrix-AMP algorithm. One challenge in establishing the reduction is that the abstract AMP iteration is defined for vector-valued iterates, whereas the matrix-AMP algorithm in \eqref{eq:GAMP} has matrix-valued iterates, with $L$ columns per iterate. We handle this mismatch by using $L$ iterations of the abstract AMP to produce each matrix-valued iterate of the matrix-AMP.

\paragraph{Performance measures} Theorem \ref{thm:GAMP} allows us to compute the limiting values of performance measures, such as the mean-squared error (MSE) and the squared overlap, via the  laws of the random vectors on the RHS of \eqref{eq:matrix-GAMP_SE_result1} and \eqref{eq:matrix-GAMP_SE_result2}. Let $\bar{B}^{k} := \Mu^{k}_{B}\bar{B}+G^k_B$, and recall  that $Z^k := \Mu^{k}_{\Theta} Z + G^k_\Theta$.
Then, for $k \ge 1$, Theorem \ref{thm:GAMP} and \eqref{eq:sum_to_exp} together imply  that the limiting MSE of the matrix-AMP estimates $\hB^k \in \reals^{p \times L}$ is given by
\begin{align}
    \frac{1}{p}\|\hB^k-B\|_F^2
    =
    \frac{1}{p}\sum_{j=1}^p\|\hB_{j,:}^k-B_{j,:}\|_2^2 \, \stackrel{a.s.}{\rightarrow}  \, 
    \E\Big[\|f_k(\bar{B}^{1},\dots,\bar{B}^{k})-\bar{B}\|_2^2\Big].
\end{align}
where $\|\cdot\|_F$ is the Frobenius norm. Similarly, the limiting squared overlap of the matrix-AMP estimates is
\begin{align}
    \frac{\Tr\big((\hB^k)^\top B\big)^2}{\|\hB^k\|_F^2\cdot\|B\|_F^2}
    =\frac{\big(\frac{1}{p}\sum_{j=1}^p\langle \hB_{j,:}^k,B_{j,:}\rangle\big)^2}{\big(\frac{1}{p}\sum_{j=1}^p(\hB_{j,:}^k)^2\big)\cdot\big(\frac{1}{p}\sum_{j=1}^pB_{j,:}^2\big)}
    \,\stackrel{a.s.}{\rightarrow}\,
    \frac{\big(\E[\langle f_k(\bar{B}^1,\dots,\bar{B}^k),\bar{B}\rangle]\big)^2}{\E[f_k(\bar{B}^1,\dots,\bar{B}^k)^2]\cdot\E[\bar{B}^2]}.
    \label{eq:sq_overlap}
\end{align}

% Similarly, the limiting normalized correlation of the matrix-AMP estimates is 
% \begin{align}
%     \frac{1}{p}\sum_{j=1}^p\langle \hB^k_{j,:},B_{j,:}\rangle
% \,     \stackrel{a.s.}{\rightarrow} \, 
%     \E\Big[\langle f_k(\bar{B}^{1},\dots,\bar{B}^{k}),\Bar{B}\rangle\Big].
%     \label{eq:norm_corr}
% \end{align}

\paragraph{Choosing the AMP denoising functions}

Theorem \ref{thm:GAMP} and Remark \ref{rem:alt_dist_rep} show that the effective noise covariance matrices of the  random vectors  $\big(\Mu^{k}_{\Theta}\big)^{-1} Z^k \stackrel{d}{=} Z + \big(\Mu^{k}_{\Theta}\big)^{-1} G^{k}_{\Theta}$ and $\big( \Mu^{k+1}_{B}\big)^{-1} \bar{B}^{k+1} \stackrel{d}{=} \bar{B} +  \big(\Mu^{k+1}_{B}\big)^{-1} G^{k+1}_{B}$ are
\begin{align}
    \begin{split}
    N_{\Theta}^k&:=\left(\Mu^{k}_{\Theta}\right)^{-1} \left\{\Tau^{k}_{\Theta}\right\}_{(k+1, k+1)}
    \left(\big(\Mu^{k}_{\Theta}\big)^{-1}\right)^\top
    \text{ and } \\
    N_{B}^{k+1}&:=\left( \Mu^{k+1}_{B}\right)^{-1}
    \left\{\Tau^{k+1}_{B}\right\}_{(k+1, k+1)} \left( \big(\Mu^{k+1}_{B}\big)^{-1} \right)^\top,
    \end{split}
    \label{eq:eff_cov_def}
\end{align}
respectively. 
Here the $L \times L$ matrix $\left\{\Tau^{k}_{\Theta}\right\}_{(k+1, k+1)}$ is as defined in \eqref{eq:SE_Tauk_Theta}, and $\big\{\Tau^{k+1}_{B}\big\}_{(k+1, k+1)}$ is defined similarly from \eqref{eq:SE_Tk1B}.
A reasonable approach is to choose $f_k$ and $g_k$ such that the trace of each effective noise covariance matrix is minimized. Assuming that the signal prior $P_{\bar{B}}$ and the distribution of auxiliary variables $P_{\Psi}$ are known, the optimal choices for $f_k, g_k$ were derived in \cite[Section 3.1]{Tan23c}. Specifically, %it was shown that
 the following statements hold for $k \ge 1$:

1) Given $\Mu^{k}_{B}$, $\Tau^{k}_{B}$, the quantity $\Tr(N_{\Theta}^k)$ is minimized when $f_k = f_k^*$, where
\begin{equation}
    f_k^*(s^1,\dots,s^k)=\E[\bar{B}\mid \Mu^1_B\bar{B}+G_B^1 = s^1, \dots, \Mu^k_B\bar{B}+G_B^k = s^k], \label{eq:fk_opt_def}
\end{equation}
where $G_B^1,\dots,G_B^k \sim \normal(0, \Tau^{k}_{B})$ and $\bar{B} \sim P_{\bar{B}}$ are independent. 

2) Given $\Mu^{k}_{\Theta}$, $\Tau^{k}_{\Theta}$, the quantity $\Tr(N_B^{k+1})$ is minimized when $g_k =g_k^*$, where 
\begin{align}
    g_k^*(z^0,\dots,z^k,y)  
    =\E[Z \mid Z^0=z^0,\dots,Z^k=z^k, \bar{Y}=y]  - \, \E[Z \mid Z^0=z^0,\dots,Z^k=z^k]. \label{eq:gk_opt_def}
\end{align}
Here $(Z,Z^0,\dots,Z^k)^\top \sim \normal(0,\Sigma^k)$ and $\bar{Y}=q(Z, \bar{\Psi})$, with $\bar{\Psi} \sim P_{\bar{\Psi}}$ independent of $Z$. We remark that our $g_k^*$ in \eqref{eq:gk_opt_def} differs from that of \cite[Section 3.1]{Tan23c} by a left multiplication of a positive definite matrix (namely, $\Cov\big[Z \mid Z^0=z^0,\dots,Z^k=z^k\big]^{-1}$). It can be shown that omitting this matrix (or indeed, left multiplying any other positive definite one) still makes no difference to $\Tr(N_B^{k+1})$ being minimized.  For the matrix-AMP algorithm  with memoryless denoisers, given in \eqref{eq:memoryless_GAMP}, the optimal choices $f_k^*(s^k)$ and $g_k^*(z^k,y)$ have the same form as in \eqref{eq:fk_opt_def} and \eqref{eq:gk_opt_def}, with the past iterates omitted from the conditional expectations \cite{Tan23c}.

%%%%%
\section{AMP for Pooled Data} \label{sec:AMP_pooled_data}

We now consider the pooled data problem described in Section \ref{subsec:pooled_setup}. We use a  random dense design matrix $X \in \reals^{n \times p}$, where $X_{ij}\stackrel{\iid}{\sim}\text{Bernoulli}(\alpha)$ for some fixed $\alpha\in(0,1)$, for $i \in [n], j \in [p]$. This is a reasonable choice since in the noiseless pooled data setting, it has been proven \cite{Sca17} that the achievability for Bernoulli testing matches the converse for arbitrary designs, so Bernoulli testing is asymptotically optimal. Since the entries of $X$ have non-zero mean, they need to be recentered and scaled to obtain a generalized white noise matrix $\tX$. Specifically, let
\begin{align}
      \tX_{ij} =  \frac{(X_{ij} - \alpha)}{\sqrt{n\alpha(1-\alpha)}} \, , \quad i \in [n], \, j \in [p]. \label{eq:X_ij_decomp}
\end{align}
The entries of $\tX$ are i.i.d., with
\begin{align}
    \tX_{ij}
    &=
    \begin{cases}
        -\alpha\sqrt{\frac{1}{n\alpha(1-\alpha)}}\text{ with probability $1-\alpha$} \\
        (1-\alpha)\sqrt{\frac{1}{n\alpha(1-\alpha)}}\text{ with probability $\alpha$}, 
    \end{cases}
    \quad \quad i \in [n], \, j \in [p]. 
    \label{eq:txij_dist}
\end{align}
It is easy to check that $\tX$ satisfies the conditions in Definition \ref{def:gen_white_noise_matrix}, with  $S_{ij}=1$, for all $(i,j)$. Hence, $\tX$ is a generalized white noise matrix. 

From \eqref{eq:X_ij_decomp}, we have 
\begin{align}
X = \alpha\bs{1}_n\bs{1}_p^\top+\sqrt{n\alpha(1-\alpha)} \, \tX.
\label{eq:X_decomp}    
\end{align}
Therefore, 
the test outcome matrix   
$Y = X B + \Psi$ in \eqref{eq:noisy_pooled_data} can be rewritten as 
    $Y=\alpha\bs{1}_n\bs{1}_p^\top B+\sqrt{n\alpha(1-\alpha)}\tX B +\Psi$, which gives
\begin{align}
    \frac{Y-\alpha\bs{1}_n\bs{1}_p^\top B}{\sqrt{n\alpha(1-\alpha)}}=\tX B+\frac{\Psi}{\sqrt{n\alpha(1-\alpha)}}.
    \label{eq:mod_test_model}
\end{align}
Observe that 
\begin{align*}
\bs{1}_p^\top B=p \, \left(\frac{1}{p} \bs{1}_p^\top B \right)=p[\hat{\pi}_1,\dots,\hat{\pi}_L],
\end{align*}
where $\hat{\pi} \in\mb{R}^L$ is the vector of proportions defined in \eqref{eq:true_proportion_vec}. This allows us to rewrite  \eqref{eq:mod_test_model} row-wise as follows:
\begin{align*}
    \frac{Y_{i,:}-\alpha p\hat{\pi}}{\sqrt{n\alpha(1-\alpha)}}
    =B^\top\tX_{i,:} \,  + \, \frac{\Psi_{i, :}}{\sqrt{n\alpha(1-\alpha)}}, \qquad i \in [n].
\end{align*}
Defining 
\begin{align}
    \tY_{i,:} = \frac{Y_{i,:}-\alpha p\hat{\pi}}{\sqrt{n\alpha(1-\alpha)}} \, , \quad \tPsi_{i,:}=\frac{\Psi_{i, :}}{\sqrt{n\alpha(1-\alpha)}}, 
    \label{eq:tYtPsi_def}
\end{align}
we have 
\begin{align}
     \tY_{i,:} =B^\top\tX_{i,:} \,  + \, \tPsi_{i,:}, \ \ i \in [n], \ \  \text{ or }  \ 
 \tY = \tX B \,  + \, \tPsi,
 \label{eq:mod_model}
\end{align}
where $\tY, \tPsi \in \reals^{n\times L}$ are matrices whose $i$th rows are $ \tY_{i,:}$ and $\tPsi_{i,:}$. The modified model \eqref{eq:mod_model} is an instance of the matrix GLM in \eqref{eq:matrix-GLM} (with $\Psi = \tPsi =0_{n \times L}$ for the noiseless case). 

If the vector of proportions $\hat{\pi}$ is known, we can compute the modified test outcome matrix $\tY$ and run the matrix-AMP algorithm in \eqref{eq:GAMP} using the modified design $\tX$. Theorem \ref{thm:GAMP} then directly gives a rigorous asymptotic characterization for the matrix-AMP algorithm. We now describe how the assumptions for Theorem \ref{thm:GAMP} can be satisfied for the pooled data problem. 

\paragraph{Model Assumptions} We require that there exist random variables $\bar{B}\sim P_{\bar{B}}$ and $\bar{\Psi}\sim P_{\bar{\Psi}}$ (where $\bar{B}, \bar{\Psi} \in\mb{R}^{L}$) with $B\stackrel{W}{\rightarrow}\bar{B}$ and $\tPsi\stackrel{W}{\rightarrow}\bar{\Psi}$, respectively. For $B$, this means that the vector of proportions $\hat{\pi}$ in \eqref{eq:true_proportion_vec} converges in Wasserstein distance (of all orders) to a well-defined categorical distribution $\pi$ on the set $[L]$. For the limiting noise distribution $\bar{\Psi}$ to be well-defined, we need each entry of $\tPsi$ to be of constant order with high probability.
(Equivalently, from \eqref{eq:tYtPsi_def}, each entry of $\Psi$ need to be of order $\sqrt{p}$ with high probability.) This means that $\bar{\Psi}$ has to follow a distribution with constant order variance (e.g., all sub-Gaussian distributions). This is consistent with the fact that the entries of $\tX B$ are of constant order with high probability, as shown in Appendix \ref{sec:scaling_deriv}.

\paragraph{Verifying Assumption (A1)} We can initialize the matrix-AMP algorithm with $\widehat{B}^0\in \reals^{p \times L}$ whose rows are chosen i.i.d.~according to the limiting signal distribution $\pi$.  With this initialization, the initial covariance $\Sigma_0 $ and the convergence in \eqref{eq:Sig0_def} becomes:
\begin{align*}
 \frac{1}{n}
            \begin{bmatrix}
                B^\top B & B^\top\widehat{B}^0 \\
                (\widehat{B}^0)^\top B & (\widehat{B}^0)^\top\widehat{B}^0
            \end{bmatrix}
    \, \stackrel{a.s.}{\rightarrow} \,  \Sigma_0 =
    \frac{1}{\delta} \begin{bmatrix}
        \text{Diag}(\pi) & \pi \pi^{\top}  \\
        \pi \pi^{\top} & \text{Diag}(\pi)
    \end{bmatrix},
\end{align*}
where the convergence holds since $\hat{\pi} \to \pi$ as $p \to \infty$, by the assumption above.

\paragraph{Verifying Assumption (A2)} The optimal choices for the matrix-AMP denoising functions, given by $f_k^*$ and $g_k^*$ in \eqref{eq:fk_opt_def} and \eqref{eq:gk_opt_def}, can be explicitly computed for noiseless pooled data or the noisy setting with Gaussian noise; the expressions are given in \eqref{eq:Bayes_AMP_fk}-\eqref{eq:Bayes_AMP_gk} in Appendix \ref{sec:eqv_of_AMP_and_SE}. It can be verified from the expressions that these functions are continuous, Lipschitz, and satisfy the polynomial growth condition.

\paragraph{Verifying Assumption (A3)}  Since the entries of $\tX$ are i.i.d.~according to \eqref{eq:txij_dist}, the matrix $\sqrt{n} \tX$ has i.i.d.~sub-Gaussian entries. Using a concentration inequality for the operator norm of sub-Gaussian matrices \cite[Theorem 4.4.5]{Ver18} together with the Borel-Cantelli lemma, we obtain that $\|\tX\|_{\text{op}}<C$ almost surely for sufficiently large $n$. Since the variance profile $S_{ij}=1$ for all $(i,j)$, the second condition in (A3) is trivially satisfied.

\paragraph{Knowledge of $\hat{\pi}$} The matrix-AMP algorithm requires knowledge of  the vector of proportions $\hat{\pi}$ to compute $\tY$, as given in \eqref{eq:tYtPsi_def}. While the assumption of knowing $\hat{\pi}$ may seem strong, it can be easily obtained in the noiseless setting by using an additional (non-random) test, where we include every item. To obtain $\hat{\pi}$, take the output of this non-random test and divide it by $p$. For the noisy setting, one can estimate $\hat{\pi}$, either via a known prior $\pi$, or by averaging the outcomes of multiple non-random tests including every item. We explore the effect of mismatch in estimating $\hat{\pi}$ via simulations in the next subsection. 

\paragraph{Choice of $\alpha$}  We observe from \eqref{eq:txij_dist} that  the entries of $\tX_{i,:}$ have zero-mean and variance $1/n$, regardless of the value of the Bernoulli parameter $\alpha$. Therefore,  $\alpha$ influences the asymptotic performance of AMP (via state evolution) only through the variance of the rescaled noise $\tPsi$. From \eqref{eq:tYtPsi_def}, we see that choosing $\alpha=0.5$ minimizes the rescaled noise variance.

\paragraph{Comparison with El Alaoui et al.~\cite{Ala18}} El Alaoui et al.~\cite{Ala18} proposed an AMP algorithm for the noiseless setting, and characterized its performance via a state evolution, without providing rigorous guarantees. To connect our results with theirs, we prove the following proposition regarding the state evolution recursion.

\begin{proposition}\label{prop:eqv_of_AMPs}
Consider the AMP algorithm in \eqref{eq:memoryless_GAMP} with memoryless denoisers, i.e., $g_k$ depends only on $\Theta^k$ and $\tY$, and $f_{k+1}$ depends only on $B^{k+1}$. Under the Bayes-optimal choice of AMP denoising functions given by \eqref{eq:fk_opt_def}-\eqref{eq:gk_opt_def}, the state evolution recursion in \eqref{eq:SE_Mk1B}-\eqref{eq:SE_Sigk1} is equivalent to the state evolution recursion presented in \cite{Ala18}. 
\end{proposition} 
The proposition is proved in Appendix \ref{sec:eqv_of_AMP_and_SE} (Sections \ref{subsec:ElA_AMP}-\ref{subsec:SE_equiv}). Furthermore, we show in Section \ref{sec:eqv_of_AMPs} that the special case of the AMP with memoryless denoisers in \eqref{eq:memoryless_GAMP} can be obtained from the one in El Alaoui et al.~\cite{Ala18} using  large-sample approximations that are standard in the AMP literature. Theorem \ref{thm:GAMP} and Proposition \ref{prop:eqv_of_AMPs} together make the AMP performance guarantees in \cite{Ala18} rigorous.

\subsection{Numerical Simulations} \label{sec:sim_pool_data}

We compare the matrix-AMP algorithm with two popular methods for compressed sensing: iterative hard thresholding and convex programming.  Due to the additional constraints in the pooled data problem, some work is required to extend these algorithms to pooled data.  We first describe these algorithms, and then present numerical simulation results for the pooled data model given by \eqref{eq:noisy_pooled_data} (or equivalently \eqref{eq:tYtPsi_def}), in both the noiseless and noisy settings. Python code for all the simulations is available at \cite{CobCode24}.

\paragraph{Optimization-based methods} We obtain these from the convex relaxation of the maximum a posteriori (MAP) estimator for pooled data.
We assume that the categorical prior $\pi$ is known, i.e., the probability of  each item belonging to category $l \in [L]$ is $\pi_l$.
We first reformulate our variables to ones that are more suitable for optimization formulations. We define
\begin{align}
    Y_{\opt}&=
    \begin{bmatrix}
        Y_{:,1} \\
        \vdots \\
        Y_{:,L}
    \end{bmatrix}\in\mb{R}^{nL};
    \quad
    B_{\opt}=
    \begin{bmatrix}
        B_{:,1} \\
        \vdots \\
        B_{:,L}
    \end{bmatrix}\in\{0,1\}^{pL}; \nonumber \\
    X_{\opt}&=\text{diag}(X,\dots,X)\in\{0,1\}^{nL\times pL};
    \quad
    C_{\opt}=[I_p,\dots,I_p]\in\{0,1\}^{p\times pL}. \label{eq:opt_defs}
\end{align}
Write $I_{pL}^{(pl)}\in\mb{R}^{p\times pL}$ for the sub-matrix of $I_{pL}$ obtained by taking the $\big(p(l-1)+1\big)$-th to $pl$-th rows of $I_{pL}$. The convex program for the noisy setting (CVX), with noise $\Psi_{i,:} \stackrel{\iid}{\sim} \normal(0,p\sigma^2I_L)$, is
\begin{align*}
    \text{minimize}&\quad \frac{1}{2p\sigma^2}\|Y_{\opt}-X_{\opt}B_{\opt}\|_2^2-\sum_{l=1}^L(\log\pi_l)1_p^\top I_{pL}^{(pl)}B_\opt
    \quad(\text{w.r.t.~ $B_\opt$}) \\
    \text{subject to}&\quad 0_{pL}\leq B_\opt \leq 1_{pL},
    \text{ and }
    C_{\opt}B_{\opt}=1_p.
\end{align*}
For the noiseless pooled data setting, the convex program above can be simplified to the following linear program (LP):
\begin{align*}
    \text{minimize}&\quad -\sum_{l=1}^L(\log\pi_l)1_p^\top I_{pL}^{(pl)}B_\opt\quad(\text{w.r.t.~ $B_\opt$}) \\
    \text{subject to}&\quad 0_{pL}\leq B_\opt \leq 1_{pL},\,
    Y_{\opt}=X_{\opt}B_{\opt},\,
    C_{\opt}B_{\opt}=1_p.
\end{align*}
The derivations of the convex and linear programs are given in Appendix \ref{sec:optm_derivations}.

\paragraph{Iterative hard thresholding} We extend the IHT algorithm \cite[Section 3.3]{Fou13} for compressed sensing to pooled data. The algorithm is as follows:
\begin{itemize}
    \item \textbf{Input.} $\tY$, $\tX$, and the sparsity levels of the columns of $B$, namely $\{\hat{\pi}_1p,\dots,\hat{\pi}_Lp\}$.
    \item \textbf{Initialization.} $B^0=0\in\mb{R}^{p\times L}$.
    \item For iteration $k=1,\dots,k_{\max}$: Execute
    \begin{align*}
        B^{k+1}=H_{\{\hat{\pi}_1p,\dots,\hat{\pi}_Lp\}}(B^k+\tX^\top(\tY-\tX B^k)),
    \end{align*}
    where $H_{\{\hat{\pi}_1p,\dots,\hat{\pi}_Lp\}}(\cdot)$ is the hard thresholding function adapted to matrix signals, and assigns each item to a category as follows:
    
        1. Look for the largest entry in the input matrix and make a hard decision on the corresponding (item, category) pair.
        
        2.  Once an item is allocated, it is removed from further consideration when looking for the next item. Once a class $l$ has been allocated $\hat{\pi}_lp$ items, we also remove it from further consideration when looking for the next (item, category) pair.
        
        3. Repeat the two steps above until all $p$ items have been categorized
    \item Return $B^{k_{\max}}$.
\end{itemize}

\paragraph{Matrix-AMP implementation} We use the matrix-AMP algorithm  in  \eqref{eq:memoryless_GAMP}, with memoryless denoisers, and implement with $g_k=g_k^*$ and $f_k=f_k^*$, the optimal choices given by \eqref{eq:gk_opt_def} and \eqref{eq:fk_opt_def}.   For the pooled data problem with additive Gaussian noise, the optimal memoryless $g_k^*(Z^k,\bar{Y})=Z^k-\bar{Y}$; see Appendix \ref{sec:imp_details_pooled_data}. With this $g_k^*$, it can be shown that
\begin{align}
f_k^*(s^1,\dots,s^k)
&=\E[\bar{B}\mid \Mu^1_B\bar{B}+G_B^1 = s^1, \dots, \Mu^k_B\bar{B}+G_B^k = s^k] \nonumber \\
&=\E[\bar{B}\mid \Mu^k_B\bar{B}+G_B^k = s^k]
=f_k^*(s^k). \label{eq:memory=memoryless}
\end{align}
This provides a partial justification for using  memoryless Bayes-optimal denoisers. However, we note that  denoisers with memory could help improve the performance of matrix-AMP if the signal prior is not known. 
The performance of matrix-AMP can be tracked via a simplified version of the state evolution equations, described in \eqref{eq:memoryless_Mu}-\eqref{eq:memoryless_Thetak_conv} in Appendix \ref{sec:eqv_of_AMP_and_SE}. Full details on the derivation of the denoisers and their derivatives are given in Appendix \ref{sec:imp_details_pooled_data}.

\paragraph{Simulation set-up}
We take $X_{ij}\stackrel{\iid}{\sim}\text{Bernoulli}(\alpha)$ for $i,j\in[n]\times[p]$, and the rows of the signal $B_{j,:} \stackrel{\iid}{\sim}\text{Categorical}(\pi)$ for $j\in[p]$, where $\pi$ is a vector of probabilities that sum to one. We use $\alpha=0.5$, since we find that the choice of $\alpha$ has very little effect on the performance of each algorithm. The rows of the AMP initializer $\hB^0\in\mb{R}^{p\times2}$ are chosen i.i.d.~according to the same distribution $\pi$, independently of the signal. We set the signal dimension $p=500$ and vary the value of $n$ in our experiments. 

In the noisy setting, we consider Gaussian noise, with $\Psi_{i,:}\stackrel{\iid}{\sim}\normal(0,p\sigma^2I_L)$ for $i\in[n]$. While it may seem unusual to add continuous noise to discrete observations, this still captures the essence of the noisy pooled data problem, and simplifies the implementation of our algorithm. Furthermore, since the noise standard deviation is of order $\sqrt{p}$, rounding the noise values to the nearest integer does not noticeably affect performance. This noisy model has been previously considered in \cite[Section 2.2., Example 3]{Sca17}.

The performance in all the plots is measured via the squared overlap between the matrix-AMP estimate and the signal (see \eqref{eq:sq_overlap}).  Each point on the plots is obtained from 100 independent runs, where in each run, the matrix-AMP algorithm is executed for 200 iterations. We report the average and error bars at 1 standard deviation of the final iteration. In our implementation, in the final iteration of the matrix-AMP algorithm, we quantize the estimates of the signal $B$, e.g., a row estimate of $[0.9,0.1]$ would get quantized to $[1,0]$.   In all our plots, curves labeled `AMP' refers to the empirical performance of the matrix-AMP algorithm, and points labeled `SE' refers to the theoretical performance of the  algorithm, calculated from the state evolution parameters in \eqref{eq:memoryless_Mu}-\eqref{eq:memoryless_Bk_conv}.

\begin{figure}[t]
\centering
\begin{subfigure}[b]{0.45\textwidth}
  \centering
  % This file was created with tikzplotlib v0.10.1.
\begin{tikzpicture}[scale=0.8]

\definecolor{coral}{RGB}{255,127,80}
\definecolor{darkgray176}{RGB}{176,176,176}
\definecolor{green}{RGB}{0,128,0}
\definecolor{lightblue}{RGB}{173,216,230}
\definecolor{lightgray204}{RGB}{204,204,204}
\definecolor{lightgreen}{RGB}{144,238,144}

\begin{axis}[
legend cell align={left},
legend style={
  fill opacity=0.8,
  draw opacity=1,
  text opacity=1,
  at={(0.97,0.03)},
  anchor=south east,
  draw=lightgray204
},
tick align=outside,
tick pos=left,
x grid style={darkgray176},
xlabel={\(\displaystyle \delta=n/p\)},
xmajorgrids,
xmin=0.055, xmax=1.045,
xtick style={color=black},
y grid style={darkgray176},
ylabel={Overlap},
ymajorgrids,
ymin=0.522356825857179, ymax=1.06727808898093,
ytick style={color=black}
]
\addplot [semithick, blue, dashed]
table {%
0.1 0.550201375280418
0.11 0.555266762590216
0.12 0.560374659509058
0.13 0.565509244476453
0.14 0.570592251310777
0.15 0.575815316380244
0.16 0.581116891921158
0.17 0.586182017478122
0.18 0.591600238535889
0.19 0.596590103137101
0.2 0.601984849609947
0.21 0.607179706177404
0.22 0.61269513565549
0.23 0.618301731703739
0.24 0.623576050137707
0.25 0.629177966472223
0.26 0.634997825462344
0.27 0.64047212454451
0.28 0.64616124570698
0.29 0.652029251046211
0.3 0.658058348661263
0.31 0.663938550025396
0.32 0.670028370276254
0.33 0.67606819238132
0.34 0.681983308208247
0.35 0.688572134406358
0.36 0.695183809985034
0.37 0.702354185791474
0.38 0.709106208327109
0.39 0.716407399025501
0.4 0.723140142675157
0.41 0.7306249356203
0.42 0.739140328110779
0.43 0.747166415894574
0.44 0.756499064672146
0.45 0.765975983293653
0.46 0.776388562855627
0.47 0.788441960025806
0.48 0.799635463589955
0.49 0.820200915196293
0.5 1
0.51 1
0.52 1
0.53 1
0.54 1
0.55 1
0.56 1
0.57 1
0.58 1
0.59 1
0.6 1
0.61 1
0.62 1
0.63 1
0.64 1
0.65 1
0.66 1
0.67 1
0.68 1
0.69 1
0.7 1
0.71 1
0.72 1
0.73 1
0.74 1
0.75 1
0.76 1
0.77 1
0.78 1
0.79 1
0.8 1
0.81 1
0.82 1
0.83 1
0.84 1
0.85 1
0.86 1
0.87 1
0.88 1
0.89 1
0.9 1
0.91 1
0.92 1
0.93 1
0.94 1
0.95 1
0.96 1
0.97 1
0.98 1
0.99 1
1 1
};
\addlegendentry{SE, [0.5,0.5]}
\addplot [semithick, red, dashed]
table {%
0.1 0.623402838161695
0.11 0.628084393456963
0.12 0.632688407255048
0.13 0.637216955195428
0.14 0.641782767576675
0.15 0.646620077173608
0.16 0.651237617370556
0.17 0.655937632417231
0.18 0.661194626191247
0.19 0.666065183502921
0.2 0.671014043667133
0.21 0.675891891016578
0.22 0.680627081768141
0.23 0.685838936410226
0.24 0.691024524291523
0.25 0.696608138053634
0.26 0.701790758339961
0.27 0.707165848656769
0.28 0.712811271114158
0.29 0.718549384474472
0.3 0.724401262771794
0.31 0.73055353920852
0.32 0.736763612072675
0.33 0.743168105103477
0.34 0.749163626690785
0.35 0.756322646883502
0.36 0.763193202881287
0.37 0.770553940095273
0.38 0.778287142332186
0.39 0.786780710505564
0.4 0.795226264748853
0.41 0.806192750440335
0.42 0.816971357722601
0.43 0.830463230780756
0.44 0.854718501765403
0.45 1
0.46 1
0.47 1
0.48 1
0.49 1
0.5 1
0.51 1
0.52 1
0.53 1
0.54 1
0.55 1
0.56 1
0.57 1
0.58 1
0.59 1
0.6 1
0.61 1
0.62 1
0.63 1
0.64 1
0.65 1
0.66 1
0.67 1
0.68 1
0.69 1
0.7 1
0.71 1
0.72 1
0.73 1
0.74 1
0.75 1
0.76 1
0.77 1
0.78 1
0.79 1
0.8 1
0.81 1
0.82 1
0.83 1
0.84 1
0.85 1
0.86 1
0.87 1
0.88 1
0.89 1
0.9 1
0.91 1
0.92 1
0.93 1
0.94 1
0.95 1
0.96 1
0.97 1
0.98 1
0.99 1
1 1
};
\addlegendentry{SE, [0.3, 0.7]}
\addplot [semithick, green, dashed]
table {%
0.1 0.844233828272526
0.11 0.847337176238311
0.12 0.850525491760371
0.13 0.854096356460707
0.14 0.857535505479883
0.15 0.861488119132835
0.16 0.865600036396713
0.17 0.869768994464795
0.18 0.874334431677965
0.19 0.878754421181344
0.2 0.883893257676079
0.21 0.889848630290007
0.22 0.896208821619196
0.23 0.904219283345129
0.24 0.91410697820521
0.25 1
0.26 1
0.27 1
0.28 1
0.29 1
0.3 1
0.31 1
0.32 1
0.33 1
0.34 1
0.35 1
0.36 1
0.37 1
0.38 1
0.39 1
0.4 1
0.41 1
0.42 1
0.43 1
0.44 1
0.45 1
0.46 1
0.47 1
0.48 1
0.49 1
0.5 1
0.51 1
0.52 1
0.53 1
0.54 1
0.55 1
0.56 1
0.57 1
0.58 1
0.59 1
0.6 1
0.61 1
0.62 1
0.63 1
0.64 1
0.65 1
0.66 1
0.67 1
0.68 1
0.69 1
0.7 1
0.71 1
0.72 1
0.73 1
0.74 1
0.75 1
0.76 1
0.77 1
0.78 1
0.79 1
0.8 1
0.81 1
0.82 1
0.83 1
0.84 1
0.85 1
0.86 1
0.87 1
0.88 1
0.89 1
0.9 1
0.91 1
0.92 1
0.93 1
0.94 1
0.95 1
0.96 1
0.97 1
0.98 1
0.99 1
1 1
};
\addlegendentry{SE, [0.1, 0.9]}
\path [draw=lightblue, very thick]
(axis cs:0.1,0.538710182736136)
--(axis cs:0.1,0.558476866805619);

\path [draw=lightblue, very thick]
(axis cs:0.2,0.584951759867322)
--(axis cs:0.2,0.613271266931962);

\path [draw=lightblue, very thick]
(axis cs:0.3,0.63821634522487)
--(axis cs:0.3,0.674508055304011);

\path [draw=lightblue, very thick]
(axis cs:0.4,0.696539716627099)
--(axis cs:0.4,0.751567885480837);

\path [draw=lightblue, very thick]
(axis cs:0.5,0.834948278010405)
--(axis cs:0.5,1.04326233738088);

\path [draw=lightblue, very thick]
(axis cs:0.6,1)
--(axis cs:0.6,1);

\path [draw=lightblue, very thick]
(axis cs:0.7,1)
--(axis cs:0.7,1);

\path [draw=lightblue, very thick]
(axis cs:0.8,1)
--(axis cs:0.8,1);

\path [draw=lightblue, very thick]
(axis cs:0.9,1)
--(axis cs:0.9,1);

\path [draw=lightblue, very thick]
(axis cs:1,1)
--(axis cs:1,1);

\path [draw=coral, very thick]
(axis cs:0.1,0.601911018265218)
--(axis cs:0.1,0.634380503631947);

\path [draw=coral, very thick]
(axis cs:0.2,0.649154467577925)
--(axis cs:0.2,0.69252210806554);

\path [draw=coral, very thick]
(axis cs:0.3,0.697570423301754)
--(axis cs:0.3,0.745117142089108);

\path [draw=coral, very thick]
(axis cs:0.4,0.752511396356644)
--(axis cs:0.4,0.865947060804138);

\path [draw=coral, very thick]
(axis cs:0.5,0.950365396735377)
--(axis cs:0.5,1.03480912515988);

\path [draw=coral, very thick]
(axis cs:0.6,1)
--(axis cs:0.6,1);

\path [draw=coral, very thick]
(axis cs:0.7,1)
--(axis cs:0.7,1);

\path [draw=coral, very thick]
(axis cs:0.8,1)
--(axis cs:0.8,1);

\path [draw=coral, very thick]
(axis cs:0.9,1)
--(axis cs:0.9,1);

\path [draw=coral, very thick]
(axis cs:1,1)
--(axis cs:1,1);

\path [draw=lightgreen, very thick]
(axis cs:0.1,0.817671186694918)
--(axis cs:0.1,0.860877189387429);

\path [draw=lightgreen, very thick]
(axis cs:0.2,0.856895283123274)
--(axis cs:0.2,0.921386440849669);

\path [draw=lightgreen, very thick]
(axis cs:0.3,1)
--(axis cs:0.3,1);

\path [draw=lightgreen, very thick]
(axis cs:0.4,1)
--(axis cs:0.4,1);

\path [draw=lightgreen, very thick]
(axis cs:0.5,1)
--(axis cs:0.5,1);

\path [draw=lightgreen, very thick]
(axis cs:0.6,1)
--(axis cs:0.6,1);

\path [draw=lightgreen, very thick]
(axis cs:0.7,1)
--(axis cs:0.7,1);

\path [draw=lightgreen, very thick]
(axis cs:0.8,1)
--(axis cs:0.8,1);

\path [draw=lightgreen, very thick]
(axis cs:0.9,1)
--(axis cs:0.9,1);

\path [draw=lightgreen, very thick]
(axis cs:1,1)
--(axis cs:1,1);

\addplot [semithick, blue, mark=asterisk, mark size=3, mark options={solid}, only marks]
table {%
0.1 0.548593524770878
0.2 0.599111513399642
0.3 0.656362200264441
0.4 0.724053801053968
0.5 0.939105307695644
0.6 1
0.7 1
0.8 1
0.9 1
1 1
};
\addlegendentry{AMP, [0.5, 0.5]}
\addplot [semithick, red, mark=asterisk, mark size=3, mark options={solid}, only marks]
table {%
0.1 0.618145760948582
0.2 0.670838287821733
0.3 0.721343782695431
0.4 0.809229228580391
0.5 0.992587260947627
0.6 1
0.7 1
0.8 1
0.9 1
1 1
};
\addlegendentry{AMP, [0.3, 0.7]}
\addplot [semithick, green, mark=asterisk, mark size=3, mark options={solid}, only marks]
table {%
0.1 0.839274188041173
0.2 0.889140861986472
0.3 1
0.4 1
0.5 1
0.6 1
0.7 1
0.8 1
0.9 1
1 1
};
\addlegendentry{AMP, [0.1, 0.9]}
\end{axis}

\end{tikzpicture}
  \vspace{-2\baselineskip}
  \caption{$L=2$}
\end{subfigure}
\begin{subfigure}[b]{0.45\textwidth}
  \centering
  % This file was created with tikzplotlib v0.10.1.
\begin{tikzpicture}[scale=0.8]

\definecolor{coral}{RGB}{255,127,80}
\definecolor{darkgray176}{RGB}{176,176,176}
\definecolor{lightblue}{RGB}{173,216,230}
\definecolor{lightgray204}{RGB}{204,204,204}

\begin{axis}[
legend cell align={left},
legend style={
  fill opacity=0.8,
  draw opacity=1,
  text opacity=1,
  at={(0.97,0.03)},
  anchor=south east,
  draw=lightgray204
},
tick align=outside,
tick pos=left,
x grid style={darkgray176},
xlabel={\(\displaystyle \delta=n/p\)},
xmajorgrids,
xmin=0.055, xmax=1.045,
xtick style={color=black},
y grid style={darkgray176},
ylabel={Overlap},
ymajorgrids,
ymin=0.367602704651861, ymax=1.06534811168587,
ytick style={color=black}
]
\addplot [semithick, blue, dashed]
table {%
0.1 0.403319716936278
0.11 0.410517406209138
0.12 0.418140704402027
0.13 0.425524278394248
0.14 0.433144770816875
0.15 0.440961819477646
0.16 0.448791884869155
0.17 0.45682262875895
0.18 0.464896290081386
0.19 0.473204563977109
0.2 0.481170095160277
0.21 0.489602481816547
0.22 0.498300145744635
0.23 0.506813954403572
0.24 0.516190743134289
0.25 0.525045372522357
0.26 0.534518496022504
0.27 0.544300266960825
0.28 0.554040049595179
0.29 0.56380810276276
0.3 0.574767335168934
0.31 0.584803534891108
0.32 0.596171390278946
0.33 0.608001752026031
0.34 0.620598847644601
0.35 0.633682235095207
0.36 0.647389732913688
0.37 0.663421452678805
0.38 0.680763873282104
0.39 0.6997136592506
0.4 0.726567013658599
0.41 1
0.42 1
0.43 1
0.44 1
0.45 1
0.46 1
0.47 1
0.48 1
0.49 1
0.5 1
0.51 1
0.52 1
0.53 1
0.54 1
0.55 1
0.56 1
0.57 1
0.58 1
0.59 1
0.6 1
0.61 1
0.62 1
0.63 1
0.64 1
0.65 1
0.66 1
0.67 1
0.68 1
0.69 1
0.7 1
0.71 1
0.72 1
0.73 1
0.74 1
0.75 1
0.76 1
0.77 1
0.78 1
0.79 1
0.8 1
0.81 1
0.82 1
0.83 1
0.84 1
0.85 1
0.86 1
0.87 1
0.88 1
0.89 1
0.9 1
0.91 1
0.92 1
0.93 1
0.94 1
0.95 1
0.96 1
0.97 1
0.98 1
0.99 1
1 1
};
\addlegendentry{SE, [1/3, 1/3, 1/3]}
\addplot [semithick, red, dashed]
table {%
0.1 0.520477319856719
0.11 0.527368854636416
0.12 0.534444349081812
0.13 0.541385531317628
0.14 0.548609719132491
0.15 0.556343056464612
0.16 0.563917337284345
0.17 0.572106157287056
0.18 0.580344313741773
0.19 0.588996470538151
0.2 0.598072400395377
0.21 0.607643650048102
0.22 0.61793887620742
0.23 0.629382899522962
0.24 0.642172245178448
0.25 0.661153627819928
0.26 0.727514399362799
0.27 0.733995999856694
0.28 0.739453298046533
0.29 0.745716228120032
0.3 0.752244872046176
0.31 0.758362560201668
0.32 0.764616768078667
0.33 0.771905306365239
0.34 0.77947037662042
0.35 0.787048007943083
0.36 0.79482383390543
0.37 0.804241330475803
0.38 0.813776058453096
0.39 0.824498950946014
0.4 0.836955411882745
0.41 0.858092188345577
0.42 1
0.43 1
0.44 1
0.45 1
0.46 1
0.47 1
0.48 1
0.49 1
0.5 1
0.51 1
0.52 1
0.53 1
0.54 1
0.55 1
0.56 1
0.57 1
0.58 1
0.59 1
0.6 1
0.61 1
0.62 1
0.63 1
0.64 1
0.65 1
0.66 1
0.67 1
0.68 1
0.69 1
0.7 1
0.71 1
0.72 1
0.73 1
0.74 1
0.75 1
0.76 1
0.77 1
0.78 1
0.79 1
0.8 1
0.81 1
0.82 1
0.83 1
0.84 1
0.85 1
0.86 1
0.87 1
0.88 1
0.89 1
0.9 1
0.91 1
0.92 1
0.93 1
0.94 1
0.95 1
0.96 1
0.97 1
0.98 1
0.99 1
1 1
};
\addlegendentry{SE, [0.1, 0.3, 0.6]}
\path [draw=lightblue, very thick]
(axis cs:0.1,0.387247112761564)
--(axis cs:0.1,0.408251641017717);

\path [draw=lightblue, very thick]
(axis cs:0.2,0.459138443311345)
--(axis cs:0.2,0.49582632179935);

\path [draw=lightblue, very thick]
(axis cs:0.3,0.54344759531526)
--(axis cs:0.3,0.595038002234662);

\path [draw=lightblue, very thick]
(axis cs:0.4,0.687377431227771)
--(axis cs:0.4,1.00758632079689);

\path [draw=lightblue, very thick]
(axis cs:0.5,1)
--(axis cs:0.5,1);

\path [draw=lightblue, very thick]
(axis cs:0.6,1)
--(axis cs:0.6,1);

\path [draw=lightblue, very thick]
(axis cs:0.7,1)
--(axis cs:0.7,1);

\path [draw=lightblue, very thick]
(axis cs:0.8,1)
--(axis cs:0.8,1);

\path [draw=lightblue, very thick]
(axis cs:0.9,1)
--(axis cs:0.9,1);

\path [draw=lightblue, very thick]
(axis cs:1,1)
--(axis cs:1,1);

\path [draw=coral, very thick]
(axis cs:0.1,0.497957855197413)
--(axis cs:0.1,0.535945889067374);

\path [draw=coral, very thick]
(axis cs:0.2,0.557749857405617)
--(axis cs:0.2,0.622300492724523);

\path [draw=coral, very thick]
(axis cs:0.3,0.715050815508672)
--(axis cs:0.3,0.775202151208381);

\path [draw=coral, very thick]
(axis cs:0.4,0.79468710619121)
--(axis cs:0.4,0.983524014970584);

\path [draw=coral, very thick]
(axis cs:0.5,1)
--(axis cs:0.5,1);

\path [draw=coral, very thick]
(axis cs:0.6,1)
--(axis cs:0.6,1);

\path [draw=coral, very thick]
(axis cs:0.7,1)
--(axis cs:0.7,1);

\path [draw=coral, very thick]
(axis cs:0.8,1)
--(axis cs:0.8,1);

\path [draw=coral, very thick]
(axis cs:0.9,1)
--(axis cs:0.9,1);

\path [draw=coral, very thick]
(axis cs:1,1)
--(axis cs:1,1);

\addplot [semithick, blue, mark=asterisk, mark size=3, mark options={solid}, only marks]
table {%
0.1 0.397749376889641
0.2 0.477482382555348
0.3 0.569242798774961
0.4 0.847481876012328
0.5 1
0.6 1
0.7 1
0.8 1
0.9 1
1 1
};
\addlegendentry{AMP, [1/3, 1/3, 1/3]}
\addplot [semithick, red, mark=asterisk, mark size=3, mark options={solid}, only marks]
table {%
0.1 0.516951872132393
0.2 0.59002517506507
0.3 0.745126483358526
0.4 0.889105560580897
0.5 1
0.6 1
0.7 1
0.8 1
0.9 1
1 1
};
\addlegendentry{AMP, [0.1, 0.3, 0.6]}
\end{axis}

\end{tikzpicture}
  \vspace{-2\baselineskip}
  \caption{$L=3$}
\end{subfigure}
\caption{AMP for pooled data: squared overlap vs.~$\delta$, with $\sigma=0$  and varying $\pi$.}
\label{fig:noiseless}
\end{figure}
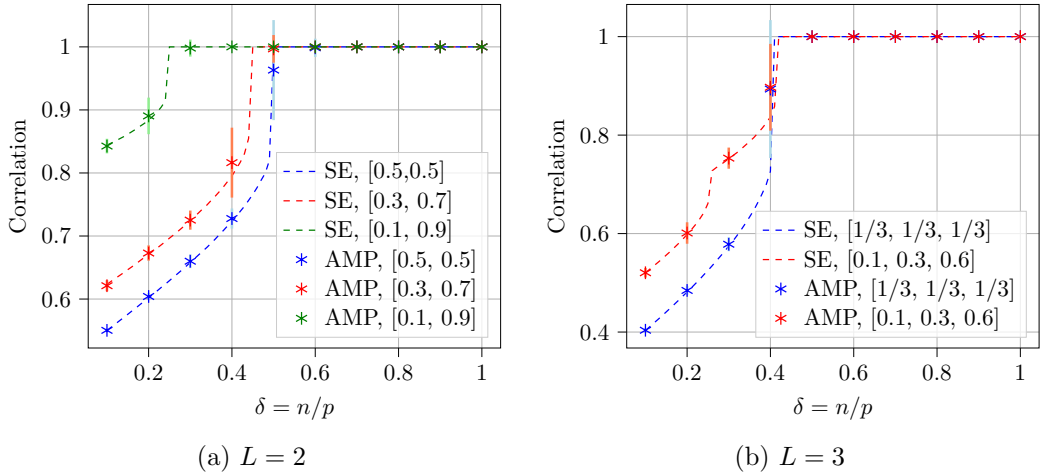

\begin{figure}[t]
\centering
\begin{subfigure}[b]{0.34\textwidth}
  \centering
  % This file was created with tikzplotlib v0.10.1.
\begin{tikzpicture}[scale=0.62]

\definecolor{coral}{RGB}{255,127,80}
\definecolor{darkgray176}{RGB}{176,176,176}
\definecolor{green}{RGB}{0,128,0}
\definecolor{lightblue}{RGB}{173,216,230}
\definecolor{lightgray204}{RGB}{204,204,204}
\definecolor{lightgreen}{RGB}{144,238,144}

\begin{axis}[
legend cell align={left},
legend style={
  fill opacity=0.8,
  draw opacity=1,
  text opacity=1,
  at={(0.97,0.03)},
  anchor=south east,
  draw=lightgray204
},
tick align=outside,
tick pos=left,
x grid style={darkgray176},
xlabel={\(\displaystyle \delta=n/p\)},
xmajorgrids,
xmin=0.055, xmax=1.045,
xtick style={color=black},
y grid style={darkgray176},
ylabel={Overlap},
ymajorgrids,
ymin=0.0476403344186668, ymax=1.11159215177566,
ytick style={color=black}
]
\addplot [semithick, blue, dashed]
table {%
0.1 0.403230617663388
0.11 0.41054217791051
0.12 0.418072646119556
0.13 0.42586601711218
0.14 0.433049684112449
0.15 0.441068112897524
0.16 0.448677486355195
0.17 0.456796115890354
0.18 0.464965565224308
0.19 0.473060411748787
0.2 0.481383952640562
0.21 0.489768081710178
0.22 0.49849278309212
0.23 0.507063825756621
0.24 0.515892837028079
0.25 0.52529028713855
0.26 0.53471322546817
0.27 0.543685117125823
0.28 0.553979541774879
0.29 0.563521208474632
0.3 0.574534881337513
0.31 0.584950335194167
0.32 0.596399280594805
0.33 0.60824276728242
0.34 0.620328597456081
0.35 0.633617102706393
0.36 0.647536846843159
0.37 0.663010132666371
0.38 0.680425387426825
0.39 0.700418484800851
0.4 0.72536240571466
0.41 1
0.42 1
0.43 1
0.44 1
0.45 1
0.46 1
0.47 1
0.48 1
0.49 1
0.5 1
0.51 1
0.52 1
0.53 1
0.54 1
0.55 1
0.56 1
0.57 1
0.58 1
0.59 1
0.6 1
0.61 1
0.62 1
0.63 1
0.64 1
0.65 1
0.66 1
0.67 1
0.68 1
0.69 1
0.7 1
0.71 1
0.72 1
0.73 1
0.74 1
0.75 1
0.76 1
0.77 1
0.78 1
0.79 1
0.8 1
0.81 1
0.82 1
0.83 1
0.84 1
0.85 1
0.86 1
0.87 1
0.88 1
0.89 1
0.9 1
0.91 1
0.92 1
0.93 1
0.94 1
0.95 1
0.96 1
0.97 1
0.98 1
0.99 1
1 1
};
\addlegendentry{SE}
\path [draw=lightblue, very thick]
(axis cs:0.1,0.384017728651882)
--(axis cs:0.1,0.41217274696714);

\path [draw=lightblue, very thick]
(axis cs:0.2,0.459912686836364)
--(axis cs:0.2,0.497467841272017);

\path [draw=lightblue, very thick]
(axis cs:0.3,0.535270398570356)
--(axis cs:0.3,0.599753249883027);

\path [draw=lightblue, very thick]
(axis cs:0.4,0.672658451359126)
--(axis cs:0.4,1.00816826899011);

\path [draw=lightblue, very thick]
(axis cs:0.425,0.884123379609285)
--(axis cs:0.425,1.06323070553216);

\path [draw=lightblue, very thick]
(axis cs:0.45,0.957112355453903)
--(axis cs:0.45,1.03505419493878);

\path [draw=lightblue, very thick]
(axis cs:0.5,1)
--(axis cs:0.5,1);

\path [draw=lightblue, very thick]
(axis cs:0.6,1)
--(axis cs:0.6,1);

\path [draw=lightblue, very thick]
(axis cs:0.7,1)
--(axis cs:0.7,1);

\path [draw=lightblue, very thick]
(axis cs:0.8,1)
--(axis cs:0.8,1);

\path [draw=lightblue, very thick]
(axis cs:0.9,1)
--(axis cs:0.9,1);

\path [draw=lightblue, very thick]
(axis cs:1,1)
--(axis cs:1,1);

\path [draw=coral, very thick]
(axis cs:0.1,0.165768443411676)
--(axis cs:0.1,0.194492292303596);

\path [draw=coral, very thick]
(axis cs:0.2,0.251535184798191)
--(axis cs:0.2,0.29456356723569);

\path [draw=coral, very thick]
(axis cs:0.3,0.360971413163041)
--(axis cs:0.3,0.424200204223362);

\path [draw=coral, very thick]
(axis cs:0.4,0.585292683911154)
--(axis cs:0.4,0.824213725461843);

\path [draw=coral, very thick]
(axis cs:0.425,0.810377783160044)
--(axis cs:0.425,1.0134420734853);

\path [draw=coral, very thick]
(axis cs:0.45,1)
--(axis cs:0.45,1);

\path [draw=coral, very thick]
(axis cs:0.5,1)
--(axis cs:0.5,1);

\path [draw=coral, very thick]
(axis cs:0.6,1)
--(axis cs:0.6,1);

\path [draw=coral, very thick]
(axis cs:0.7,1)
--(axis cs:0.7,1);

\path [draw=coral, very thick]
(axis cs:0.8,1)
--(axis cs:0.8,1);

\path [draw=coral, very thick]
(axis cs:0.9,1)
--(axis cs:0.9,1);

\path [draw=coral, very thick]
(axis cs:1,1)
--(axis cs:1,1);

\path [draw=lightgreen, very thick]
(axis cs:0.1,0.0960017806621666)
--(axis cs:0.1,0.118899339337833);

\path [draw=lightgreen, very thick]
(axis cs:0.2,0.111723617805319)
--(axis cs:0.2,0.133978702194681);

\path [draw=lightgreen, very thick]
(axis cs:0.3,0.127310921743377)
--(axis cs:0.3,0.152098678256623);

\path [draw=lightgreen, very thick]
(axis cs:0.4,0.14144509108244)
--(axis cs:0.4,0.16983290891756);

\path [draw=lightgreen, very thick]
(axis cs:0.425,0.14759324986449)
--(axis cs:0.425,0.17404659013551);

\path [draw=lightgreen, very thick]
(axis cs:0.45,0.152016487872327)
--(axis cs:0.45,0.178844632127672);

\path [draw=lightgreen, very thick]
(axis cs:0.5,0.162652265840351)
--(axis cs:0.5,0.191578534159649);

\path [draw=lightgreen, very thick]
(axis cs:0.6,0.178901668550148)
--(axis cs:0.6,0.209507211449852);

\path [draw=lightgreen, very thick]
(axis cs:0.7,0.204948044287118)
--(axis cs:0.7,0.231097075712882);

\path [draw=lightgreen, very thick]
(axis cs:0.8,0.223159092349352)
--(axis cs:0.8,0.256573627650648);

\path [draw=lightgreen, very thick]
(axis cs:0.9,0.251969611526324)
--(axis cs:0.9,0.286221908473676);

\path [draw=lightgreen, very thick]
(axis cs:1,0.280293265617695)
--(axis cs:1,0.321276974382305);

\addplot [semithick, blue, mark=asterisk, mark size=3, mark options={solid}, only marks]
table {%
0.1 0.398095237809511
0.2 0.478690264054191
0.3 0.567511824226692
0.4 0.840413360174617
0.425 0.973677042570724
0.45 0.996083275196342
0.5 1
0.6 1
0.7 1
0.8 1
0.9 1
1 1
};
\addlegendentry{AMP}
\addplot [semithick, red, dotted, mark=asterisk, mark size=3, mark options={solid}]
table {%
0.1 0.180130367857636
0.2 0.273049376016941
0.3 0.392585808693201
0.4 0.704753204686499
0.425 0.911909928322673
0.45 1
0.5 1
0.6 1
0.7 1
0.8 1
0.9 1
1 1

};
\addlegendentry{LP}

\addplot [semithick, green, dotted, mark=asterisk, mark size=3, mark options={solid}]
table {%
0.1 0.10745056
0.2 0.12285116
0.3 0.1397048
0.4 0.155639
0.425 0.16081992
0.45 0.16543056
0.5 0.1771154
0.6 0.19420444
0.7 0.21802256
0.8 0.23986636
0.9 0.26909576
1 0.30078512
};
\addlegendentry{IHT}

\end{axis}

\end{tikzpicture}
  \vspace{-2\baselineskip}
  \caption{$\sigma=0$}
\end{subfigure}
\begin{subfigure}[b]{0.32\textwidth}
  \centering
  % This file was created with tikzplotlib v0.10.1.
\begin{tikzpicture}[scale=0.62]

\definecolor{coral}{RGB}{255,127,80}
\definecolor{darkgray176}{RGB}{176,176,176}
\definecolor{green}{RGB}{0,128,0}
\definecolor{lightblue}{RGB}{173,216,230}
\definecolor{lightgray204}{RGB}{204,204,204}
\definecolor{lightgreen}{RGB}{144,238,144}

\begin{axis}[
legend cell align={left},
legend style={
  fill opacity=0.8,
  draw opacity=1,
  text opacity=1,
  at={(0.03,0.97)},
  anchor=north west,
  draw=lightgray204
},
tick align=outside,
tick pos=left,
x grid style={darkgray176},
xlabel={\(\displaystyle \delta=n/p\)},
xmajorgrids,
xmin=0.055, xmax=1.045,
xtick style={color=black},
y grid style={darkgray176},
ylabel={},
ymajorgrids,
ymin=0.0511230254246438, ymax=1.0494485105155,
ytick style={color=black}
]
\addplot [semithick, blue, dashed]
table {%
0.1 0.394750737594008
0.11 0.401102734110985
0.12 0.407504153598908
0.13 0.413854222643159
0.14 0.420364847034653
0.15 0.426902972688321
0.16 0.433450928510796
0.17 0.440007290063962
0.18 0.44653433682457
0.19 0.453503026577744
0.2 0.460026958964305
0.21 0.467109578107782
0.22 0.474015935026378
0.23 0.480839640215684
0.24 0.487996909198077
0.25 0.494723851247812
0.26 0.501809424111564
0.27 0.509026433893795
0.28 0.516577607089385
0.29 0.523442005374094
0.3 0.530533958525853
0.31 0.538676684901148
0.32 0.546121055487821
0.33 0.553706457989931
0.34 0.561366399610509
0.35 0.569460507844272
0.36 0.577257023659157
0.37 0.585554348288076
0.38 0.59346380290169
0.39 0.602074825643166
0.4 0.61050569179641
0.41 0.619097841984781
0.42 0.628040493932219
0.43 0.637029933253247
0.44 0.646141256579013
0.45 0.655522714202737
0.46 0.664336668902073
0.47 0.674620676190409
0.48 0.685171556998934
0.49 0.695963801990186
0.5 0.706039363440961
0.51 0.718046350675684
0.52 0.729131498585593
0.53 0.743074228604873
0.54 0.756420842173124
0.55 0.770201256846159
0.56 0.785121548945657
0.57 0.801040628799922
0.58 0.822976636435099
0.59 0.846251808204573
0.6 0.874352779245844
0.61 0.914950215335003
0.62 0.941770149614612
0.63 0.955886331726243
0.64 0.963037388013565
0.65 0.968581720023253
0.66 0.972622658949717
0.67 0.975919453528258
0.68 0.978526786192333
0.69 0.980631232869335
0.7 0.982058846376454
0.71 0.983848053915333
0.72 0.985089542556667
0.73 0.986280393336366
0.74 0.987330709541969
0.75 0.988318616967567
0.76 0.989164882720643
0.77 0.989865570628637
0.78 0.990659738398747
0.79 0.991143825213413
0.8 0.991753461024542
0.81 0.99233267591852
0.82 0.992809769126776
0.83 0.993220881541492
0.84 0.993620550629024
0.85 0.994012237339405
0.86 0.994316536405089
0.87 0.994662734556923
0.88 0.994921783241566
0.89 0.995182276538951
0.9 0.995407849775595
0.91 0.995646721248128
0.92 0.995899352133109
0.93 0.996140190390667
0.94 0.996356021650168
0.95 0.996503550348595
0.96 0.996662531844595
0.97 0.996866923679183
0.98 0.997010307123754
0.99 0.997148081202414
1 0.997274212075975
};
\addlegendentry{SE}
\path [draw=lightblue, very thick]
(axis cs:0.1,0.377759112886967)
--(axis cs:0.1,0.400320214505027);

\path [draw=lightblue, very thick]
(axis cs:0.2,0.439858368346888)
--(axis cs:0.2,0.471900080134251);

\path [draw=lightblue, very thick]
(axis cs:0.3,0.503122292541632)
--(axis cs:0.3,0.549398303857115);

\path [draw=lightblue, very thick]
(axis cs:0.4,0.579059944228449)
--(axis cs:0.4,0.633367807502222);

\path [draw=lightblue, very thick]
(axis cs:0.5,0.645286055663266)
--(axis cs:0.5,0.744216083730753);

\path [draw=lightblue, very thick]
(axis cs:0.6,0.853391308612245)
--(axis cs:0.6,1.00012368808223);

\path [draw=lightblue, very thick]
(axis cs:0.65,0.954430755591379)
--(axis cs:0.65,1.00407007937501);

\path [draw=lightblue, very thick]
(axis cs:0.7,0.984743361014746)
--(axis cs:0.7,0.99948634701178);

\path [draw=lightblue, very thick]
(axis cs:0.8,0.994301201636112)
--(axis cs:0.8,0.999953100385583);

\path [draw=lightblue, very thick]
(axis cs:0.9,0.996907578441162)
--(axis cs:0.9,1.00037380099218);

\path [draw=lightblue, very thick]
(axis cs:1,0.997939722298454)
--(axis cs:1,1.00059747207646);

\path [draw=coral, very thick]
(axis cs:0.1,0.381717772977455)
--(axis cs:0.1,0.401249314793814);

\path [draw=coral, very thick]
(axis cs:0.2,0.427540991043901)
--(axis cs:0.2,0.456761144332656);

\path [draw=coral, very thick]
(axis cs:0.3,0.468833200838628)
--(axis cs:0.3,0.511509197100767);

\path [draw=coral, very thick]
(axis cs:0.4,0.537382093236265)
--(axis cs:0.4,0.602051877793751);

\path [draw=coral, very thick]
(axis cs:0.5,0.665145946217567)
--(axis cs:0.5,0.743904704237514);

\path [draw=coral, very thick]
(axis cs:0.6,0.796262446673691)
--(axis cs:0.6,0.8518568372361);

\path [draw=coral, very thick]
(axis cs:0.65,0.837269938399195)
--(axis cs:0.65,0.881970085818851);

\path [draw=coral, very thick]
(axis cs:0.7,0.870638162032521)
--(axis cs:0.7,0.898901834996415);

\path [draw=coral, very thick]
(axis cs:0.8,0.90870617406306)
--(axis cs:0.8,0.929647210394011);

\path [draw=coral, very thick]
(axis cs:0.9,0.928609906249488)
--(axis cs:0.9,0.94284695155216);

\path [draw=coral, very thick]
(axis cs:1,0.941222778509346)
--(axis cs:1,0.953017355822782);

\path [draw=lightgreen, very thick]
(axis cs:0.1,0.0965014565651374)
--(axis cs:0.1,0.117639423434863);

\path [draw=lightgreen, very thick]
(axis cs:0.2,0.110204666976724)
--(axis cs:0.2,0.133270213023276);

\path [draw=lightgreen, very thick]
(axis cs:0.3,0.127172959567874)
--(axis cs:0.3,0.150462480432126);

\path [draw=lightgreen, very thick]
(axis cs:0.4,0.145682237203235)
--(axis cs:0.4,0.173230642796765);

\path [draw=lightgreen, very thick]
(axis cs:0.5,0.161675119172374)
--(axis cs:0.5,0.190399200827626);

\path [draw=lightgreen, very thick]
(axis cs:0.6,0.178813102825892)
--(axis cs:0.6,0.216251617174108);

\path [draw=lightgreen, very thick]
(axis cs:0.65,0.193079609482456)
--(axis cs:0.65,0.221907910517543);

\path [draw=lightgreen, very thick]
(axis cs:0.7,0.20260276566847)
--(axis cs:0.7,0.23584347433153);

\path [draw=lightgreen, very thick]
(axis cs:0.8,0.227902754710607)
--(axis cs:0.8,0.263177885289393);

\path [draw=lightgreen, very thick]
(axis cs:0.9,0.250651663226786)
--(axis cs:0.9,0.285275456773213);

\path [draw=lightgreen, very thick]
(axis cs:1,0.274531328749673)
--(axis cs:1,0.320184591250327);

\addplot [semithick, blue, mark=asterisk, mark size=3, mark options={solid}, only marks]
table {%
0.1 0.389039663695997
0.2 0.45587922424057
0.3 0.526260298199373
0.4 0.606213875865336
0.5 0.69475106969701
0.6 0.926757498347237
0.65 0.979250417483194
0.7 0.992114854013263
0.8 0.997127151010848
0.9 0.998640689716669
1 0.999268597187455
};
\addlegendentry{AMP}
\addplot [semithick, red, dotted, mark=asterisk, mark size=3, mark options={solid}]
table {%
0.1 0.391483543885635
0.2 0.442151067688278
0.3 0.490171198969698
0.4 0.569716985515008
0.5 0.704525325227541
0.6 0.824059641954895
0.65 0.859620012109023
0.7 0.884769998514468
0.8 0.919176692228536
0.9 0.935728428900824
1 0.947120067166064
};
\addlegendentry{CVX}
\addplot [semithick, green, dotted, mark=asterisk, mark size=3, mark options={solid}]
table {%
0.1 0.10707044
0.2 0.12173744
0.3 0.13881772
0.4 0.15945644
0.5 0.17603716
0.6 0.19753236
0.65 0.20749376
0.7 0.21922312
0.8 0.24554032
0.9 0.26796356
1 0.29735796
};
\addlegendentry{IHT}
\end{axis}

\end{tikzpicture}
  \vspace{-2\baselineskip}
  \caption{$\sigma=0.1$}
\end{subfigure}
\begin{subfigure}[b]{0.32\textwidth}
  \centering
  % This file was created with tikzplotlib v0.10.1.
\begin{tikzpicture}[scale=0.62]

\definecolor{coral}{RGB}{255,127,80}
\definecolor{darkgray176}{RGB}{176,176,176}
\definecolor{green}{RGB}{0,128,0}
\definecolor{lightblue}{RGB}{173,216,230}
\definecolor{lightgray204}{RGB}{204,204,204}
\definecolor{lightgreen}{RGB}{144,238,144}

\begin{axis}[
legend cell align={left},
legend style={
  fill opacity=0.8,
  draw opacity=1,
  text opacity=1,
  at={(0.97,0.03)},
  anchor=south east,
  draw=lightgray204
},
tick align=outside,
tick pos=left,
x grid style={darkgray176},
xlabel={\(\displaystyle \delta=n/p\)},
xmajorgrids,
xmin=0.125, xmax=8.375,
xtick style={color=black},
y grid style={darkgray176},
ylabel={},
ymajorgrids,
ymin=0.12844508404101, ymax=1.04301350686402,
ytick style={color=black}
]
\addplot [semithick, blue, dashed]
table {%
0.1 0.365277923861676
0.2 0.396878700378293
0.3 0.427862276062578
0.4 0.458235948119386
0.5 0.487982075748194
0.6 0.516919561318509
0.7 0.545134094728138
0.8 0.572411889498814
0.9 0.599142147956491
1 0.624361140371856
1.1 0.648627924040528
1.2 0.671757811632503
1.3 0.694833888529464
1.4 0.715991197378406
1.5 0.736272319572272
1.6 0.754665816371956
1.7 0.772879391357914
1.8 0.789780557907651
1.9 0.805411607764069
2 0.819710631537266
2.1 0.833050004553917
2.2 0.845984773540095
2.3 0.857026550471233
2.4 0.867999887147377
2.5 0.878036703541913
2.6 0.887044958898585
2.7 0.895214394311874
2.8 0.902884764815302
2.9 0.910265015208059
3 0.916755158238744
3.1 0.922649565686425
3.2 0.928304195273914
3.3 0.933056114425877
3.4 0.937881794186309
3.5 0.942300539553774
3.6 0.946256318838111
3.7 0.949999045956711
3.8 0.953376365501665
3.9 0.956533969953421
4 0.959308776677712
4.1 0.962045918116405
4.2 0.964537191090653
4.3 0.966755423588338
4.4 0.968980958971377
4.5 0.970962819520954
4.6 0.972859315840536
4.7 0.974404012624887
4.8 0.975979639742785
4.9 0.977692983937694
5 0.978961647892365
5.1 0.980200065110014
5.2 0.981377834273689
5.3 0.982413989385294
5.4 0.983564311620661
5.5 0.984430676539658
5.6 0.985442998559089
5.7 0.986212768317988
5.8 0.987067015409794
5.9 0.98772808355753
6 0.988469337213461
6.1 0.989073881213885
6.2 0.989741700174506
6.3 0.990256085915156
6.4 0.990767973642877
6.5 0.991334270437752
6.6 0.991831247205402
6.7 0.992269144164146
6.8 0.99266348819405
6.9 0.99299525299524
7 0.993388030894329
7.1 0.993727630292459
7.2 0.994035859005326
7.3 0.994335877258922
7.4 0.994704150275344
7.5 0.994928277213971
7.6 0.99516938858466
7.7 0.995488098370274
7.8 0.995652707158074
7.9 0.995878644435148
8 0.996084407462083
};
\addlegendentry{SE}
\path [draw=lightblue, very thick]
(axis cs:0.5,0.467059314286518)
--(axis cs:0.5,0.503000083798745);

\path [draw=lightblue, very thick]
(axis cs:1,0.603263332124828)
--(axis cs:1,0.65053941976902);

\path [draw=lightblue, very thick]
(axis cs:1.5,0.721108189022945)
--(axis cs:1.5,0.767436375721971);

\path [draw=lightblue, very thick]
(axis cs:2,0.808008878305889)
--(axis cs:2,0.849673940675011);

\path [draw=lightblue, very thick]
(axis cs:2.5,0.875358854800623)
--(axis cs:2.5,0.911503589226701);

\path [draw=lightblue, very thick]
(axis cs:3,0.915921844677431)
--(axis cs:3,0.942647121455138);

\path [draw=lightblue, very thick]
(axis cs:3.5,0.944609628244464)
--(axis cs:3.5,0.965305641536539);

\path [draw=lightblue, very thick]
(axis cs:4,0.961746892388531)
--(axis cs:4,0.977186692936406);

\path [draw=lightblue, very thick]
(axis cs:4.5,0.973569896970004)
--(axis cs:4.5,0.9888772457423);

\path [draw=lightblue, very thick]
(axis cs:5,0.980730364470341)
--(axis cs:5,0.993083867642158);

\path [draw=lightblue, very thick]
(axis cs:5.5,0.987855762091749)
--(axis cs:5.5,0.996305035817669);

\path [draw=lightblue, very thick]
(axis cs:6,0.990492073638534)
--(axis cs:6,0.998155277084026);

\path [draw=lightblue, very thick]
(axis cs:6.5,0.993260277618384)
--(axis cs:6.5,0.999166015327874);

\path [draw=lightblue, very thick]
(axis cs:7,0.994698193097563)
--(axis cs:7,1.00007706859367);

\path [draw=lightblue, very thick]
(axis cs:7.5,0.995494886170005)
--(axis cs:7.5,1.00029792865341);

\path [draw=lightblue, very thick]
(axis cs:8,0.996615537305598)
--(axis cs:8,1.00024687541595);

\path [draw=coral, very thick]
(axis cs:0.5,0.382404240777499)
--(axis cs:0.5,0.435506643777366);

\path [draw=coral, very thick]
(axis cs:1,0.582701285786804)
--(axis cs:1,0.640351657920853);

\path [draw=coral, very thick]
(axis cs:1.5,0.720284375236177)
--(axis cs:1.5,0.763049191741931);

\path [draw=coral, very thick]
(axis cs:2,0.796890828191543)
--(axis cs:2,0.831046731528439);

\path [draw=coral, very thick]
(axis cs:2.5,0.847914745020423)
--(axis cs:2.5,0.873108895288515);

\path [draw=coral, very thick]
(axis cs:3,0.876512619329428)
--(axis cs:3,0.896530848076491);

\path [draw=coral, very thick]
(axis cs:3.5,0.89876368271148)
--(axis cs:3.5,0.915488115062978);

\path [draw=coral, very thick]
(axis cs:4,0.912867633319886)
--(axis cs:4,0.925357340811593);

\path [draw=coral, very thick]
(axis cs:4.5,0.924560336456959)
--(axis cs:4.5,0.938485243428912);

\path [draw=coral, very thick]
(axis cs:5,0.934144485442574)
--(axis cs:5,0.944872601657897);

\path [draw=coral, very thick]
(axis cs:5.5,0.941417904387318)
--(axis cs:5.5,0.950343186566512);

\path [draw=coral, very thick]
(axis cs:6,0.94616212185044)
--(axis cs:6,0.955058795391017);

\path [draw=coral, very thick]
(axis cs:6.5,0.952278915289811)
--(axis cs:6.5,0.959523855408836);

\path [draw=coral, very thick]
(axis cs:7,0.955194073698554)
--(axis cs:7,0.961766384366315);

\path [draw=coral, very thick]
(axis cs:7.5,0.957996479385211)
--(axis cs:7.5,0.964703939110249);

\path [draw=coral, very thick]
(axis cs:8,0.961161101055174)
--(axis cs:8,0.967071835295805);

\path [draw=lightgreen, very thick]
(axis cs:0.5,0.170016375987511)
--(axis cs:0.5,0.201946584012489);

\path [draw=lightgreen, very thick]
(axis cs:1,0.276598653504844)
--(axis cs:1,0.331764386495156);

\path [draw=lightgreen, very thick]
(axis cs:1.5,0.44110564167471)
--(axis cs:1.5,0.53173483832529);

\path [draw=lightgreen, very thick]
(axis cs:2,0.657342617517199)
--(axis cs:2,0.742981222482801);

\path [draw=lightgreen, very thick]
(axis cs:2.5,0.796662358701322)
--(axis cs:2.5,0.864836281298678);

\path [draw=lightgreen, very thick]
(axis cs:3,0.862399141553541)
--(axis cs:3,0.911923418446459);

\path [draw=lightgreen, very thick]
(axis cs:3.5,0.912041311072068)
--(axis cs:3.5,0.947569008927932);

\path [draw=lightgreen, very thick]
(axis cs:4,0.935745019328055)
--(axis cs:4,0.965937380671945);

\path [draw=lightgreen, very thick]
(axis cs:4.5,0.953462789825911)
--(axis cs:4.5,0.982766730174089);

\path [draw=lightgreen, very thick]
(axis cs:5,0.966818462697112)
--(axis cs:5,0.988932337302888);

\path [draw=lightgreen, very thick]
(axis cs:5.5,0.976828132279351)
--(axis cs:5.5,0.994275947720649);

\path [draw=lightgreen, very thick]
(axis cs:6,0.982344563251175)
--(axis cs:6,0.998135036748825);

\path [draw=lightgreen, very thick]
(axis cs:6.5,0.986034685267847)
--(axis cs:6.5,0.999294594732153);

\path [draw=lightgreen, very thick]
(axis cs:7,0.989590945359631)
--(axis cs:7,1.00099657464037);

\path [draw=lightgreen, very thick]
(axis cs:7.5,0.991934665082485)
--(axis cs:7.5,1.00144221491752);

\path [draw=lightgreen, very thick]
(axis cs:8,0.993053124226738)
--(axis cs:8,1.00119935577326);

\addplot [semithick, blue, mark=asterisk, mark size=3, mark options={solid}, only marks]
table {%
0.5 0.485029699042631
1 0.626901375946924
1.5 0.744272282372458
2 0.82884140949045
2.5 0.893431222013662
3 0.929284483066284
3.5 0.954957634890501
4 0.969466792662468
4.5 0.981223571356152
5 0.98690711605625
5.5 0.992080398954709
6 0.99432367536128
6.5 0.996213146473129
7 0.997387630845617
7.5 0.997896407411706
8 0.998431206360773
};
\addlegendentry{AMP}
\addplot [semithick, red, dotted, mark=asterisk, mark size=3, mark options={solid}]
table {%
0.5 0.408955442277432
1 0.611526471853828
1.5 0.741666783489054
2 0.813968779859991
2.5 0.860511820154469
3 0.886521733702959
3.5 0.907125898887229
4 0.91911248706574
4.5 0.931522789942935
5 0.939508543550235
5.5 0.945880545476915
6 0.950610458620729
6.5 0.955901385349323
7 0.958480229032435
7.5 0.96135020924773
8 0.96411646817549
};
\addlegendentry{CVX}
\addplot [semithick, green, dotted, mark=asterisk, mark size=3, mark options={solid}]
table {%
0.5 0.18598148
1 0.30418152
1.5 0.48642024
2 0.70016192
2.5 0.83074932
3 0.88716128
3.5 0.92980516
4 0.9508412
4.5 0.96811476
5 0.9778754
5.5 0.98555204
6 0.9902398
6.5 0.99266464
7 0.99529376
7.5 0.99668844
8 0.99712624
};
\addlegendentry{IHT}
\end{axis}

\end{tikzpicture}
  \vspace{-2\baselineskip}
  \caption{$\sigma=0.3$}
\end{subfigure}
\caption{AMP vs. other algorithms for pooled data: squared overlap vs.~$\delta$, with $L=3$ and $\pi=[1/3,1/3,1/3]$. The plots are similar for the case of non-uniform priors.}
\label{fig:AMP_v_others_noiseless}
\end{figure}
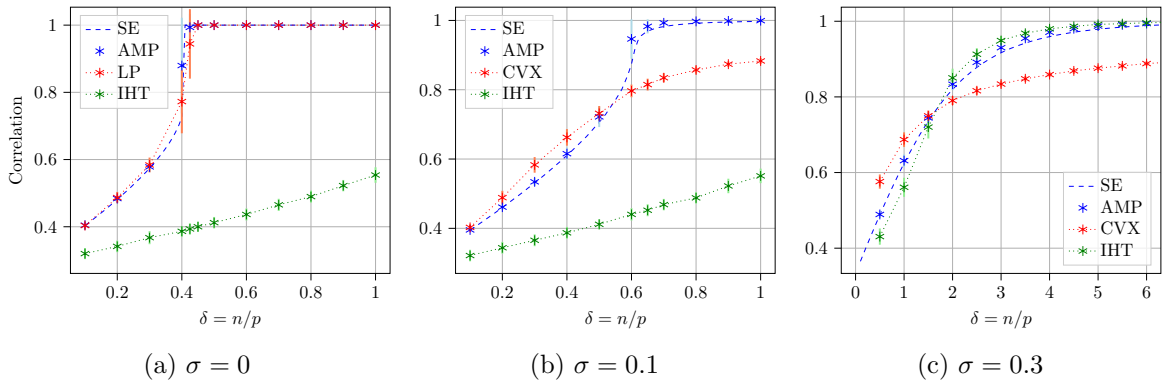

\begin{figure}[t]
\centering
\begin{subfigure}[b]{0.45\textwidth}
  \centering
  \include{new_corr_figs/pool_fig3a_new}
  \vspace{-2\baselineskip}
  \caption{$L=2$}
\end{subfigure}
\begin{subfigure}[b]{0.45\textwidth}
  \centering
  \include{new_corr_figs/pool_fig3b_new}
  \vspace{-2\baselineskip}
  \caption{$L=3$}
\end{subfigure}
\caption{AMP for pooled data: squared overlap vs.~$\delta$ for varying $\sigma$ with uniform prior $\pi$.}
\label{fig:noisy}
\end{figure}

\paragraph{Simulation results} Figure \ref{fig:noiseless} shows the performance of the matrix-AMP algorithm in the noiseless setting, for various $\pi$. The squared overlap is plotted as function of the sampling ratio $\delta=n/p$, for different categorical priors. The state evolution predictions closely match the performance of the matrix-AMP algorithm, validating the result of Theorem \ref{thm:GAMP}. As expected, the overlap improves with increasing $\delta$. Furthermore, for both $L=2$ and $L=3$, the more non-uniform the $\pi$, the better the performance of the matrix-AMP algorithm. Figure \ref{fig:AMP_v_others_noiseless} shows how the matrix-AMP algorithm compares with LP and IHT for different noise levels. We observe that the matrix-AMP algorithm consistently has the best performance, for all levels of noise $\sigma$ and sampling ratio $\delta$ considered in our experiments.
% We observe that the optimization-based methods (LP for noiseless, CVX for noisy) outperforms the matrix-AMP algorithm when there is noise present and $\delta$ is small, indicating the sub-optimality of AMP in those setting. For noise level $\sigma=0.3$, IHT slightly outperforms the AMP algorithm for larger values of $\delta$.
We note that there are currently no theoretical guarantees for LP, CVX, or IHT applied to pooled data. Furthermore, AMP has a few advantages: i) Under the assumption of Gaussian noise, the Bayes-optimal denoiser does not require the noise variance (see Appendix \ref{sec:imp_details_pooled_data}), and ii) AMP is significantly faster than the optimization-based methods for larger values of $n$ and $p$.
Figure \ref{fig:noisy} shows the performance curves of the matrix-AMP algorithm for different values of noise level $\sigma$; as expected, overlap improves as $\sigma$ decreases.

\begin{remark}
  In the journal version \cite{tan24} of this paper, instead of the squared overlap, we used the correlation $\frac{1}{p}\Tr\big((B^k)^\top B\big)$ in the simulations of pooled data. This led to the convex estimator performing slightly better than the matrix-AMP algorithm for noisy pooled data when the sampling ratio $\delta$ is small. It turns out that this is an artifact of the correlation, namely, the correlation is sensitive to the scaling of $\hB^k$ whereas the squared overlap is insensitive to the scaling of $\hB^k$. In all our experiments, we also found AMP has lower MSE compared to the convex and IHT estimators.  We would like to thank Galen Reeves and Kevin Xu for highlighting the issue of scaling.  
\end{remark} 

So far, we have assumed that the matrix-AMP algorithm has access to the empirical category proportions $\hat{\pi}$. We now study the performance when  $\hat{\pi}$ is not known exactly, and the algorithm uses an estimate $\pi_{\text{est}}$. We consider $L=2$ with  $\pi_{\text{est}}=[\hat{\pi}_1+\epsilon,\hat{\pi}_2-\epsilon]$, where $\epsilon >0$ quantifies the accuracy of the estimate.  Figure \ref{fig:est_pi} shows the performance for $\epsilon=0.01$ and $\epsilon=0.05$. We observe that the performance of the matrix-AMP algorithm is very sensitive to the accuracy of the estimate $\pi_{\text{est}}$, i.e., a small increase in $\epsilon$ can substantially degrade the performance. The underlying reason that AMP requires each entry of $\tY$ to be of constant order (as $p$ grows) with high probability.  This is true if $\hat{\pi}$ is known exactly; we show in Appendix \ref{sec:scaling_deriv} that $\E[\tY_{il}]=0$ and $\Var[\tY_{il}]=\frac{p\pi_l}{n}=\Theta(1)$, for $i \in [n], l \in [L]$. 
However, a constant shift of $\epsilon$ in the entries of $\hat{\pi}$ causes the  the entries of $\tY$ to jump from $\Theta(1)$ to $\Theta(\sqrt{p})$ -- see Appendix \ref{sec:scaling_deriv} for details. This indicates that it is worth performing additional (non-random) test(s) to obtain an accurate estimate of $\hat{\pi}$ before running the matrix-AMP algorithm.

\begin{figure}[t]
\centering
% This file was created with tikzplotlib v0.10.1.
\begin{tikzpicture}[scale=1]

\definecolor{coral}{RGB}{255,127,80}
\definecolor{darkgray176}{RGB}{176,176,176}
\definecolor{green}{RGB}{0,128,0}
\definecolor{lightblue}{RGB}{173,216,230}
\definecolor{lightgray204}{RGB}{204,204,204}
\definecolor{lightgreen}{RGB}{144,238,144}

\begin{axis}[
legend cell align={left},
legend style={
  fill opacity=0.8,
  draw opacity=1,
  text opacity=1,
  at={(0.03,0.97)},
  anchor=north west,
  draw=lightgray204
},
tick align=outside,
tick pos=left,
x grid style={darkgray176},
xlabel={\(\displaystyle \delta=n/p\)},
xmajorgrids,
xmin=0.055, xmax=1.045,
xtick style={color=black},
y grid style={darkgray176},
ylabel={Overlap},
ymajorgrids,
ymin=0.377199606697121, ymax=1.05805237674326,
ytick style={color=black}
]
\path [draw=lightblue, very thick]
(axis cs:0.1,0.539709068676808)
--(axis cs:0.1,0.559520195181957);

\path [draw=lightblue, very thick]
(axis cs:0.2,0.58252149019317)
--(axis cs:0.2,0.615853597755849);

\path [draw=lightblue, very thick]
(axis cs:0.3,0.639270807486914)
--(axis cs:0.3,0.670236651824084);

\path [draw=lightblue, very thick]
(axis cs:0.4,0.69978295204449)
--(axis cs:0.4,0.749454048571342);

\path [draw=lightblue, very thick]
(axis cs:0.5,0.810726218812206)
--(axis cs:0.5,1.02690693428899);

\path [draw=lightblue, very thick]
(axis cs:0.6,1)
--(axis cs:0.6,1);

\path [draw=lightblue, very thick]
(axis cs:0.7,1)
--(axis cs:0.7,1);

\path [draw=lightblue, very thick]
(axis cs:0.8,1)
--(axis cs:0.8,1);

\path [draw=lightblue, very thick]
(axis cs:0.9,1)
--(axis cs:0.9,1);

\path [draw=lightblue, very thick]
(axis cs:1,1)
--(axis cs:1,1);

\path [draw=coral, very thick]
(axis cs:0.1,0.526483935508361)
--(axis cs:0.1,0.551917623949018);

\path [draw=coral, very thick]
(axis cs:0.2,0.557876849538037)
--(axis cs:0.2,0.595432822346675);

\path [draw=coral, very thick]
(axis cs:0.3,0.595869743016167)
--(axis cs:0.3,0.632198866859257);

\path [draw=coral, very thick]
(axis cs:0.4,0.635746441944884)
--(axis cs:0.4,0.680103023069901);

\path [draw=coral, very thick]
(axis cs:0.5,0.668166060013925)
--(axis cs:0.5,0.725041765832378);

\path [draw=coral, very thick]
(axis cs:0.6,0.712713919296966)
--(axis cs:0.6,0.768622891051032);

\path [draw=coral, very thick]
(axis cs:0.7,0.743474719526259)
--(axis cs:0.7,0.830462378792025);

\path [draw=coral, very thick]
(axis cs:0.8,0.793841704674699)
--(axis cs:0.8,0.877920644087636);

\path [draw=coral, very thick]
(axis cs:0.9,0.842722680681687)
--(axis cs:0.9,0.911162037062779);

\path [draw=coral, very thick]
(axis cs:1,0.898609305475511)
--(axis cs:1,0.954251622199397);

\path [draw=lightgreen, very thick]
(axis cs:0.1,0.42877154172654)
--(axis cs:0.1,0.46780470661265);

\path [draw=lightgreen, very thick]
(axis cs:0.2,0.419767142670382)
--(axis cs:0.2,0.4641252271145);

\path [draw=lightgreen, very thick]
(axis cs:0.3,0.420391374533707)
--(axis cs:0.3,0.467369186312067);

\path [draw=lightgreen, very thick]
(axis cs:0.4,0.428828819380468)
--(axis cs:0.4,0.46939842769818);

\path [draw=lightgreen, very thick]
(axis cs:0.5,0.433227185281009)
--(axis cs:0.5,0.479358156881609);

\path [draw=lightgreen, very thick]
(axis cs:0.6,0.43954593428544)
--(axis cs:0.6,0.490018110192165);

\path [draw=lightgreen, very thick]
(axis cs:0.7,0.451764026181902)
--(axis cs:0.7,0.500180216831681);

\path [draw=lightgreen, very thick]
(axis cs:0.8,0.456220658215766)
--(axis cs:0.8,0.50357163211023);

\path [draw=lightgreen, very thick]
(axis cs:0.9,0.461041941545406)
--(axis cs:0.9,0.507950911887376);

\path [draw=lightgreen, very thick]
(axis cs:1,0.470380473451606)
--(axis cs:1,0.521910197844338);

\addplot [semithick, blue, dotted, mark=asterisk, mark size=3, mark options={solid}]
table {%
0.1 0.549614631929383
0.2 0.59918754397451
0.3 0.654753729655499
0.4 0.724618500307916
0.5 0.918816576550599
0.6 1
0.7 1
0.8 1
0.9 1
1 1
};
\addlegendentry{AMP, known $\hat{\pi}$}
\addplot [semithick, red, dotted, mark=asterisk, mark size=3, mark options={solid}]
table {%
0.1 0.539200779728689
0.2 0.576654835942356
0.3 0.614034304937712
0.4 0.657924732507392
0.5 0.696603912923151
0.6 0.740668405173999
0.7 0.786968549159142
0.8 0.835881174381167
0.9 0.876942358872233
1 0.926430463837454
};
\addlegendentry{AMP, $\hat{\pi}\pm 0.01$}
\addplot [semithick, green, dotted, mark=asterisk, mark size=3, mark options={solid}]
table {%
0.1 0.448288124169595
0.2 0.441946184892441
0.3 0.443880280422887
0.4 0.449113623539324
0.5 0.456292671081309
0.6 0.464782022238803
0.7 0.475972121506792
0.8 0.479896145162998
0.9 0.484496426716391
1 0.496145335647972
};
\addlegendentry{AMP, $\hat{\pi}\pm 0.05$}
\end{axis}

\end{tikzpicture}
\vspace{-2\baselineskip}
\caption{Pooled data with mismatched proportions: squared overlap vs.~$\delta$ with $\pi=[0.5,0.5]$ and $\sigma=0$. Red curves and green curves show AMP performance with the estimate $\pi_{\text{est}}=[\hat{\pi}_1+\epsilon,\hat{\pi}_2-\epsilon]$ for $\epsilon=0.01, 0.05$.}
\label{fig:est_pi}
\end{figure}
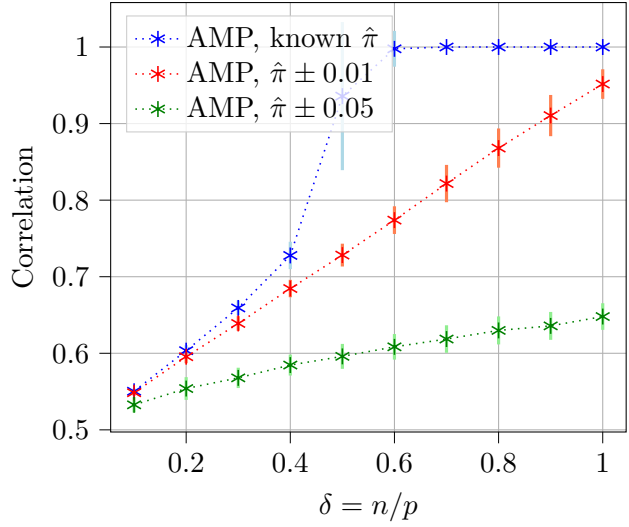

%%%%%
\section{AMP for Quantitative Group Testing}
\label{sec:AMP_QGT}

Recall that Quantitative Group Testing (QGT) is a special case of the pooled data problem with $L=2$, where items are either defective or non-defective. As mentioned in Section \ref{subsec:pooled_setup}, in QGT the signal is often represented by a vector $ \beta \in \{0,1\}^p$, where $1$ corresponds to a `defective' item and $0$ corresponds to a non-defective item. Using a vector instead of a $p \times 2$ matrix to represent the QGT signal has two major benefits: (i) It allows us to present guarantees for the false positive rate (FPR) and false negative rate (FNR) in a simpler way, and (ii) It provides a simpler and more efficient AMP algorithm.

We use $\beta$ (instead of $B$)  for the signal vector to avoid confusion between the pooled data problem and QGT.  The design matrix is the same as that for pooled data, i.e., $X_{ij}\sim_{\iid}\text{Bernoulli}(\alpha)$ where $\alpha\in(0,1)$. 
The vector of test outcomes $Y \in \reals^n$  
is given by
\begin{align}
 Y= X \beta \, + \, \Psi, \quad i \in [n],
    \label{eq:QGT}
\end{align}
where $\Psi \in \reals^n$ is the additive noise. Recall that $\hat{\pi}$ is the proportion of defective items in the population.  Substituting the decomposition of $X$ in \eqref{eq:X_decomp} into \eqref{eq:QGT} and rearranging, we get
\begin{align*}
    \frac{Y-\alpha\bs{1}_n\bs{1}_p^\top\beta}{\sqrt{n\alpha(1-\alpha)}}=\tX\beta+\frac{\Psi}{\sqrt{n\alpha(1-\alpha)}}.
\end{align*}
Since $\alpha\bs{1}_n\bs{1}_p^\top \beta=\alpha p \left(\frac{1}{p}\bs{1}_n\bs{1}_p^\top \beta\right)=\alpha p[\hat{\pi},\dots,\hat{\pi}]^\top$, we can rewrite the noisy model as
\begin{align}
    \frac{Y_i-\alpha p\hat{\pi}}{\sqrt{n\alpha(1-\alpha)}}
    =\tX_{i,:}^\top\beta+\frac{\Psi_i}{\sqrt{n\alpha(1-\alpha)}},
    \label{eq:QGT_scaling}
\end{align}
which can be further rewritten as
\begin{align}
    \tY_i=\tX_{i,:}^\top \, \beta \,  + \, \tPsi_i, \quad i \in [n], \quad \text{ or } \quad    \tY=\tX\beta+\tPsi.
    \label{eq:noisy_QGT_rescaled}
\end{align}
As discussed in Section \ref{sec:AMP_pooled_data}, choosing $\alpha=0.5$ minimizes the  variance of the rescaled noise $\tPsi$. The noiseless QGT model is obtained by setting $\Psi =\tPsi=0_n$.

\subsection{Algorithm} 

As in the pooled data problem, we assume knowledge of $\hat{\pi}$ and apply the AMP algorithm with the modified data $\tX$ and $\tY$. To differentiate the AMP algorithm for QGT from the pooled data setting, we change  notation, and  replace $B, B^k$, $\hB^k$, with $\beta, \beta^k$, $\hat{\beta}^k$, respectively, for $k \ge 1$. The notation for state evolution parameters is also changed to emphasize that in QGT, they are scalars. We let $\mu_\beta^k:=\Mu_\beta^k$, $(\sigma_\beta^t)^2:=\big\{\Tau_\beta^k\big\}_{t,t}$ for $t\in[k]$, $\mu_\Theta^k:=\Mu_\Theta^k$, and $(\sigma_\Theta^t)^2:=\big\{\Tau_\Theta^k\big\}_{t+1,t+1}$ for $t\in\{0,\dots,k\}$. In these definitions, the quantities on the right  are entries of the original state evolution matrices in \eqref{eq:SE_Mk1B}-\eqref{eq:Tk_theta_def}. 

We assume that the empirical distributions of $\beta \in \reals^{p}$ and $\tPsi \in \reals^n$ converge to well-defined limits as $p,n \to \infty$.
Specifically, assume that  $\beta  \stackrel{W}{\rightarrow} \bar{\beta}$ and 
$\tPsi  \stackrel{W}{\rightarrow} \bar{\Psi}$. Here $\bar{\beta}$ is a binary random variable whose law represents the limiting proportion of $1$s in $\beta$. 

The AMP algorithm is initialized with a random $\hat{\beta}^0$ whose components are generated i.i.d.~according to the law of $\bar{\beta}$. Then the initial covariance $\Sigma^0$ in \eqref{eq:Sig0_def} can be computed as:
\begin{align*}
 \frac{1}{n}
            \begin{bmatrix}
                \beta^\top \beta & \beta^\top\hat{\beta}^0 \\
                (\hat{\beta}^0)^\top \beta & (\hat{\beta}^0)^\top\hat{\beta}^0
            \end{bmatrix}
 \, \stackrel{a.s.}{\rightarrow} \, 
  \Sigma^0 =  \frac{1}{\delta} \begin{bmatrix}
        \E\{ \bar{\beta}^2 \} & (\E\{ \bar{\beta}\})^2  \\
       (\E\{ \bar{\beta}\})^2 &  \E\{ \bar{\beta}^2 \}
    \end{bmatrix}.
\end{align*}

With these assumptions, Theorem \ref{thm:GAMP} directly gives the following state evolution result.

\begin{theorem} \label{thm:GAMP_vec}
Consider the AMP algorithm in \eqref{eq:GAMP} for the QGT model in \eqref{eq:noisy_QGT_rescaled}, with the  notational changes,  assumptions, and initialization described above.  Assume that the denoising functions $f_{k+1}, g_k$ used in the AMP algorithm satisfy Assumption 
\textbf{(A2)} in Section \ref{sec:AMP_algo}. Then for each $k \ge 0$,
\begin{align}
    \big(\beta,\hbeta^0,\beta^1,\dots,\beta^{k+1}\big)
    &\stackrel{W_2}{\rightarrow}
    \big(\bar{\beta},\bar{\beta}^0,\mu^{1}_{\beta}\bar{\beta}+G^{1}_\beta,\dots,\mu^{k+1}_{\beta}\bar{\beta}+G^{k+1}_\beta\big) \label{eq:GAMP_vec1} \\
    \big(\tPsi,\Theta,\Theta^0,\dots,\Theta^k\big)
    &\stackrel{W_2}{\rightarrow}
    \big(\bar{\Psi},Z,\mu^{0}_{\Theta}Z+G^{0}_\Theta,\dots,\mu^{k}_{\Theta}Z+G^{k}_\Theta\big),
    \label{eq:GAMP_vec2} 
\end{align}
almost surely as $n,p \to \infty$ with $n/p \to \delta$.
Here $G_\beta^1,\dots,G_\beta^{k+1}\sim\normal(0,\Tau_\beta^{k+1})$ is independent of $\bar{\beta}$, and $G_\Theta^0,\dots,G_\Theta^k\sim\normal(0,\Tau_\Theta^k)$ is independent of $(Z,\bar{\Psi})$. 
\end{theorem}
%%%
%%%
The optimal choices for the AMP denoising functions are given by $f_k^*$ and $g_k^*$ in \eqref{eq:fk_opt_def}-\eqref{eq:gk_opt_def}. 

\paragraph{Performance measures} Theorem \ref{thm:GAMP_vec} allows us to compute the limiting values of performance measures such as
the squared overlap, defined as:
\begin{align}
    \frac{\langle\hbeta^k,\beta\rangle^2}{\|\hbeta^k\|_2^2\|\beta\|_2^2}
    \stackrel{a.s.}{\rightarrow}
    \frac{(\E[f_k(\mu_\beta^k\bar{\beta}+G_\beta^k)\cdot\bar{\beta}])^2}{\E[f_k(\mu_\beta^k\bar{\beta}+G_\beta^k)^2]\cdot\E[\bar{\beta}^2]}.
    \label{eq:norm_corr_vec}
\end{align}
In practical applications of QGT, we may want to understand the false positive and false negative rates (FPR and FNR) separately, or weigh one more than the other based on the practical needs. We proceed to investigate the FPR and FNR of AMP estimates for QGT.  To get a final $K$th estimate in the signal domain $\{0,1\}^p$, we use thresholding to output a hard decision. Theorem \ref{thm:GAMP_vec} guarantees that the empirical distribution of $\beta^K$ converges to the law of $\mu_\beta^K\bar{\beta}+G_\beta^K$. We define the hard decision to be
\begin{align}
    \mathds{1}\left\{\hbeta^K>\zeta\cdot1_p\right\}
    =\mathds{1}\left\{f_K\left(\beta^K\right)>\zeta\cdot 1_p\right\},  \label{eq:threshold_def}
\end{align}
where $1_p$ is the all-ones vector of length $p$, and the indicator function is applied component-wise to $\hbeta^K$. The function \eqref{eq:threshold_def} sets the entries of $\hbeta^K$ larger than the threshold $\zeta$ to one (i.e., defective) and the others to zero (i.e., non-defective). Based on the above function, let us denote the estimated defective set by
\begin{align}
    \widehat{\mathcal{S}}=\left\{j:\hbeta_j^K>\zeta\right\}.
\end{align} 
The FPR and the FNR are then defined as follows:
\begin{align}
    \text{FPR}&=
    \frac{\sum_{j=1}^p\mathds{1}\{\beta_j=0\cap j\in\widehat{\mathcal{S}}\}}{p-\sum_{j=1}^p\beta_j}, \quad
    \text{ and } \quad
    \text{FNR}=
    \frac{\sum_{j=1}^p\mathds{1}\{\beta_j=1\cap j\notin\widehat{\mathcal{S}}\}}{\sum_{j=1}^p\beta_j}.
    \label{eq:FPR_and_FNR}
\end{align}

\begin{corollary} \label{cor:FPR_FNR}
Under the same assumptions as  Theorem \ref{thm:GAMP_vec}, applying the thresholding function in \eqref{eq:threshold_def} in the final iteration $K > 1$, we have
\begin{align}
    \textup{FPR}
    \stackrel{a.s.}{\rightarrow}
    \mb{P}\big[f_K(G_\beta^K)>\zeta\big] \quad 
    \text{ and } \quad 
    \textup{FNR}
    \stackrel{a.s.}{\rightarrow}
    \mb{P}\big[f_K(\mu_\beta^K+G_\beta^K)\leq\zeta\big].
\end{align}
\end{corollary}
The proof of the corollary is given in Appendix \ref{sec:FPR_FNR_proof}.

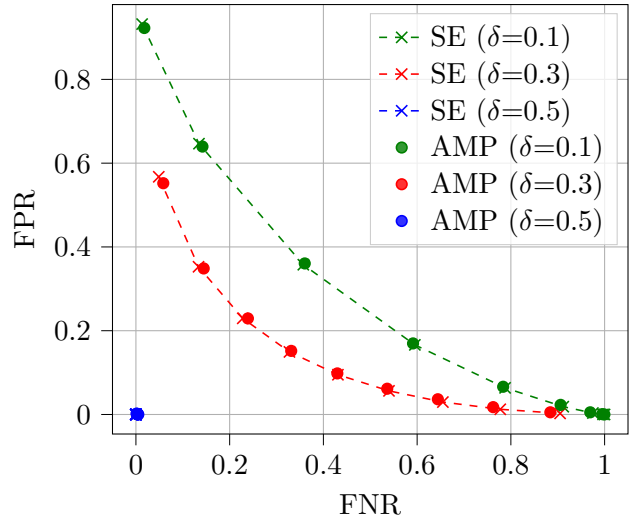
\begin{figure}[t]
    \centering
    % This file was created with tikzplotlib v0.10.1.
\begin{tikzpicture}

\definecolor{darkgray176}{RGB}{176,176,176}
\definecolor{green}{RGB}{0,128,0}
\definecolor{lightgray204}{RGB}{204,204,204}

\begin{axis}[
legend cell align={left},
legend style={fill opacity=0.8, draw opacity=1, text opacity=1, draw=lightgray204},
tick align=outside,
tick pos=left,
x grid style={darkgray176},
xlabel={FNR},
xmajorgrids,
xmin=-0.0499968347260539, xmax=1.04993352924713,
xtick style={color=black},
y grid style={darkgray176},
ylabel={FPR},
ymajorgrids,
ymin=-0.0466010733452594, ymax=0.978622540250447,
ytick style={color=black}
]
\addplot [semithick, green, dashed, mark=x, mark size=3, mark options={solid}]
table {%
0.0131975264217078 0.932021466905188
0.134527474577852 0.646499054107346
0.357149520877481 0.357544565546646
0.595168126024549 0.166187052415598
0.787643436887769 0.0635877051101065
0.911785481055002 0.018593506204027
0.975040981967934 0.00345065069413089
0.996614822811296 0.000275766287357044
0.999936694521077 2.85768173426989e-06
};
\addlegendentry{SE ($\delta$=0.1)}
\addplot [semithick, red, dashed, mark=x, mark size=3, mark options={solid}]
table {%
0.0487673560020332 0.567429611622784
0.133597742047674 0.352625226829185
0.227952192407502 0.229524029627601
0.327573637604002 0.149369158078543
0.431027448567377 0.0948930621654626
0.539958934160892 0.0563663391196174
0.654939939538243 0.0296418666757969
0.777385029026994 0.0124185412172881
0.905023542630889 0.00259498293788019
};
\addlegendentry{SE ($\delta$=0.3)}
\addplot [semithick, blue, dashed, mark=x, mark size=3, mark options={solid}]
table {%
0 0
0 0
0 0
0 0
0 0
0 0
0 0
0 0
0 0
};
\addlegendentry{SE ($\delta$=0.5)}
\addplot [semithick, green, mark=*, mark size=2, mark options={solid}, only marks]
table {%
0.0179063475840989 0.922688249421178
0.141967045850779 0.639607230763226
0.360224214909088 0.360583015465813
0.591255612266219 0.169654917694556
0.783555692712219 0.0662027397435085
0.906129095905726 0.0229534847103938
0.969373943142483 0.00519606825278394
0.99497756652836 0.000797555046422666
0.99953199560301 5.60224089635854e-05
};
\addlegendentry{AMP ($\delta$=0.1)}
\addplot [semithick, red, mark=*, mark size=2, mark options={solid}, only marks]
table {%
0.0583604480961251 0.552113319140856
0.144618289099965 0.348638953147389
0.238922242110587 0.229347646754549
0.331308572711236 0.151833083311463
0.429744984504734 0.0980529587600954
0.536184115246835 0.0614984226191368
0.644239130688889 0.0363597366026153
0.762236252016169 0.0175179207725796
0.884229420884733 0.00522886109927159
};
\addlegendentry{AMP ($\delta$=0.3)}
\addplot [semithick, blue, mark=*, mark size=2, mark options={solid}, only marks]
table {%
0.00130718954248366 0.00204610951008646
0.00222222222222222 0.00146974063400576
0.00248366013071895 0.00112391930835735
0.0026797385620915 0.000864553314121037
0.0030718954248366 0.000749279538904899
0.00333333333333333 0.000634005763688761
0.00366013071895425 0.00037463976945245
0.00411764705882353 0.000288184438040346
0.00503267973856209 0.000115273775216138
};
\addlegendentry{AMP ($\delta$=0.5)}
\end{axis}

\end{tikzpicture}
    \vspace{-2\baselineskip}
    \caption{QGT, FPR vs.~FNR: We use $\alpha=0.5$, $\pi=0.1$, $p=500$, $\zeta\in\{0.1, 0.2, \dots, 0.9\}$.}
    \label{fig:QGT_FPR_v_FNR}
    \vspace{-3pt}
\end{figure}

\subsection{Numerical Simulations} \label{sec:QGT_sims}

We present numerical simulation results for the QGT model in \eqref{eq:QGT} (and equivalently \eqref{eq:noisy_QGT_rescaled}). We take $X_{ij}\stackrel{\iid}{\sim}\text{Bernoulli}(0.5)$ for $i,j\in[n]\times[p]$ and $\beta_j\stackrel{\iid}{\sim}\text{Bernoulli}(\pi)$ for $\pi\in(0,1)$ and $j\in[p]$; the initializer $\hbeta^0\in\mb{R}^{p}$ is chosen randomly according to the same distribution as the signal vector $\beta$, but independent of it. We implement the AMP algorithm in \eqref{eq:memoryless_GAMP} with memoryless Bayes-optimal denoisers $g_k=g_k^*$ and $f_k=f_k^*$, with the thresholding function \eqref{eq:threshold_def} used in the final iteration $K$. Implementation details of the denoisers and their derivatives are given in Appendix \ref{sec:imp_details_GT}. The performance in all the plots is  measured via either the squared overlap between the AMP estimate and the signal (see \eqref{eq:norm_corr_vec}) or the FPR and FNR (see \eqref{eq:FPR_and_FNR}). We set the number of items to be $p=500$ and vary the value of the number of tests $n\leq p$ in our experiments. Each point on the plots is obtained from 100 independent runs, where in each run, the AMP algorithm is executed for 200 iterations. In the squared overlap plots, we also report the average and error bars at 1 standard deviation of the final iteration. For each setting, the state evolution (SE) recursion is computed using \eqref{eq:memoryless_Mu}-\eqref{eq:memoryless_Bk_conv} for a maximum of 200 iterations. 

\paragraph{QGT} For the noiseless setting, Figure \ref{fig:QGT_FPR_v_FNR} shows that the FPR vs.~FNR tradeoff curve improves (i.e., becomes lower) as $\delta$ increases (i.e., more tests are used). For all three curves, the theoretical results from Corollary \ref{cor:FPR_FNR} closely match the empirical result from the algorithm, giving an empirical verification of Theorem \ref{thm:GAMP_vec} and Corollary \ref{cor:FPR_FNR}. 

Our benchmark for the AMP algorithm is a  linear programming (LP) estimator for the QGT problem. The optimization problem for LP is written as
\begin{align*}
    \text{minimize}\quad & \|\beta\|_1 \\
    \text{subject to}\quad & y=X\beta,
    \text{ and }
    0\leq\beta_j\leq1 \text{ for } j\in[p].
\end{align*}
Similar reconstruction algorithms have been widely used in compressed sensing \cite{Fou13} and also studied for Boolean group testing \cite{Ald19}. Similar to the AMP algorithm, we can produce a FPR vs.~FNR trade-off curve for the LP estimator.  After running the linear program to obtain $\beta^{(\text{lp})}$, we set 
$\hbeta_j=\mathds{1}\big\{\beta_j^{\text{(lp})}>\zeta\big\}$ and vary the threshold $\zeta$ to obtain different points on the FPR vs.~FNR curve. Figures \ref{fig:QGT_AMP_v_others_corr_v_delta} and \ref{fig:QGT_AMP_v_others_FPR_FNR}
show the squared overlap and the FPR vs.~FNR tradeoff, respectively, for the signal prior $\bar{\beta} \sim \text{Bernoulli}(\pi)$. The AMP algorithm consistently outperforms LP for for both $\pi=0.1$ and $\pi=0.3$, and also allows for a wider range of FPR vs. FNR tradeoffs by varying the threshold.

Several other QGT algorithms have been proposed (some examples are \cite{Kar19a,Kar19b,Fei20,Hah22}) that were used for the sublinear category regime with no noise. We omit comparisons with them because they do not appear to offer a simple mechanism for controlling the trade-off between FPR and FNR.

\begin{figure}[t]
\centering
\begin{subfigure}[b]{0.45\textwidth}
  \centering
  % This file was created with tikzplotlib v0.10.1.
\begin{tikzpicture}[scale=0.8]

\definecolor{coral}{RGB}{255,127,80}
\definecolor{darkgray176}{RGB}{176,176,176}
\definecolor{lightblue}{RGB}{173,216,230}
\definecolor{lightgray204}{RGB}{204,204,204}

\begin{axis}[
legend cell align={left},
legend style={
  fill opacity=0.8,
  draw opacity=1,
  text opacity=1,
  at={(0.97,0.03)},
  anchor=south east,
  draw=lightgray204
},
tick align=outside,
tick pos=left,
x grid style={darkgray176},
xlabel={\(\displaystyle \delta=n/p\)},
xmajorgrids,
xmin=0.05, xmax=1.15,
xtick style={color=black},
y grid style={darkgray176},
ylabel={Overlap},
ymajorgrids,
ymin=-0.0151397563026971, ymax=1.1183189814821,
ytick style={color=black}
]
\addplot [semithick, blue, dashed]
table {%
0.1 0.222316517406069
0.12 0.255050358681003
0.14 0.291199680864292
0.16 0.331512015150056
0.18 0.377332761292543
0.2 0.431406337293157
0.22 0.501158158291396
0.24 0.999999426456898
0.26 0.999999426436836
0.28 0.999999426416836
0.3 0.999999426396836
0.32 0.999999426376836
0.34 0.999999426356836
0.36 0.999999426336836
0.38 0.999999426316836
0.4 0.999999426296836
0.42 0.999999426276836
0.44 0.999999426256836
0.46 0.999999426236836
0.48 0.999999426216836
0.5 0.999999426196836
0.52 0.999999426176836
0.54 0.999999426156836
0.56 0.999999426136836
0.58 0.999999426116836
0.6 0.999999426096836
0.62 0.999999426076836
0.64 0.999999426056836
0.66 0.999999426036836
0.68 0.999999426016836
0.7 0.999999425996836
0.72 0.999999425976836
0.74 0.999999425956836
0.76 0.999999425936836
0.78 0.999999425916836
0.8 0.999999425896836
0.82 0.999999425876836
0.84 0.999999425856836
0.86 0.999999425836836
0.88 0.999999425816836
0.9 0.999999425796836
0.92 0.999999425776836
0.94 0.999999425756836
0.96 0.999999425736836
0.98 0.999999425716776
1 0.999999425696776
1.02 0.999999425676776
1.04 0.999999425656776
1.06 0.999999425636776
1.08 0.999999425616776
1.1 0.999999425596776
};
\addlegendentry{SE}
\path [draw=lightblue, very thick]
(axis cs:0.1,0.180762150829881)
--(axis cs:0.1,0.251455015762505);

\path [draw=lightblue, very thick]
(axis cs:0.2,0.297938509459133)
--(axis cs:0.2,0.57845205229906);

\path [draw=lightblue, very thick]
(axis cs:0.3,0.931748389163788)
--(axis cs:0.3,1.05140605671163);

\path [draw=lightblue, very thick]
(axis cs:0.4,1)
--(axis cs:0.4,1);

\path [draw=lightblue, very thick]
(axis cs:0.5,1)
--(axis cs:0.5,1);

\path [draw=lightblue, very thick]
(axis cs:0.6,1)
--(axis cs:0.6,1);

\path [draw=lightblue, very thick]
(axis cs:0.7,1)
--(axis cs:0.7,1);

\path [draw=lightblue, very thick]
(axis cs:0.8,1)
--(axis cs:0.8,1);

\path [draw=lightblue, very thick]
(axis cs:0.9,1)
--(axis cs:0.9,1);

\path [draw=lightblue, very thick]
(axis cs:1,1)
--(axis cs:1,1);

\path [draw=lightblue, very thick]
(axis cs:1.1,1)
--(axis cs:1.1,1);

\path [draw=coral, very thick]
(axis cs:0.1,0.0363810954147936)
--(axis cs:0.1,0.0985092511963671);

\path [draw=coral, very thick]
(axis cs:0.2,0.164953734642845)
--(axis cs:0.2,0.48792520713622);

\path [draw=coral, very thick]
(axis cs:0.3,0.889366683611189)
--(axis cs:0.3,1.06679812976461);

\path [draw=coral, very thick]
(axis cs:0.4,1)
--(axis cs:0.4,1);

\path [draw=coral, very thick]
(axis cs:0.5,1)
--(axis cs:0.5,1);

\path [draw=coral, very thick]
(axis cs:0.6,1)
--(axis cs:0.6,1);

\path [draw=coral, very thick]
(axis cs:0.7,1)
--(axis cs:0.7,1);

\path [draw=coral, very thick]
(axis cs:0.8,1)
--(axis cs:0.8,1);

\path [draw=coral, very thick]
(axis cs:0.9,1)
--(axis cs:0.9,1);

\path [draw=coral, very thick]
(axis cs:1,1)
--(axis cs:1,1);

\path [draw=coral, very thick]
(axis cs:1.1,1)
--(axis cs:1.1,1);

\addplot [semithick, blue, mark=asterisk, mark size=3, mark options={solid}, only marks]
table {%
0.1 0.216108583296193
0.2 0.438195280879096
0.3 0.991577222937711
0.4 1
0.5 1
0.6 1
0.7 1
0.8 1
0.9 1
1 1
1.1 1
};
\addlegendentry{AMP}
\addplot [semithick, red, dotted, mark=asterisk, mark size=3, mark options={solid}]
table {%
0.1 0.0674451733055803
0.2 0.326439470889532
0.3 0.978082406687898
0.4 1
0.5 1
0.6 1
0.7 1
0.8 1
0.9 1
1 1
1.1 1
};
\addlegendentry{LP}
\end{axis}

\end{tikzpicture}
  \vspace{-2\baselineskip}
  \caption{Defective probability $\pi=0.1$}
\end{subfigure}
\begin{subfigure}[b]{0.45\textwidth}
  \centering
  % This file was created with tikzplotlib v0.10.1.
\begin{tikzpicture}[scale=0.8]

\definecolor{coral}{RGB}{255,127,80}
\definecolor{darkgray176}{RGB}{176,176,176}
\definecolor{lightblue}{RGB}{173,216,230}
\definecolor{lightgray204}{RGB}{204,204,204}

\begin{axis}[
legend cell align={left},
legend style={
  fill opacity=0.8,
  draw opacity=1,
  text opacity=1,
  at={(0.97,0.03)},
  anchor=south east,
  draw=lightgray204
},
tick align=outside,
tick pos=left,
x grid style={darkgray176},
xlabel={\(\displaystyle \delta=n/p\)},
xmajorgrids,
xmin=0.05, xmax=1.15,
xtick style={color=black},
y grid style={darkgray176},
ylabel={Overlap},
ymajorgrids,
ymin=0.067199144465342, ymax=1.07937330861929,
ytick style={color=black}
]
\addplot [semithick, blue, dashed]
table {%
0.1 0.372741206678121
0.12 0.388016026875203
0.14 0.403580159056105
0.16 0.419464288090143
0.18 0.435706237135899
0.2 0.452352983951659
0.22 0.469462231621668
0.24 0.487106592395222
0.26 0.505379989633053
0.28 0.524404946590207
0.3 0.544348542827835
0.32 0.565446200432558
0.34 0.588052363509111
0.36 0.612741733421028
0.38 0.640570160134388
0.4 0.673949522902346
0.42 0.72245407891497
0.44 0.999999426550178
0.46 0.999999426543511
0.48 0.999999426536844
0.5 0.999999426530178
0.52 0.999999426523511
0.54 0.999999426516844
0.56 0.999999426510178
0.58 0.999999426503511
0.6 0.999999426496844
0.62 0.999999426490178
0.64 0.999999426483511
0.66 0.999999426476844
0.68 0.999999426470178
0.7 0.999999426463511
0.72 0.999999426456898
0.74 0.999999426450169
0.76 0.999999426443503
0.78 0.999999426436836
0.8 0.999999426430169
0.82 0.999999426423503
0.84 0.999999426416836
0.86 0.999999426410169
0.88 0.999999426403503
0.9 0.999999426396836
0.92 0.999999426390169
0.94 0.999999426383503
0.96 0.999999426376836
0.98 0.999999426370169
1 0.999999426363503
1.02 0.999999426356836
1.04 0.999999426350169
1.06 0.999999426343503
1.08 0.999999426336836
1.1 0.999999426330169
};
\addlegendentry{SE}
\path [draw=lightblue, very thick]
(axis cs:0.1,0.347293954982066)
--(axis cs:0.1,0.398242138213869);

\path [draw=lightblue, very thick]
(axis cs:0.2,0.427453356975718)
--(axis cs:0.2,0.478851042299469);

\path [draw=lightblue, very thick]
(axis cs:0.3,0.492828078625986)
--(axis cs:0.3,0.577424834028512);

\path [draw=lightblue, very thick]
(axis cs:0.4,0.583166826491555)
--(axis cs:0.4,0.766594575730129);

\path [draw=lightblue, very thick]
(axis cs:0.5,1)
--(axis cs:0.5,1);

\path [draw=lightblue, very thick]
(axis cs:0.6,1)
--(axis cs:0.6,1);

\path [draw=lightblue, very thick]
(axis cs:0.7,1)
--(axis cs:0.7,1);

\path [draw=lightblue, very thick]
(axis cs:0.8,1)
--(axis cs:0.8,1);

\path [draw=lightblue, very thick]
(axis cs:0.9,1)
--(axis cs:0.9,1);

\path [draw=lightblue, very thick]
(axis cs:1,1)
--(axis cs:1,1);

\path [draw=lightblue, very thick]
(axis cs:1.1,1)
--(axis cs:1.1,1);

\path [draw=coral, very thick]
(axis cs:0.1,0.113207061017794)
--(axis cs:0.1,0.170471054866266);

\path [draw=coral, very thick]
(axis cs:0.2,0.199249762872767)
--(axis cs:0.2,0.271812214615834);

\path [draw=coral, very thick]
(axis cs:0.3,0.303220496176181)
--(axis cs:0.3,0.405161335990225);

\path [draw=coral, very thick]
(axis cs:0.4,0.473653764850401)
--(axis cs:0.4,0.682069061874782);

\path [draw=coral, very thick]
(axis cs:0.5,0.949915599941359)
--(axis cs:0.5,1.03336539206684);

\path [draw=coral, very thick]
(axis cs:0.6,1)
--(axis cs:0.6,1);

\path [draw=coral, very thick]
(axis cs:0.7,1)
--(axis cs:0.7,1);

\path [draw=coral, very thick]
(axis cs:0.8,1)
--(axis cs:0.8,1);

\path [draw=coral, very thick]
(axis cs:0.9,1)
--(axis cs:0.9,1);

\path [draw=coral, very thick]
(axis cs:1,1)
--(axis cs:1,1);

\path [draw=coral, very thick]
(axis cs:1.1,1)
--(axis cs:1.1,1);

\addplot [semithick, blue, mark=asterisk, mark size=3, mark options={solid}, only marks]
table {%
0.1 0.372768046597967
0.2 0.453152199637593
0.3 0.535126456327249
0.4 0.674880701110842
0.5 1
0.6 1
0.7 1
0.8 1
0.9 1
1 1
1.1 1
};
\addlegendentry{AMP}
\addplot [semithick, red, dotted, mark=asterisk, mark size=3, mark options={solid}]
table {%
0.1 0.14183905794203
0.2 0.235530988744301
0.3 0.354190916083203
0.4 0.577861413362592
0.5 0.991640496004099
0.6 1
0.7 1
0.8 1
0.9 1
1 1
1.1 1
};
\addlegendentry{LP}
\end{axis}

\end{tikzpicture}
  \vspace{-2\baselineskip}
  \caption{$\pi=0.3$}
\end{subfigure}
\caption{QGT, Squared overlap vs.~$\delta$: We use $\alpha=0.5$ and $p=500$.}
\label{fig:QGT_AMP_v_others_corr_v_delta}
\vspace{-5pt}
\end{figure}
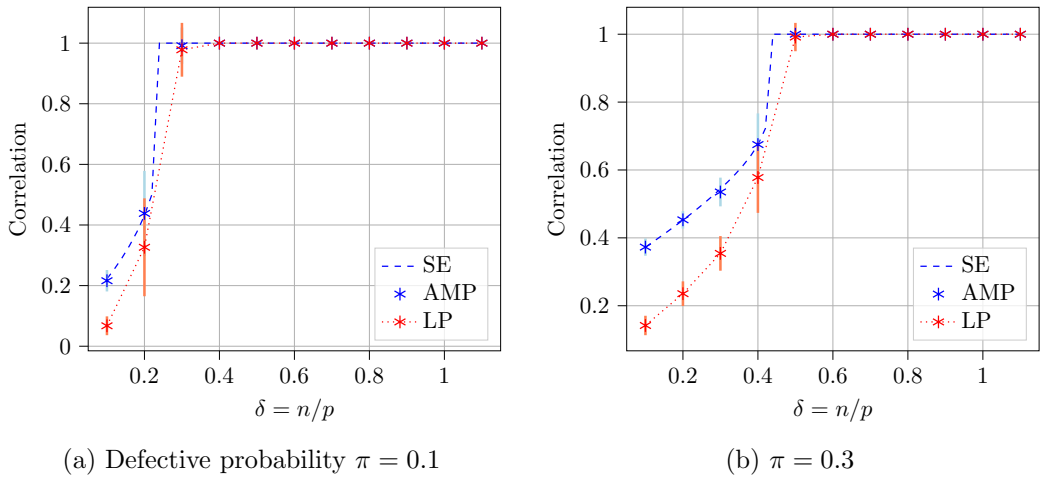
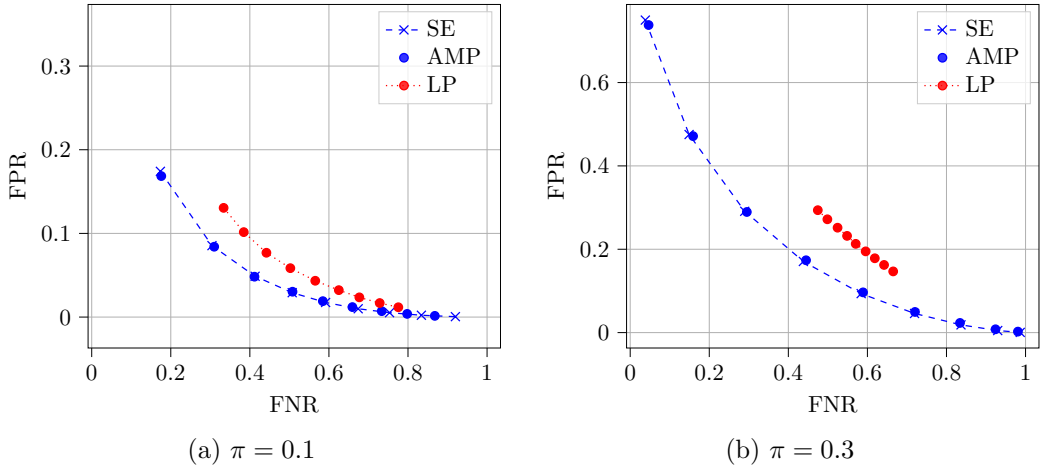
\begin{figure}[t]
\centering
\begin{subfigure}[b]{0.45\textwidth}
  \centering
  % This file was created with tikzplotlib v0.10.1.
\begin{tikzpicture}[scale=0.8]

\definecolor{darkgray176}{RGB}{176,176,176}
\definecolor{lightgray204}{RGB}{204,204,204}

\begin{axis}[
legend cell align={left},
legend style={fill opacity=0.8, draw opacity=1, text opacity=1, draw=lightgray204},
tick align=outside,
tick pos=left,
x grid style={darkgray176},
xlabel={FNR},
xmajorgrids,
xmin=-0.00918514233662787, xmax=1.03397561843212,
xtick style={color=black},
y grid style={darkgray176},
ylabel={FPR},
ymajorgrids,
ymin=-0.0369705906329958, ymax=0.37363540279663,
ytick style={color=black}
]
\addplot [semithick, blue, dashed, mark=x, mark size=3, mark options={solid}]
table {%
0.173719354195735 0.174189138118845
0.304185768387483 0.0853889247048885
0.413334662148694 0.0487828735300002
0.507125772373929 0.0291343396394193
0.59165836157066 0.0174454794033301
0.6740382698824 0.0099415334681991
0.753348614709986 0.00524976102083009
0.833924656169025 0.00218971611497677
0.919882399840542 0.000500188960273881
};
\addlegendentry{SE}
\addplot [semithick, blue, mark=*, mark size=2, mark options={solid}, only marks]
table {%
0.175986946003512 0.168525013407105
0.309811896306081 0.0840704023100851
0.411738167625881 0.048267170908554
0.507907973162762 0.0303061891620011
0.58463415843613 0.0190187685839839
0.65936665316808 0.011940376208161
0.733580910019587 0.00709485536818869
0.797900622161046 0.0038076541549703
0.867862567562312 0.00143933178439779
};
\addlegendentry{AMP}
\addplot [semithick, red, dotted, mark=*, mark size=2, mark options={solid}]
table {%
0.333638732692942 0.130517165466138
0.38460259916842 0.101536689352351
0.441877699090772 0.0768730382303667
0.502435043370988 0.058453330430212
0.56542480845554 0.0433975379323587
0.624993769610426 0.0321911222098418
0.676989324202324 0.02357068878619
0.728499343302476 0.0167671444096947
0.775562382485317 0.0117426470925562
};
\addlegendentry{LP}
\end{axis}

\end{tikzpicture}
  \vspace{-2\baselineskip}
  \caption{$\pi=0.1$}
\end{subfigure}
\begin{subfigure}[b]{0.45\textwidth}
  \centering
  % This file was created with tikzplotlib v0.10.1.
\begin{tikzpicture}[scale=0.8]

\definecolor{darkgray176}{RGB}{176,176,176}
\definecolor{lightgray204}{RGB}{204,204,204}

\begin{axis}[
legend cell align={left},
legend style={fill opacity=0.8, draw opacity=1, text opacity=1, draw=lightgray204},
tick align=outside,
tick pos=left,
x grid style={darkgray176},
xlabel={FNR},
xmajorgrids,
xmin=-0.00918514233662787, xmax=1.03397561843212,
xtick style={color=black},
y grid style={darkgray176},
ylabel={FPR},
ymajorgrids,
ymin=-0.0369705906329958, ymax=0.787363540279663,
ytick style={color=black}
]
\addplot [semithick, blue, dashed, mark=x, mark size=3, mark options={solid}]
table {%
0.0382312558801334 0.74989380705636
0.14898221747933 0.475283860531837
0.289686536946353 0.291101674916583
0.438027134218304 0.170649557565157
0.584216807789868 0.0935456143386933
0.719262901386631 0.0459125370602647
0.83676251076293 0.0185354956171276
0.929631217981323 0.00492563633529224
0.986559220215358 0.000499142590306908
};
\addlegendentry{SE}
\addplot [semithick, blue, mark=*, mark size=2, mark options={solid}, only marks]
table {%
0.0465290258083116 0.738029004839207
0.159065092305415 0.471447581598422
0.295006881475852 0.289434349718369
0.444826922084801 0.173569077781412
0.588974989783354 0.0965460266793533
0.720296612997666 0.0495552812847058
0.834201165005697 0.0231918018336855
0.924071012830774 0.00797619098985807
0.980705154801189 0.00238476629034648
};
\addlegendentry{AMP}
\addplot [semithick, red, dotted, mark=*, mark size=2, mark options={solid}]
table {%
0.474539332158354 0.293752416218673
0.498742273273905 0.271893928487739
0.524518750011839 0.251725673065315
0.548537781149809 0.232079483459782
0.570417355733028 0.212988472888602
0.595557152241838 0.194912893065992
0.618717306074094 0.178362736730264
0.642142574810885 0.162354157696528
0.665198190421895 0.146838887497477
};
\addlegendentry{LP}
\end{axis}

\end{tikzpicture}
  \vspace{-2\baselineskip}
  \caption{$\pi=0.3$}
\end{subfigure}
\caption{QGT, FPR vs.~FNR: $\alpha=0.5$, $\delta=0.2$, $p=500$, and threshold $\zeta\in\{0, 0.1, \dots, 1.0\}$ for both AMP and LP.}
\label{fig:QGT_AMP_v_others_FPR_FNR}
\end{figure}

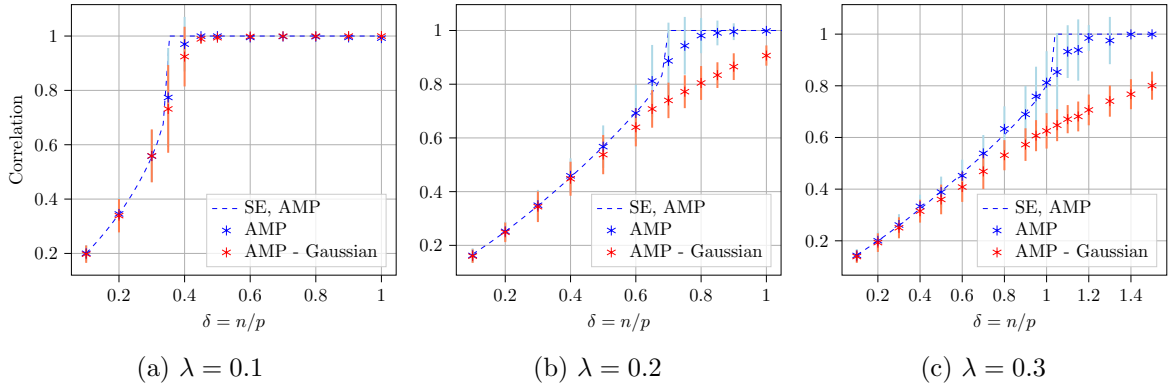
\begin{figure}[t]
\centering
\begin{subfigure}[b]{0.34\textwidth}
  \centering
  % This file was created with tikzplotlib v0.10.1.
\begin{tikzpicture}[scale=0.63]

\definecolor{coral}{RGB}{255,127,80}
\definecolor{darkgray176}{RGB}{176,176,176}
\definecolor{lightblue}{RGB}{173,216,230}
\definecolor{lightgray204}{RGB}{204,204,204}

\begin{axis}[
legend cell align={left},
legend style={
  fill opacity=0.8,
  draw opacity=1,
  text opacity=1,
  at={(0.97,0.03)},
  anchor= south east,
  draw=lightgray204
},
tick align=outside,
tick pos=left,
x grid style={darkgray176},
xlabel={\(\displaystyle \delta=n/p\)},
xmajorgrids,
xmin=0.055, xmax=1.045,
xtick style={color=black},
y grid style={darkgray176},
ylabel={Overlap},
ymajorgrids,
ymin=0.120024852454111, ymax=1.11561303024497,
ytick style={color=black}
]
\addplot [semithick, blue, dashed]
table {%
0.1 0.200767915328118
0.119565217391304 0.225493019707963
0.139130434782609 0.251793330516149
0.158695652173913 0.279823234714205
0.178260869565217 0.309610593772194
0.197826086956522 0.341392926159932
0.217391304347826 0.375373022166325
0.23695652173913 0.411876747433244
0.256521739130435 0.45152577876986
0.276086956521739 0.495076653948585
0.295652173913043 0.544052779814818
0.315217391304348 0.601492438400297
0.334782608695652 0.675578707859378
0.354347826086957 0.999999426342509
0.373913043478261 0.999999426322944
0.393478260869565 0.999999999592823
0.41304347826087 0.999999999586956
0.432608695652174 0.999999999436661
0.452173913043478 0.999999999547825
0.471739130434783 0.999999999466328
0.491304347826087 0.999999426205553
0.510869565217391 0.999999999489131
0.530434782608696 0.999999426166422
0.55 0.999999998406189
0.569565217391304 0.999999999430435
0.589130434782609 0.999999426107727
0.608695652173913 0.999999426088161
0.628260869565217 0.99999999776451
0.647826086956522 0.999999999352157
0.667391304347826 0.999999999332607
0.68695652173913 0.9999994260099
0.706521739130435 0.999999425990335
0.726086956521739 0.99999942597077
0.745652173913044 0.999999425951205
0.765217391304348 0.99999942593164
0.784782608695652 0.999999998986539
0.804347826086957 0.999999999194875
0.823913043478261 0.999999999176087
0.843478260869565 0.999999999156522
0.86304347826087 0.99999999913694
0.882608695652174 0.999999425814249
0.902173913043478 0.999999425794686
0.921739130434783 0.999999425775121
0.941304347826087 0.999999425755556
0.960869565217391 0.999999425735991
0.980434782608696 0.999999425716426
1 0.99999942569686
};
\addlegendentry{SE, AMP}
\path [draw=lightblue, very thick]
(axis cs:0.1,0.170861346040386)
--(axis cs:0.1,0.23010984146373);

\path [draw=lightblue, very thick]
(axis cs:0.2,0.285372252267378)
--(axis cs:0.2,0.405191303927538);

\path [draw=lightblue, very thick]
(axis cs:0.3,0.462862957429346)
--(axis cs:0.3,0.653588888188715);

\path [draw=lightblue, very thick]
(axis cs:0.35,0.591602341290123)
--(axis cs:0.35,0.956085430460179);

\path [draw=lightblue, very thick]
(axis cs:0.4,0.869844298303485)
--(axis cs:0.4,1.07035902216357);

\path [draw=lightblue, very thick]
(axis cs:0.45,0.997942575637495)
--(axis cs:0.45,1.00132445674746);

\path [draw=lightblue, very thick]
(axis cs:0.5,0.998946164397655)
--(axis cs:0.5,1.00053138890676);

\path [draw=lightblue, very thick]
(axis cs:0.6,0.994397577933382)
--(axis cs:0.6,1.00522585514881);

\path [draw=lightblue, very thick]
(axis cs:0.7,0.993765536885042)
--(axis cs:0.7,1.00271690910117);

\path [draw=lightblue, very thick]
(axis cs:0.8,0.990131676882765)
--(axis cs:0.8,1.00477943696252);

\path [draw=lightblue, very thick]
(axis cs:0.9,0.984880939069789)
--(axis cs:0.9,1.00491673125096);

\path [draw=lightblue, very thick]
(axis cs:1,0.98138897939016)
--(axis cs:1,1.00338885449069);

\path [draw=coral, very thick]
(axis cs:0.1,0.165278860535513)
--(axis cs:0.1,0.229358887932009);

\path [draw=coral, very thick]
(axis cs:0.2,0.277463860425126)
--(axis cs:0.2,0.400576752159132);

\path [draw=coral, very thick]
(axis cs:0.3,0.461623756108694)
--(axis cs:0.3,0.656519283708987);

\path [draw=coral, very thick]
(axis cs:0.35,0.570449630540761)
--(axis cs:0.35,0.893093965925668);

\path [draw=coral, very thick]
(axis cs:0.4,0.814510960311315)
--(axis cs:0.4,1.03424284459);

\path [draw=coral, very thick]
(axis cs:0.45,0.974659480995529)
--(axis cs:0.45,1.00425334639664);

\path [draw=coral, very thick]
(axis cs:0.5,0.985296338888141)
--(axis cs:0.5,1.00232652653411);

\path [draw=coral, very thick]
(axis cs:0.6,0.994392271791335)
--(axis cs:0.6,1.0025965074404);

\path [draw=coral, very thick]
(axis cs:0.7,0.996660516255559)
--(axis cs:0.7,1.00208308745558);

\path [draw=coral, very thick]
(axis cs:0.8,0.99886177357434)
--(axis cs:0.8,1.00068436605724);

\path [draw=coral, very thick]
(axis cs:0.9,0.999992040493579)
--(axis cs:0.9,1.0000052472494);

\path [draw=coral, very thick]
(axis cs:1,0.999915188053119)
--(axis cs:1,1.00006925501257);

\addplot [semithick, blue, mark=asterisk, mark size=3, mark options={solid}, only marks]
table {%
0.1 0.200485593752058
0.2 0.345281778097458
0.3 0.558225922809031
0.35 0.773843885875151
0.4 0.970101660233526
0.45 0.999633516192477
0.5 0.999738776652207
0.6 0.994811716541097
0.7 0.998241222993104
0.8 0.997455556922645
0.9 0.994898835160376
1 0.992388916940423
};
\addlegendentry{AMP}
\addplot [semithick, red, mark=asterisk, mark size=3, mark options={solid}, only marks]
table {%
0.1 0.197318874233761
0.2 0.339020306292129
0.3 0.559071519908841
0.35 0.731771798233214
0.4 0.924376902450659
0.45 0.989456413696085
0.5 0.993811432711123
0.6 0.998494389615866
0.7 0.999371801855571
0.8 0.99977306981579
0.9 0.999998643871488
1 0.999992221532845
};
\addlegendentry{AMP - Gaussian}
\end{axis}

\end{tikzpicture}
  \vspace{-2\baselineskip}
  \caption{$\lambda=0.1$}
\end{subfigure}
\begin{subfigure}[b]{0.32\textwidth}
  \centering
  % This file was created with tikzplotlib v0.10.1.
\begin{tikzpicture}[scale=0.63]

\definecolor{coral}{RGB}{255,127,80}
\definecolor{darkgray176}{RGB}{176,176,176}
\definecolor{lightblue}{RGB}{173,216,230}
\definecolor{lightgray204}{RGB}{204,204,204}

\begin{axis}[
legend cell align={left},
legend style={
  fill opacity=0.8,
  draw opacity=1,
  text opacity=1,
  at={(0.97,0.03)},
  anchor=south east,
  draw=lightgray204
},
tick align=outside,
tick pos=left,
x grid style={darkgray176},
xlabel={\(\displaystyle \delta=n/p\)},
xmajorgrids,
xmin=0.05, xmax=1.045,
xtick style={color=black},
y grid style={darkgray176},
ylabel={},
ymajorgrids,
ymin=0.0890194470020814, ymax=1.09631508166769,
ytick style={color=black}
]
\addplot [semithick, blue, dashed]
table {%
0.1 0.167388765053642
0.12 0.183081934404304
0.14 0.19934362570398
0.16 0.216272349266741
0.18 0.233636733197761
0.2 0.251506552614294
0.22 0.269968236689249
0.24 0.288851655877768
0.26 0.307972778079043
0.28 0.327784936148815
0.3 0.347685362855765
0.32 0.368426738380697
0.34 0.389141310772149
0.36 0.410154887867958
0.38 0.431667994842096
0.4 0.453824777431535
0.42 0.47573439103069
0.44 0.497994557276638
0.46 0.521324906044718
0.48 0.544249501082504
0.5 0.569253210034366
0.52 0.592148338250537
0.54 0.616710072470551
0.56 0.642153606158724
0.58 0.668335857285204
0.6 0.694360547213752
0.62 0.722002727416685
0.64 0.752224917605368
0.66 0.786689763913033
0.68 0.83326113173586
0.7 0.999999425996857
0.72 0.999999425976857
0.74 0.999999999129434
0.76 0.999999996136705
0.78 0.999999425916857
0.8 0.999999425896857
0.82 0.999999425876857
0.84 0.99999999916
0.86 0.999999425836857
0.88 0.99999999912
0.9 0.999999425796857
0.92 0.999999999073849
0.94 0.99999999906
0.96 0.999999425736857
0.98 0.999999425716857
1 0.99999999899287
1.02 0.99999999898
1.04 0.999999425656857
1.06 0.999999425636857
1.08 0.999999425616857
1.1 0.999999425596857
};
\addlegendentry{SE, AMP}
\path [draw=lightblue, very thick]
(axis cs:0.1,0.13607509203852)
--(axis cs:0.1,0.188951796159539);

\path [draw=lightblue, very thick]
(axis cs:0.2,0.216009343158277)
--(axis cs:0.2,0.287942319024618);

\path [draw=lightblue, very thick]
(axis cs:0.3,0.292278018768992)
--(axis cs:0.3,0.406586513950157);

\path [draw=lightblue, very thick]
(axis cs:0.4,0.393890790523901)
--(axis cs:0.4,0.523962403025165);

\path [draw=lightblue, very thick]
(axis cs:0.5,0.488385430750293)
--(axis cs:0.5,0.64658229097886);

\path [draw=lightblue, very thick]
(axis cs:0.6,0.587277578127393)
--(axis cs:0.6,0.797719028160318);

\path [draw=lightblue, very thick]
(axis cs:0.65,0.677271283074856)
--(axis cs:0.65,0.946135579358198);

\path [draw=lightblue, very thick]
(axis cs:0.7,0.744592004554677)
--(axis cs:0.7,1.02916255188944);

\path [draw=lightblue, very thick]
(axis cs:0.75,0.835945209536782)
--(axis cs:0.75,1.05052891645561);

\path [draw=lightblue, very thick]
(axis cs:0.8,0.915486099067018)
--(axis cs:0.8,1.04640737046978);

\path [draw=lightblue, very thick]
(axis cs:0.85,0.944735946651053)
--(axis cs:0.85,1.03691262449034);

\path [draw=lightblue, very thick]
(axis cs:0.9,0.964324190407467)
--(axis cs:0.9,1.02679718900058);

\path [draw=lightblue, very thick]
(axis cs:1,0.992641376828041)
--(axis cs:1,1.00437901187591);

\path [draw=coral, very thick]
(axis cs:0.1,0.134805612214154)
--(axis cs:0.1,0.184939153579853);

\path [draw=coral, very thick]
(axis cs:0.2,0.212163573427919)
--(axis cs:0.2,0.285067088628657);

\path [draw=coral, very thick]
(axis cs:0.3,0.286690813633544)
--(axis cs:0.3,0.401105700865256);

\path [draw=coral, very thick]
(axis cs:0.4,0.384669713172706)
--(axis cs:0.4,0.511386601252076);

\path [draw=coral, very thick]
(axis cs:0.5,0.464845941559608)
--(axis cs:0.5,0.610729469294443);

\path [draw=coral, very thick]
(axis cs:0.6,0.568574808004994)
--(axis cs:0.6,0.710897071952511);

\path [draw=coral, very thick]
(axis cs:0.65,0.638083746016916)
--(axis cs:0.65,0.777736111919288);

\path [draw=coral, very thick]
(axis cs:0.7,0.673615612975224)
--(axis cs:0.7,0.805775908054411);

\path [draw=coral, very thick]
(axis cs:0.75,0.71086374247939)
--(axis cs:0.75,0.833408002939852);

\path [draw=coral, very thick]
(axis cs:0.8,0.74152779405874)
--(axis cs:0.8,0.867613832177424);

\path [draw=coral, very thick]
(axis cs:0.85,0.785985705122386)
--(axis cs:0.85,0.881711929968839);

\path [draw=coral, very thick]
(axis cs:0.9,0.815611042972777)
--(axis cs:0.9,0.915141729002638);

\path [draw=coral, very thick]
(axis cs:1,0.86904901387252)
--(axis cs:1,0.944392806193673);

\addplot [semithick, blue, mark=asterisk, mark size=3, mark options={solid}, only marks]
table {%
0.1 0.16251344409903
0.2 0.251975831091448
0.3 0.349432266359575
0.4 0.458926596774533
0.5 0.567483860864576
0.6 0.692498303143856
0.65 0.811703431216527
0.7 0.886877278222056
0.75 0.943237062996198
0.8 0.980946734768399
0.85 0.990824285570698
0.9 0.995560689704024
1 0.998510194351976
};
\addlegendentry{AMP}
\addplot [semithick, red, mark=asterisk, mark size=3, mark options={solid}, only marks]
table {%
0.1 0.159872382897004
0.2 0.248615331028288
0.3 0.3438982572494
0.4 0.448028157212391
0.5 0.537787705427026
0.6 0.639735939978753
0.65 0.707909928968102
0.7 0.739695760514818
0.75 0.772135872709621
0.8 0.804570813118082
0.85 0.833848817545613
0.9 0.865376385987707
1 0.906720910033096
};
\addlegendentry{AMP - Gaussian}
\end{axis}

\end{tikzpicture}
  \vspace{-2\baselineskip}
  \caption{$\lambda=0.2$}
\end{subfigure}
\begin{subfigure}[b]{0.32\textwidth}
  \centering
  % This file was created with tikzplotlib v0.10.1.
\begin{tikzpicture}[scale=0.63]

\definecolor{coral}{RGB}{255,127,80}
\definecolor{darkgray176}{RGB}{176,176,176}
\definecolor{lightblue}{RGB}{173,216,230}
\definecolor{lightgray204}{RGB}{204,204,204}

\begin{axis}[
legend cell align={left},
legend style={
  fill opacity=0.8,
  draw opacity=1,
  text opacity=1,
  at={(0.97,0.03)},
  anchor=south east,
  draw=lightgray204
},
tick align=outside,
tick pos=left,
x grid style={darkgray176},
xlabel={\(\displaystyle \delta=n/p\)},
xmajorgrids,
xmin=0.03, xmax=1.57,
xtick style={color=black},
y grid style={darkgray176},
ylabel={},
ymajorgrids,
ymin=0.0670165888603375, ymax=1.11417284571177,
ytick style={color=black}
]
\addplot [semithick, blue, dashed]
table {%
0.1 0.145471533781863
0.12 0.155804032429363
0.14 0.166303078058109
0.16 0.177265530663628
0.18 0.188163874161815
0.2 0.199676126508217
0.22 0.211428079863602
0.24 0.222884145916128
0.26 0.235048672676634
0.28 0.247099663506174
0.3 0.259486951357391
0.32 0.272038101613317
0.34 0.28445587985113
0.36 0.297707418010628
0.38 0.310911818463266
0.4 0.323186144755895
0.42 0.336553258478236
0.44 0.349891666472557
0.46 0.363744662310949
0.48 0.377894060346176
0.5 0.390836738568797
0.52 0.405274977227472
0.54 0.418952087125595
0.56 0.433126268922599
0.58 0.446547661514091
0.6 0.46186720870675
0.62 0.47687411944744
0.64 0.490752741149837
0.66 0.504795751032941
0.68 0.518952614920568
0.7 0.534148517296541
0.72 0.54890692509121
0.74 0.565168406963057
0.76 0.579507737442374
0.78 0.596191176702609
0.8 0.609961573746034
0.82 0.627926626129858
0.84 0.643645147446782
0.86 0.660408526165675
0.88 0.678959049511933
0.9 0.696855729413485
0.92 0.714412985135643
0.94 0.73301564217518
0.96 0.752000629454473
0.98 0.776585459234443
1 0.801840195474521
1.02 0.834034885219759
1.04 0.999999425656857
1.06 0.999999425636857
1.08 0.999999425616857
1.1 0.999999425596857
1.12 0.999999425576857
1.14 0.99999999869036
1.16 0.99999999884
1.18 0.999999997948896
1.2 0.999999425496857
1.22 0.999999425476857
1.24 0.999999998456521
1.26 0.999999425436857
1.28 0.999999425416857
1.3 0.999999998662324
1.32 0.999999998673601
1.34 0.999999425356857
1.36 0.999999996604433
1.38 0.99999999862
1.4 0.999999998599983
1.42 0.999999425276857
1.44 0.999999425256857
1.46 0.999999985994678
1.48 0.999999998292652
1.5 0.99999999849999
};
\addlegendentry{SE, AMP}
\path [draw=lightblue, very thick]
(axis cs:0.1,0.119149938637834)
--(axis cs:0.1,0.167795436678741);

\path [draw=lightblue, very thick]
(axis cs:0.2,0.16509528313408)
--(axis cs:0.2,0.235214345950936);

\path [draw=lightblue, very thick]
(axis cs:0.3,0.219237432939064)
--(axis cs:0.3,0.30310921170873);

\path [draw=lightblue, very thick]
(axis cs:0.4,0.2885976643914)
--(axis cs:0.4,0.378624517974466);

\path [draw=lightblue, very thick]
(axis cs:0.5,0.328493918249691)
--(axis cs:0.5,0.448449595641861);

\path [draw=lightblue, very thick]
(axis cs:0.6,0.389859091028597)
--(axis cs:0.6,0.514216146437628);

\path [draw=lightblue, very thick]
(axis cs:0.7,0.467388850499524)
--(axis cs:0.7,0.609487837767526);

\path [draw=lightblue, very thick]
(axis cs:0.8,0.545963906580149)
--(axis cs:0.8,0.720508703543695);

\path [draw=lightblue, very thick]
(axis cs:0.9,0.579548883648805)
--(axis cs:0.9,0.799681236892151);

\path [draw=lightblue, very thick]
(axis cs:0.95,0.644366909644345)
--(axis cs:0.95,0.873838424482225);

\path [draw=lightblue, very thick]
(axis cs:1,0.690224803325846)
--(axis cs:1,0.933852383186672);

\path [draw=lightblue, very thick]
(axis cs:1.05,0.701929714298818)
--(axis cs:1.05,1.00451119705855);

\path [draw=lightblue, very thick]
(axis cs:1.1,0.830177959743951)
--(axis cs:1.1,1.03441483323743);

\path [draw=lightblue, very thick]
(axis cs:1.15,0.81980703839602)
--(axis cs:1.15,1.05696941588725);

\path [draw=lightblue, very thick]
(axis cs:1.2,0.934688511949112)
--(axis cs:1.2,1.03506440493505);

\path [draw=lightblue, very thick]
(axis cs:1.3,0.883672835798226)
--(axis cs:1.3,1.0665748340367);

\path [draw=lightblue, very thick]
(axis cs:1.4,0.99617943256614)
--(axis cs:1.4,1.00204692207522);

\path [draw=lightblue, very thick]
(axis cs:1.5,0.997319774482534)
--(axis cs:1.5,1.00142700099241);

\path [draw=coral, very thick]
(axis cs:0.1,0.114614600535402)
--(axis cs:0.1,0.163608009798318);

\path [draw=coral, very thick]
(axis cs:0.2,0.156927017598367)
--(axis cs:0.2,0.229137413122769);

\path [draw=coral, very thick]
(axis cs:0.3,0.209774564367658)
--(axis cs:0.3,0.291180116837689);

\path [draw=coral, very thick]
(axis cs:0.4,0.27049876892142)
--(axis cs:0.4,0.35906821883439);

\path [draw=coral, very thick]
(axis cs:0.5,0.302696795715861)
--(axis cs:0.5,0.417410409482436);

\path [draw=coral, very thick]
(axis cs:0.6,0.350607645529532)
--(axis cs:0.6,0.465803382659326);

\path [draw=coral, very thick]
(axis cs:0.7,0.40031563347332)
--(axis cs:0.7,0.536864716113957);

\path [draw=coral, very thick]
(axis cs:0.8,0.472991214864764)
--(axis cs:0.8,0.58974480262774);

\path [draw=coral, very thick]
(axis cs:0.9,0.509199506669326)
--(axis cs:0.9,0.63558035570497);

\path [draw=coral, very thick]
(axis cs:0.95,0.546839673710588)
--(axis cs:0.95,0.667083580230721);

\path [draw=coral, very thick]
(axis cs:1,0.557047714702641)
--(axis cs:1,0.693904944494445);

\path [draw=coral, very thick]
(axis cs:1.05,0.5856440822897)
--(axis cs:1.05,0.709171062617701);

\path [draw=coral, very thick]
(axis cs:1.1,0.61626532857334)
--(axis cs:1.1,0.725334517121934);

\path [draw=coral, very thick]
(axis cs:1.15,0.62354849874788)
--(axis cs:1.15,0.739123490466594);

\path [draw=coral, very thick]
(axis cs:1.2,0.6470660646831)
--(axis cs:1.2,0.766253519353365);

\path [draw=coral, very thick]
(axis cs:1.3,0.680412658195954)
--(axis cs:1.3,0.800905384405887);

\path [draw=coral, very thick]
(axis cs:1.4,0.709510806371832)
--(axis cs:1.4,0.825477486578648);

\path [draw=coral, very thick]
(axis cs:1.5,0.746108546863524)
--(axis cs:1.5,0.85495297694683);

\addplot [semithick, blue, mark=asterisk, mark size=3, mark options={solid}, only marks]
table {%
0.1 0.143472687658288
0.2 0.200154814542508
0.3 0.261173322323897
0.4 0.333611091182933
0.5 0.388471756945776
0.6 0.452037618733113
0.7 0.538438344133525
0.8 0.633236305061922
0.9 0.689615060270478
0.95 0.759102667063285
1 0.812038593256259
1.05 0.853220455678686
1.1 0.932296396490689
1.15 0.938388227141634
1.2 0.984876458442079
1.3 0.975123834917463
1.4 0.99911317732068
1.5 0.99937338773747
};
\addlegendentry{AMP}
\addplot [semithick, red, mark=asterisk, mark size=3, mark options={solid}, only marks]
table {%
0.1 0.13911130516686
0.2 0.193032215360568
0.3 0.250477340602673
0.4 0.314783493877905
0.5 0.360053602599149
0.6 0.408205514094429
0.7 0.468590174793638
0.8 0.531368008746252
0.9 0.572389931187148
0.95 0.606961626970654
1 0.625476329598543
1.05 0.6474075724537
1.1 0.670799922847637
1.15 0.681335994607237
1.2 0.706659792018232
1.3 0.740659021300921
1.4 0.76749414647524
1.5 0.800530761905177
};
\addlegendentry{AMP - Gaussian}
\end{axis}

\end{tikzpicture}
  \vspace{-2\baselineskip}
  \caption{$\lambda=0.3$}
\end{subfigure}
\caption{Noisy QGT with noise $\Psi_i\sim\text{Uniform}[-\lambda\sqrt{p},\lambda\sqrt{p}]$. Plots show squared overlap vs.~$\delta$: $\alpha=0.5$, $\pi=0.1$, and $p=500$.}
\label{fig:noisy QGT_corr_v_delta_and_FPR_v_FNR}
\end{figure}

\begin{figure}[t]
    \centering
    % This file was created with tikzplotlib v0.10.1.
\begin{tikzpicture}[scale=0.8]

\pgfkeys{
        /pgfplots/every y tick label/.append style  =
            { 
              /pgf/number format/.cd,
               precision = 2, 
               fixed
            }
    }

\definecolor{darkgray176}{RGB}{176,176,176}
\definecolor{green}{RGB}{0,128,0}
\definecolor{lightgray204}{RGB}{204,204,204}

\begin{axis}[
legend cell align={left},
legend style={fill opacity=0.8, draw opacity=1, text opacity=1, draw=lightgray204},
tick align=outside,
tick pos=left,
x grid style={darkgray176},
xlabel={FNR},
xmajorgrids,
xmin=0.0842561562793741, xmax=1.0408284688323,
xtick style={color=black},
y grid style={darkgray176},
ylabel={FPR},
ymajorgrids,
ymin=-0.0130195027346926, ymax=0.273409557428545,
ytick style={color=black}
]
\addplot [semithick, blue, dashed, mark=x, mark size=3, mark options={solid}]
table {%
0.127736715940871 0.12742587900609
0.21585097882541 0.0672267413432902
0.290751098681582 0.0415577632573232
0.360817019576508 0.027004711739343
0.430263683579704 0.0177345868338001
0.499540551338394 0.0113715162021603
0.576667998401918 0.00668755834111215
0.666640031961646 0.00341712228297106
0.785637235317619 0.00112681690892119
};
\addlegendentry{SE ($\lambda$=0.1)}
\addplot [semithick, red, dashed, mark=x, mark size=3, mark options={solid}]
table {%
0.211701881593371 0.211984058860957
0.38424328972113 0.0956668217548356
0.517181733157678 0.050104212587015
0.625622253286 0.0273373434173956
0.715639047839335 0.0148296121426263
0.796031658235868 0.00761967791259468
0.864326156259742 0.00345450568144594
0.924263604091032 0.00119702896965565
0.974234942727556 0.000199874967103912
};
\addlegendentry{SE ($\lambda$=0.2)}
\addplot [semithick, green, dashed, mark=x, mark size=3, mark options={solid}]
table {%
0.261522724537687 0.260390054693852
0.497793734205143 0.101189180303104
0.661779453648361 0.0449846937188709
0.775843395242489 0.0207167960332517
0.85919009988367 0.00916821802897752
0.917977054835733 0.00381101963380444
0.959284367603995 0.0012840392587228
0.984736652091941 0.000317677532867406
0.997312367122628 2.44367332974928e-05
};
\addlegendentry{SE ($\lambda$=0.3)}
\addplot [semithick, blue, mark=*, mark size=2, mark options={solid}, only marks]
table {%
0.142328390009646 0.127999238080547
0.225215974371336 0.0675562390185389
0.300921590106175 0.0425271259428626
0.365924101775772 0.0277399778085263
0.440477589647647 0.017715658472866
0.509812855657622 0.011892630282276
0.583220462640269 0.00687477205116529
0.673996715724133 0.00381516225573196
0.784701879489701 0.00145297328051485
};
\addlegendentry{AMP ($\lambda$=0.1)}
\addplot [semithick, red, mark=*, mark size=2, mark options={solid}, only marks]
table {%
0.215259783635366 0.216609036433892
0.389253774336801 0.0978589748714494
0.528982551087078 0.0533925738889482
0.638926621928401 0.0293428259405125
0.732825469963421 0.0155622820825804
0.808164971046077 0.0078644859548882
0.876875126548537 0.00378100029427542
0.92749588232487 0.00112339394489579
0.972029326421272 8.81653992093979e-05
};
\addlegendentry{AMP ($\lambda$=0.2)}
\addplot [semithick, green, mark=*, mark size=2, mark options={solid}, only marks]
table {%
0.277944559383438 0.258926310646385
0.515749090255591 0.101312685459262
0.67266997793248 0.043619471362081
0.776812788409188 0.0192895204232817
0.855100603557718 0.00833900575544106
0.915263044272147 0.00359220111524373
0.956125177178782 0.00133927144477003
0.982395999189813 0.000291616856459621
0.997347909170808 0
};
\addlegendentry{AMP ($\lambda$=0.3)}
\end{axis}

\end{tikzpicture}
    \vspace{-2\baselineskip}
    \caption{Noisy QGT with noise $\Psi_i\sim\text{Uniform}[-\lambda\sqrt{p},\lambda\sqrt{p}]$, FPR vs.~FNR: $\alpha=0.5$, $\pi=0.1$, $\delta=0.3$, $p=500$, $\zeta\in\{0.1, 0.2, \dots, 0.9\}$.}
    \label{fig:noisy_QGT_FPR_v_FNR}
\end{figure}
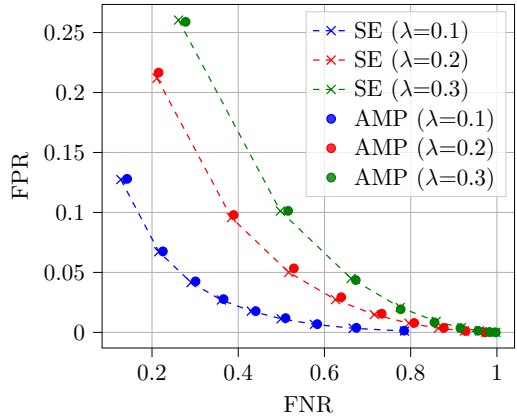

\paragraph{Noisy QGT} To highlight that the AMP algorithm works under different noise distributions, we consider QGT with uniformly distributed noise (recall that Gaussian noise was used for the pooled data simulations). Specifically, $\Psi_i\stackrel{\iid}{\sim}  \text{Uniform}[-\lambda\sqrt{p},\lambda\sqrt{p}]$ for some constant $\lambda$. After rescaling  as in \eqref{eq:noisy_QGT_rescaled},  we get $\tPsi_i={\tPsi_i}/{\sqrt{n\alpha(1-\alpha)}}$ whose empirical distribution converges to $\Bar{\Psi} \sim \text{Uniform}[-\tilde{\lambda},\tilde{\lambda}]$, where $\tilde{\lambda}=\lambda/\sqrt{\delta\alpha(1-\alpha)}$.
As expected, Figure \ref{fig:noisy QGT_corr_v_delta_and_FPR_v_FNR} and Figure \ref{fig:noisy_QGT_FPR_v_FNR} show that the performance of the AMP algorithm improves as the noise level $\lambda$ decreases.  The red curves in Figure \ref{fig:noisy QGT_corr_v_delta_and_FPR_v_FNR} show the performance AMP with a denoiser $g_k$ that is optimal for Gaussian noise of the same variance; the blue curves correspond to $g_k$ that is Bayes-optimal for the uniform noise distribution.   The comparison highlights the usefulness of tailoring the denoiser to the noise distribution, with the difference between the curves more pronounced as the noise level $\lambda$ increases.

We benchmark AMP against  a convex programming estimator that is a variant of basis pursuit denoising (BPDN) \cite{Che01}.  The convex program for BPDN can be written as
\begin{align*}
    \text{minimize}\quad & \|\beta\|_1 \\
    \text{subject to}\quad & \|y-X\beta\|_\infty\leq \lambda\sqrt{p},
    \text{ and }
    0\leq\beta_j\leq1 \text{ for } j\in[p],
\end{align*}
where the $\ell_\infty$-constraint is tailored to uniform noise. Similar to the AMP algorithm, we can obtain a FPR vs.~FNR trade-off curve by using a thresholding function, with a varying threshold.  Figure \ref{fig:noisy QGT_AMP_v_others_corr_v_delta} shows that the AMP algorithm outperforms BPDN for all values of $\delta\leq1$. Figure \ref{fig:noisy QGT_AMP_v_others_FPR_FNR} backs this up, showing that the FPR vs.~FNR curve for the AMP algorithm generally lies below that of BPDN.

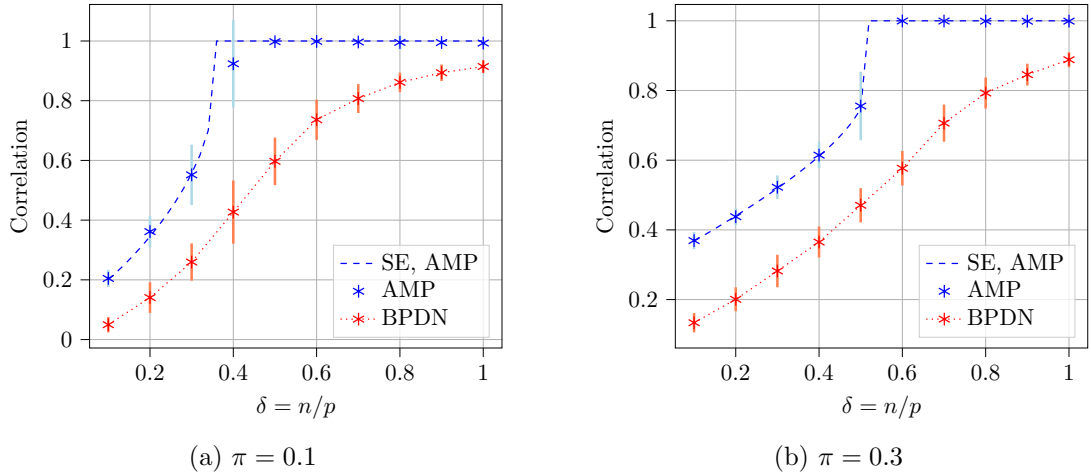
\begin{figure}[t]
\centering
\begin{subfigure}[b]{0.45\textwidth}
  \centering
  % This file was created with tikzplotlib v0.10.1.
\begin{tikzpicture}[scale=0.8]

\definecolor{coral}{RGB}{255,127,80}
\definecolor{darkgray176}{RGB}{176,176,176}
\definecolor{lightblue}{RGB}{173,216,230}
\definecolor{lightgray204}{RGB}{204,204,204}

\begin{axis}[
legend cell align={left},
legend style={
  fill opacity=0.8,
  draw opacity=1,
  text opacity=1,
  at={(0.97,0.03)},
  anchor=south east,
  draw=lightgray204
},
tick align=outside,
tick pos=left,
x grid style={darkgray176},
xlabel={\(\displaystyle \delta=n/p\)},
xmajorgrids,
xmin=0.055, xmax=1.045,
xtick style={color=black},
y grid style={darkgray176},
ylabel={Overlap},
ymajorgrids,
ymin=-0.028271425986096, ymax=1.12238741783134,
ytick style={color=black}
]
\addplot [semithick, blue, dashed]
table {%
0.1 0.20076568656703
0.12 0.22605214677253
0.14 0.253027484204135
0.16 0.28175767932255
0.18 0.31237221721954
0.2 0.345064753710109
0.22 0.380069028509801
0.24 0.417813509603197
0.26 0.45890987087389
0.28 0.504453869605085
0.3 0.55611198474379
0.32 0.617614261431547
0.34 0.701152095078864
0.36 0.999999426336857
0.38 0.999999426316857
0.4 0.999999426296857
0.42 0.999999998634709
0.44 0.999999426256857
0.46 0.999999426236857
0.48 0.999999999519998
0.5 0.999999426196857
0.52 0.999999426176857
0.54 0.999999426156857
0.56 0.999999999439914
0.58 0.999999426116857
0.6 0.999999426096857
0.62 0.999999426076857
0.64 0.999999999356976
0.66 0.99999999934
0.68 0.999999426016857
0.7 0.999999425996857
0.72 0.999999425976857
0.74 0.999999425956857
0.76 0.999999425936857
0.78 0.999999998544896
0.8 0.999999999196067
0.82 0.999999999179998
0.84 0.99999999916
0.86 0.99999999914
0.88 0.999999425816858
0.9 0.999999425796859
0.92 0.99999942577686
0.94 0.99999942575686
0.96 0.99999942573686
0.98 0.99999942571686
1 0.99999942569686
};
\addlegendentry{SE, AMP}
\path [draw=lightblue, very thick]
(axis cs:0.1,0.175946517641851)
--(axis cs:0.1,0.233261026746925);

\path [draw=lightblue, very thick]
(axis cs:0.2,0.308615143586235)
--(axis cs:0.2,0.414466842001875);

\path [draw=lightblue, very thick]
(axis cs:0.3,0.450480858082634)
--(axis cs:0.3,0.652879704008887);

\path [draw=lightblue, very thick]
(axis cs:0.4,0.777407990261076)
--(axis cs:0.4,1.07008474311237);

\path [draw=lightblue, very thick]
(axis cs:0.5,0.9904193514186)
--(axis cs:0.5,1.00519207934914);

\path [draw=lightblue, very thick]
(axis cs:0.6,0.995809087743473)
--(axis cs:0.6,1.00155320682603);

\path [draw=lightblue, very thick]
(axis cs:0.7,0.9870215808512)
--(axis cs:0.7,1.00637746401868);

\path [draw=lightblue, very thick]
(axis cs:0.8,0.970799082081444)
--(axis cs:0.8,1.01975546441646);

\path [draw=lightblue, very thick]
(axis cs:0.9,0.986098318447375)
--(axis cs:0.9,1.00361687391537);

\path [draw=lightblue, very thick]
(axis cs:1,0.983239838285682)
--(axis cs:1,1.00298225352815);

\path [draw=coral, very thick]
(axis cs:0.1,0.0240312487328785)
--(axis cs:0.1,0.0754834408120254);

\path [draw=coral, very thick]
(axis cs:0.2,0.0893563752269874)
--(axis cs:0.2,0.191978588325605);

\path [draw=coral, very thick]
(axis cs:0.3,0.197239302095149)
--(axis cs:0.3,0.321996954847857);

\path [draw=coral, very thick]
(axis cs:0.4,0.320984499350293)
--(axis cs:0.4,0.532709311902547);

\path [draw=coral, very thick]
(axis cs:0.5,0.517384963447208)
--(axis cs:0.5,0.676504482572392);

\path [draw=coral, very thick]
(axis cs:0.6,0.668671447049169)
--(axis cs:0.6,0.803858911019094);

\path [draw=coral, very thick]
(axis cs:0.7,0.758883656058897)
--(axis cs:0.7,0.855592131714777);

\path [draw=coral, very thick]
(axis cs:0.8,0.828644377995757)
--(axis cs:0.8,0.894183879193184);

\path [draw=coral, very thick]
(axis cs:0.9,0.865858640384381)
--(axis cs:0.9,0.921430551056449);

\path [draw=coral, very thick]
(axis cs:1,0.893939014643241)
--(axis cs:1,0.935196692408154);

\addplot [semithick, blue, mark=asterisk, mark size=3, mark options={solid}, only marks]
table {%
0.1 0.204603772194388
0.2 0.361540992794055
0.3 0.551680281045761
0.4 0.923746366686723
0.5 0.997805715383869
0.6 0.998681147284751
0.7 0.996699522434938
0.8 0.995277273248951
0.9 0.994857596181374
1 0.993111045906917
};
\addlegendentry{AMP}
\addplot [semithick, red, dotted, mark=asterisk, mark size=3, mark options={solid}]
table {%
0.1 0.049757344772452
0.2 0.140667481776296
0.3 0.259618128471503
0.4 0.42684690562642
0.5 0.5969447230098
0.6 0.736265179034132
0.7 0.807237893886837
0.8 0.86141412859447
0.9 0.893644595720415
1 0.914567853525698
};
\addlegendentry{BPDN}
\end{axis}

\end{tikzpicture}
  \vspace{-2\baselineskip}
  \caption{$\pi=0.1$}
\end{subfigure}
\hspace{0.5cm}
\begin{subfigure}[b]{0.45\textwidth}
  \centering
  % This file was created with tikzplotlib v0.10.1.
\begin{tikzpicture}[scale=0.8]

\definecolor{coral}{RGB}{255,127,80}
\definecolor{darkgray176}{RGB}{176,176,176}
\definecolor{lightblue}{RGB}{173,216,230}
\definecolor{lightgray204}{RGB}{204,204,204}

\begin{axis}[
legend cell align={left},
legend style={
  fill opacity=0.8,
  draw opacity=1,
  text opacity=1,
  at={(0.97,0.03)},
  anchor=south east,
  draw=lightgray204
},
tick align=outside,
tick pos=left,
x grid style={darkgray176},
xlabel={\(\displaystyle \delta=n/p\)},
xmajorgrids,
xmin=0.055, xmax=1.045,
xtick style={color=black},
y grid style={darkgray176},
ylabel={Overlap},
ymajorgrids,
ymin=0.0610502418809091, ymax=1.0468808065147,
ytick style={color=black}
]
\addplot [semithick, blue, dashed]
table {%
0.1 0.367783226218656
0.12 0.381836596660677
0.14 0.396075254236371
0.16 0.410515172252168
0.18 0.425173846664275
0.2 0.440066998227394
0.22 0.455232330070225
0.24 0.470689927763187
0.26 0.486471418949723
0.28 0.502625570640251
0.3 0.519201196856162
0.32 0.536263767708899
0.34 0.553893690642418
0.36 0.572195555000318
0.38 0.591295625948903
0.4 0.611402488416724
0.42 0.632747089100948
0.44 0.655795625565958
0.46 0.681259384741387
0.48 0.710448102450062
0.5 0.747029517067073
0.52 0.999999426523524
0.54 0.999999426516857
0.56 0.99999999970601
0.58 0.999999426503524
0.6 0.999999426496857
0.62 0.999999999793334
0.64 0.999999426483524
0.66 0.99999999978
0.68 0.99999942647019
0.7 0.999999426463524
0.72 0.999999426456857
0.74 0.999999999753333
0.76 0.999999426443524
0.78 0.999999998656426
0.8 0.999999999733334
0.82 0.999999426423524
0.84 0.999999426416857
0.86 0.99999942641019
0.88 0.999999999693575
0.9 0.9999999997
0.92 0.999999999691176
0.94 0.999999426383524
0.96 0.999999426376857
0.98 0.99999942637019
1 0.999999426363524
};
\addlegendentry{SE, AMP}
\path [draw=lightblue, very thick]
(axis cs:0.1,0.344607433416852)
--(axis cs:0.1,0.393977868033025);

\path [draw=lightblue, very thick]
(axis cs:0.2,0.413600035685119)
--(axis cs:0.2,0.461933428377274);

\path [draw=lightblue, very thick]
(axis cs:0.3,0.488695111057625)
--(axis cs:0.3,0.556241080294531);

\path [draw=lightblue, very thick]
(axis cs:0.4,0.577572421478383)
--(axis cs:0.4,0.651543292727169);

\path [draw=lightblue, very thick]
(axis cs:0.5,0.657233529639651)
--(axis cs:0.5,0.853674098436342);

\path [draw=lightblue, very thick]
(axis cs:0.6,0.998101360314041)
--(axis cs:0.6,1.00085250886666);

\path [draw=lightblue, very thick]
(axis cs:0.7,0.996560799363102)
--(axis cs:0.7,1.00207032630407);

\path [draw=lightblue, very thick]
(axis cs:0.8,0.996482833053805)
--(axis cs:0.8,1.0014720522762);

\path [draw=lightblue, very thick]
(axis cs:0.9,0.996464641777867)
--(axis cs:0.9,1.00129751222778);

\path [draw=lightblue, very thick]
(axis cs:1,0.996840486791841)
--(axis cs:1,1.0013012408448);

\path [draw=coral, very thick]
(axis cs:0.1,0.105860722091536)
--(axis cs:0.1,0.161323893158088);

\path [draw=coral, very thick]
(axis cs:0.2,0.166506463544948)
--(axis cs:0.2,0.234985198290698);

\path [draw=coral, very thick]
(axis cs:0.3,0.235422790012887)
--(axis cs:0.3,0.328499596114838);

\path [draw=coral, very thick]
(axis cs:0.4,0.320717309919406)
--(axis cs:0.4,0.409532976295338);

\path [draw=coral, very thick]
(axis cs:0.5,0.421391277425692)
--(axis cs:0.5,0.519683719604205);

\path [draw=coral, very thick]
(axis cs:0.6,0.527492711740712)
--(axis cs:0.6,0.626236596033471);

\path [draw=coral, very thick]
(axis cs:0.7,0.653154897260308)
--(axis cs:0.7,0.759453004452621);

\path [draw=coral, very thick]
(axis cs:0.8,0.747989951120786)
--(axis cs:0.8,0.837457372263822);

\path [draw=coral, very thick]
(axis cs:0.9,0.814282419572576)
--(axis cs:0.9,0.876733936105087);

\path [draw=coral, very thick]
(axis cs:1,0.866631659460966)
--(axis cs:1,0.910133052625631);

\addplot [semithick, blue, mark=asterisk, mark size=3, mark options={solid}, only marks]
table {%
0.1 0.369292650724939
0.2 0.437766732031196
0.3 0.522468095676078
0.4 0.614557857102776
0.5 0.755453814037996
0.6 0.999476934590351
0.7 0.999315562833587
0.8 0.998977442665
0.9 0.998881077002822
1 0.99907086381832
};
\addlegendentry{AMP}
\addplot [semithick, red, dotted, mark=asterisk, mark size=3, mark options={solid}]
table {%
0.1 0.133592307624812
0.2 0.200745830917823
0.3 0.281961193063863
0.4 0.365125143107372
0.5 0.470537498514948
0.6 0.576864653887092
0.7 0.706303950856465
0.8 0.792723661692304
0.9 0.845508177838832
1 0.888382356043299
};
\addlegendentry{BPDN}
\end{axis}

\end{tikzpicture}
  \vspace{-2\baselineskip}
  \caption{$\pi=0.3$}
\end{subfigure}
\caption{Noisy QGT with noise $\Psi_i\sim\text{Uniform}[-\lambda\sqrt{p},\lambda\sqrt{p}]$, squared overlap vs.~$\delta$: $\alpha=0.5$, $\lambda=0.1$, and $p=500$ for both AMP and BPDN.}
\label{fig:noisy QGT_AMP_v_others_corr_v_delta}
\end{figure}
\begin{figure}[t]
\centering
\begin{subfigure}[b]{0.45\textwidth}
  \centering
  % This file was created with tikzplotlib v0.10.1.
\begin{tikzpicture}[scale=0.8]

\definecolor{darkgray176}{RGB}{176,176,176}
\definecolor{lightgray204}{RGB}{204,204,204}

\pgfkeys{
        /pgfplots/every y tick label/.append style  =
            { 
              /pgf/number format/.cd,
               precision = 2, 
               fixed
            }
    }

\begin{axis}[
legend cell align={left},
legend style={fill opacity=0.8, draw opacity=1, text opacity=1, draw=lightgray204},
tick align=outside,
tick pos=left,
x grid style={darkgray176},
xlabel={FNR},
xmajorgrids,
xmin=0.0925709201874654, xmax=0.841382570287082,
xtick style={color=black},
y grid style={darkgray176},
ylabel={FPR},
ymajorgrids,
ymin=-0.00527713076642693, ymax=0.133654307762062,
ytick style={color=black}
]
\addplot [semithick, blue, dashed, mark=x, mark size=3, mark options={solid}]
table {%
0.126607813373812 0.127339242374403
0.215806503155708 0.0676457775841683
0.291164016936966 0.0418363219330921
0.361228728928657 0.0272218913080199
0.42995526084525 0.0177515713485593
0.500299592554126 0.0111294595627784
0.575996644563394 0.00656099143870629
0.665125429415994 0.00336828676333313
0.784453143724535 0.00103793462123165
};
\addlegendentry{SE, AMP}
\addplot [semithick, blue, mark=*, mark size=2, mark options={solid}, only marks]
table {%
0.13663389643522 0.124654049420179
0.224440835664118 0.0662900328041285
0.294618482033286 0.0408409035503691
0.359904315381454 0.0272518611716491
0.427786662816896 0.0186123229025989
0.490675101470979 0.0122419684793685
0.560192632586518 0.00754626399069147
0.649407751924307 0.00396426790947798
0.756517274257096 0.00156066202952662
};
\addlegendentry{AMP}
\addplot [semithick, red, dotted, mark=*, mark size=2, mark options={solid}]
table {%
0.394442967706135 0.125324232720871
0.448627322291919 0.0987007679487499
0.500891425985549 0.075168197447596
0.559332203849181 0.0571860726498079
0.617841415365903 0.0430170812242174
0.671904536421688 0.0319313664144723
0.72596863431238 0.0236069368271756
0.767316524655001 0.0168938201498267
0.807345677100736 0.0114352629500408
};
\addlegendentry{BPDN}
\end{axis}

\end{tikzpicture}
  \vspace{-2\baselineskip}
  \caption{$\pi=0.1$}
\end{subfigure}
\hspace{0.5cm}
\begin{subfigure}[b]{0.45\textwidth}
  \centering
  % This file was created with tikzplotlib v0.10.1.
\begin{tikzpicture}[scale=0.8]

\definecolor{darkgray176}{RGB}{176,176,176}
\definecolor{lightgray204}{RGB}{204,204,204}

\begin{axis}[
legend cell align={left},
legend style={fill opacity=0.8, draw opacity=1, text opacity=1, draw=lightgray204},
tick align=outside,
tick pos=left,
x grid style={darkgray176},
xlabel={FNR},
xmajorgrids,
xmin=0.00110286475209297, xmax=0.980240945326825,
xtick style={color=black},
y grid style={darkgray176},
ylabel={FPR},
ymajorgrids,
ymin=-0.0285961160691548, ymax=0.64358590582132,
ytick style={color=black}
]
\addplot [semithick, blue, dashed, mark=x, mark size=3, mark options={solid}]
table {%
0.0476344620682981 0.612831141312186
0.139843220749056 0.381637283445397
0.244980093907367 0.245092761689629
0.356250646572581 0.155963577846855
0.469179762194776 0.0960129536527102
0.58593992464718 0.0551943833236
0.704241849848327 0.027460829907631
0.823925354655489 0.0103834835216458
0.934768752273433 0.00195761219859412
};
\addlegendentry{SE, AMP}
\addplot [semithick, blue, mark=*, mark size=2, mark options={solid}, only marks]
table {%
0.0456091411418535 0.613032177553572
0.142909545517891 0.380004755640193
0.251254434463913 0.243186378202508
0.362666888460506 0.155831612180663
0.476376208707086 0.0955671422366251
0.585875382892488 0.0530010881900061
0.705675902264178 0.0265536138771784
0.825620981271123 0.0110704486170465
0.935734668937064 0.00250180894495603
};
\addlegendentry{AMP}
\addplot [semithick, red, dotted, mark=*, mark size=2, mark options={solid}]
table {%
0.405699719954067 0.288105764807341
0.438948691563893 0.260740385900421
0.468919238088942 0.237220815023775
0.503028415388955 0.214463644767347
0.537978671113723 0.192496574986101
0.571982669923217 0.172959099593896
0.601563086172324 0.153963376224921
0.632190060285685 0.135893235611623
0.660020613538654 0.119687180480514
};
\addlegendentry{BPDN}
\end{axis}

\end{tikzpicture}
  \vspace{-2\baselineskip}
  \caption{$\pi=0.3$}
\end{subfigure}
\caption{Noisy QGT with noise $\Psi_i\sim\text{Uniform}[-\lambda\sqrt{p},\lambda\sqrt{p}]$, FPR vs.~FNR: $\alpha=0.5$, $\lambda=0.1$, $\delta=0.3$, $p=500$, and $\zeta\in\{0, 0.1, \dots, 1.0\}$ for both AMP and BPDN.}
\label{fig:noisy QGT_AMP_v_others_FPR_FNR}
\end{figure}
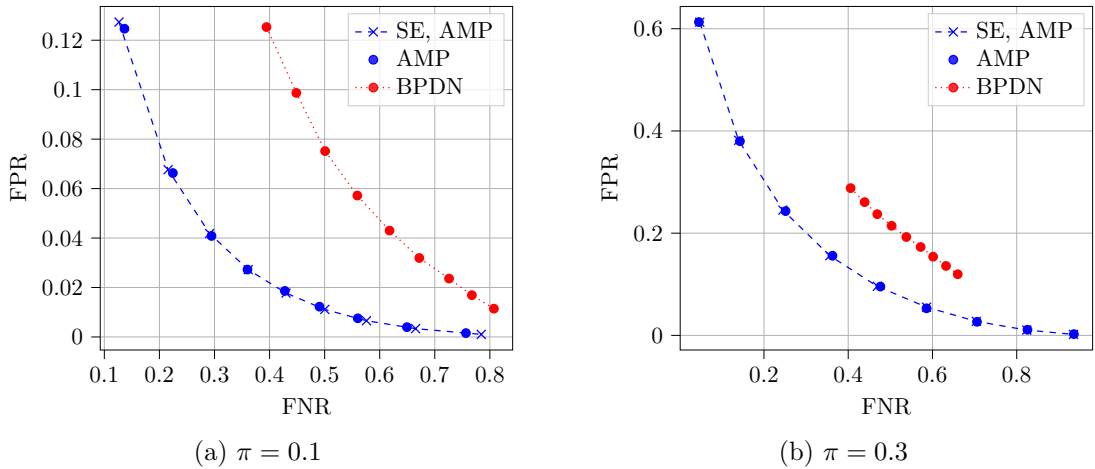

\section{Discussion and Future Directions} \label{sec:disc}

In this paper, we established a  state evolution result for the AMP algorithm applied to a matrix GLM with a generalized white noise design matrix. Using this result, we obtained rigorous performance guarantees for pooled data decoding and QGT, with i.i.d.~Bernoulli design matrices.

% \paragraph{AMP vs convex estimators} The numerical results in Section \ref{sec:sim_pool_data} and \ref{sec:QGT_sims} show that AMP outperforms the convex programming estimator for noiseless pooled data and QGT, but the convex estimator performs slightly better for noisy pooled data with three categories, up to a certain value of $\delta$.  The fact that the convex estimator beats AMP in any regime is surprising, since for linear regression with a vector-valued signal, it has been shown that  AMP with Bayes-optimal denoisers is superior to convex procedures for a large class of priors, and moreover, is conjectured to be optimal among polynomial-time algorithms \cite{Cel22}. However, this result does not apply to the pooled data problem where the signal is matrix-valued. (Even if we vectorize the pooled data problem in \eqref{eq:opt_defs}, the effective design matrix is then block-diagonal, which does not satisfy the assumptions in  \cite{Cel22} or other analyses of convex estimators \cite{Han23}.) Obtaining asymptotic performance  guarantees for the convex pooled data estimator, e.g. via Gaussian comparison theorems \cite{Thram15} and associated universality results \cite{Han23}, is an interesting direction for future work.
 
\paragraph{Beyond i.i.d. designs} Based on results for Gaussian linear and generalized linear models \cite{Don13,Cob23}, we expect that the performance of AMP can be improved using spatially coupled Bernoulli designs. 
%We expect that our analysis based on the universality result of \cite{Wan22} can be extended to spatially coupled designs. 
Our analysis can also be generalized to designs where the Bernoulli parameter may vary across tests, e.g.,  due to varying constraints on the number of items that can be included in each test.

\paragraph{Non-additive noise models} One could consider a general noise model for pooled data and QGT, where the observations can be represented as
$Y=q(XB, \Psi)$, for a \emph{nonlinear} function $q$ that acts row-wise. The main challenge in analyzing such a model is that to use AMP and obtain guarantees via Theorem \ref{thm:GAMP}, it  needs to be expressed in terms of an equivalent model  $\tY=\tilde{q}(\tX B, \tPsi)$, where $\tX$ is a centered and rescaled matrix  with independent zero-mean entries of order $1/n$. The additive noise model allows us to obtain the equivalent model in a straightforward way (see \eqref{eq:mod_model} and \eqref{eq:noisy_QGT_rescaled}). It would be interesting to investigate general noise models that can be expressed in terms of a suitable equivalent model through  transformations of the data.

\paragraph{AMP for Boolean group testing} Recall that in noiseless BGT, the outcome of test $i \in [n]$ is $Y_i=\mathds{1}\{X_{i,:}^\top\beta>0\}$. As in QGT, the signal $\beta \in  \reals^p$ is a binary vector, with ones in the entries corresponding to defective items. In the linear regime, where the number of defectives $d$ is proportional to $p$, it is known that individual testing is optimal for \emph{exact} recovery \cite{Ald18,Bay20}. On the other hand, precisely characterizing the performance under an approximate recovery criterion, specifically the trade-off between FPR and FNR  in the linear regime,  is an open question  \cite{Tan23b}.
The similarity between BGT and QGT suggests that the AMP analysis for the latter can be generalized to provide rigorous guarantees for BGT with i.i.d.~Bernoulli matrices, similar to those of Theorem \ref{thm:GAMP_vec} for QGT. However, this is not the case. We briefly describe  the challenges in extending our results to BGT.

Consider noiseless BGT with an i.i.d.~$\text{Bernoulli}(\alpha)$  design matrix $X$. Using the decomposition in \eqref{eq:X_ij_decomp}, we can rewrite \eqref{eq:BGT} as the following re-centered and rescaled model: 
$$Y_i=\mathds{1}\left\{\tX_{i,:}^\top\beta>\frac{-\alpha p \hat{\pi}}{\sqrt{n\alpha(1-\alpha)}}\right\}=q(\tX_{i,:}^\top\beta).$$
With $\tY_i:=Y_i$, this is a special case of the matrix-GLM model in \eqref{eq:matrix-GLM}. For each test in BGT to be informative (i.e., have a non-vanishing probability of being either positive or negative),  we need $\alpha=\Theta\big(\frac{1}{p}\big)$, since the number of defectives is a constant fraction of $p$ \cite[Section 2.1]{Ald19}. However, with this choice of $\alpha$, the re-centered and rescaled  matrix $\tX$ does not satisfy the second condition in Definition \ref{def:gen_white_noise_matrix}, which is required for rigorous AMP guarantees. In contrast, for QGT, the choice of $\alpha = \Theta(1)$ both leads to informative tests and gives an $\tX$ satisfying all the conditions in Definition \ref{def:gen_white_noise_matrix}. We can attempt to overcome this issue by choosing $\alpha=\Theta\big(\frac{\log n}{n}\big)$, which ensures that $\tX$ satisifes  Definition \ref{def:gen_white_noise_matrix}, but this choice leads to uninformative tests. Intuitively, when $\alpha = \Theta\big(\frac{\log n}{n}\big)$, the average number of defective items in each test is of order $\frac{p\log n}{n}=\Theta(\log n)$ (where we used $n=\Theta(p)$), causing the test outcome to be positive with probability tending to $1$. That is, each individual test provides almost no information as $p$ grows large.

A potential solution would be to shift away from the linear regime where the number of defectives $d=\Theta(p)$, and instead apply the AMP algorithm to the mildly sublinear regime $d=\frac{p}{\log p}$. This would lead to each test being informative with $\alpha=\Theta\big(\frac{\log n}{n}\big)$, but does not satisfy the AMP assumption that requires the empirical distribution of the signal $\beta$ to converge to a well-defined limiting distribution. As a result, we would need to develop a new AMP algorithm that works for signals with sublinear sparsity -- we leave this as an open problem.

\newpage

\bibliographystyle{siamplain}
\bibliography{NT_References}

\newpage

\appendix

%%%%%%

\section{Proof of Theorem \ref{thm:GAMP}} \label{sec:GAMP_thm_proof}
We begin by stating the abstract AMP iteration from \cite{Wan22} for a generalized white noise matrix $\tX\in\mb{R}^{n\times p}$, as defined in Definition \ref{def:gen_white_noise_matrix}. For $t \ge 1$, the iterates of the  abstract AMP, denoted by $h^t\in\mb{R}^{p}$ and $e^t\in\mb{R}^n$, are produced using functions $f_t^v:\mb{R}^{t+2L}\rightarrow\mb{R}$, $f_{t+1}^u:\mb{R}^{t+L_\Psi}\rightarrow\mb{R}$.
Given an initializer $u^1\in\mb{R}^n$, side information $c^1,\dots,c^{L_\Psi}\in\mb{R}^n$ and $d^1,\dots,d^{2L}\in\mb{R}^p$, all independent of $\tX$, the iterates of the abstract AMP are computed as:
\begin{align}
\begin{split}
    h^t&=\sqrt{\delta}\tX^\top u^t-\sum_{s=1}^{t-1}b_{s}^tv^s, \qquad 
    v^t=f_t^v(h^1,\dots,h^t,d^1,\dots,d^{2L}) \\
    e^t&=\sqrt{\delta}\tX v^t-\sum_{s=1}^ta_{s}^tu^s, \qquad
    u^{t+1}=f_{t+1}^u(e^1,\dots,e^t,c^1,\dots,c^{L_\Psi}),
\end{split} \label{eq:abs_AMP}
\end{align}
where the memory coefficients  $\{b_{s}^t\}_{s<t}$ and $\{a_{s}^t\}_{s\leq t}$ are defined below in \eqref{eq:abstract_AMP_onsager}. We note that the matrix $\tX$ used in \eqref{eq:abs_AMP} has a different scaling to the one used in \cite{Wan22}: the $(i,j)$th entry of $\tX$  has variance $n^{-1}S_{ij}$,  whereas in \cite{Wan22} the variance is $p^{-1}S_{ij}$. 
This difference is accounted for by the $\sqrt{\delta}$ factor multiplying $\tX$ and $\tX^\top$ in \eqref{eq:abs_AMP}.

Next, we have the following assumptions on the initializer and the side information vectors: 
\begin{assumption} \label{ass:universality_result}
When $n,p\rightarrow\infty$, we have $n/p=\delta\in(0,\infty)$, for fixed $L$. Furthermore, we have
\begin{align*}
    (u^1,c^1,\dots,c^{L_\Psi})\stackrel{W}{\rightarrow}
    (\bu^1,\bc^1,\dots,\bc^{L_\Psi})
    \text{ and }
    (d^1,\dots,d^{2L})\stackrel{W}{\rightarrow}
    (\bd^1,\dots,\bd^{2L}),
\end{align*}
for joint limit laws $(\bu^1,\bc^1,\dots,\bc^{L_\Psi})$ and $(\bd^1,\dots,\bd^{2L})$ having finite moments of all orders, where $\E[(\bu^1)^2]\geq0$. Multivariate polynomials are dense in the real $L^2$-spaces
of functions $f:\mb{R}^{L_\Psi+1}\rightarrow\mb{R}$, $g:\mb{R}^{2L}\rightarrow\mb{R}$ with the inner-products
\begin{align*}
    \langle f,\tilde{f} \rangle
    :=\E\big[f(\bu^1,\bc^1,\dots,\bc^{L_\Psi})\tilde{f}(\bu^1,\bc^1,\dots,\bc^{L_\Psi})\big]
    \text{ and }
    \langle g,\tilde{g} \rangle
    :=\E\big[g(\bd^1,\dots,\bd^{2L})\tilde{g}(\bd^1,\dots,\bd^{2L})\big].
\end{align*}
\end{assumption}

The state evolution result below states that the joint empirical distribution of $(h^1, \ldots, h^t)$ converges to a Gaussian law $\normal(0, \Omega^t)$. Similarly, the joint empirical distribution of $(e^1, \ldots, e^t)$ converges to $\normal(0, \Gamma^t)$. The covariance matrices $\Omega^t, \Gamma^t \in \reals^{t \times t}$ are iteratively defined as follows, starting from $\Omega^1=\delta\E[(\bu^1)^2]\in\mb{R}^{1\times 1}$. Given $\Omega^t$, for $t \ge 1$, let $(\bh^1,\dots,\bh^t)\sim\normal(0,\Omega^t)$ independent of $(\bd^1,\dots,\bd^{2L})$ and define
\begin{align*}
    \bv^s=f_s^v(\bh^1,\dots,\bh^s,\bd^1,\dots,\bd^{2L}), \quad  s\in[t].
\end{align*}
Then, $\Gamma^t=(\E[\bv^r\bv^s])_{r,s=1}^t\in\mb{R}^{t\times t}$. Next, let 
$(\be^1,\dots,\be^t)\sim\normal(0,\Gamma^t)$ independent of $(\bu^1,\bc^1,\dots,$ $\bc^{L_\Psi})$ and define
\begin{align*}
    \bu^{s+1}=f_{s+1}^u(\be^1,\dots,\be^s,\bc^1,\dots,\bc^{L_\Psi}), \quad s \in [t].
\end{align*}
Then,  $\Omega^{t+1}=(\delta\cdot\E[\bu^r\bu^s])_{r,s=1}^{t+1}\in\mb{R}^{(t+1)\times(t+1)}$.
We then define the memory coefficients in \eqref{eq:abs_AMP} as
\begin{align}
    a_{s}^t=\E\big[\partial_sf_t^v(\bh^1,\dots,\bh^t,\bd^1,\dots,\bd^{2L})\big]
    \text{ and }
    b_{s}^t=\delta\cdot\E\big[\partial_sf_t^u(\be^1.\dots,\be^{t-1},\bc^1,\dots,\bc^{L_\Psi})\big],
    \label{eq:abstract_AMP_onsager}
\end{align}
where $\partial_s$ denotes partial derivative in the $s$th argument. The following theorem gives the state evolution result for the abstract AMP iteration.
\begin{theorem} \label{thm:abs_AMP} \textup{\cite[Theorem 2.21]{Wan22}}
Let $\tX\in\mb{R}^{n\times p}$ be a generalized white noise matrix (as defined in Definition \ref{def:gen_white_noise_matrix}) with variance profile $S\in\mb{R}^{n\times p}$, and let $u^1,c^1,\dots,c^{L_\Psi},d^1,\dots,d^{2L}$ be independent of $\tX$ and satisfy Assumption \ref{ass:universality_result}. Suppose that
\begin{enumerate}
    \item Each function $f_t^v:\mb{R}^{t+2L}\rightarrow\mb{R}$ and $f_{t+1}^u:\mb{R}^{t+L_\Psi}\rightarrow\mb{R}$ is continuous, is Lipschitz in its first $t$ arguments, and satisfies the polynomial growth condition in \eqref{eq:poly_growth_cond} for some order $r\geq 1$.
    \item $\|\tX\|_{\textup{op}}<C$, for some constant $C$ almost surely for all large $n$ and $p$.
    \item For any fixed polynomial functions $f^\dag:\mb{R}^{L_\Psi+1\rightarrow\mb{R}}$ and $f^\ddag:\mb{R}^{2L}\rightarrow\mb{R}$, almost surely as $n,p\rightarrow\infty$,
    \begin{align*}
        \max_{i=1}^n\left|\langle f^\dag(u^1,c^1,\dots,c^{L_\Psi})\odot S_{i,:}\rangle-\langle f^\dag(u^1,c^1,\dots,c^{L_\Psi})\rangle\cdot\langle S_{i,:}\rangle\right|
        \rightarrow0 \\
        \max_{j=1}^p\left|\langle f^\ddag(d^1,\dots,d^{2L})\odot S_{:,j}\rangle-\langle f^\ddag(d^1,\dots,d^{2L})\rangle\cdot\langle S_{:,j}\rangle\right|
        \rightarrow0,
    \end{align*}
\end{enumerate}
Also assume that the matrices $\Omega^t$ and $\Gamma^t$ are non-singular, for $t \ge 1$. Then for any fixed $t\geq1$, almost surely as $n,p\rightarrow\infty$ with $n/p=\delta\in(0,\infty)$, the iterates of the abstract AMP in \eqref{eq:abs_AMP} satisfy
\begin{align*}
    & (u^1,c^1,\dots,c^{L_\Psi},e^1,\dots,e^t)
    \stackrel{W_2}{\rightarrow}(\bu^1,\bc^1,\dots,\bc^{L_\Psi},\be^1,\dots,\be^t) \\
    & (d^1,\dots,d^{2L},h^1,\dots,h^t)
    \stackrel{W_2}{\rightarrow}
    (\bd^1,\dots,\bd^{2L},\bh^1,\dots,\bh^t),
\end{align*}
where $(\bh^1,\dots,\bh^t)\sim\normal(0,\Omega^t)$ and $(\be^1,\dots,\be^t)\sim\normal(0,\Gamma^t)$ are independent of $(\bu^1,\bc^1.\dots,\bc^{L_\Psi})$ and $\bd^1,\dots,\bd^{2L}$.
\end{theorem}
The theorem above shows that the state evolution result for an abstract AMP iteration defined via a generalized white noise matrix is the same  as the one for an i.i.d. Gaussian matrix (where $\tX_{ij}\stackrel{\iid}{\sim}\normal(0,\frac{1}{n})$), which was previously given in \cite{Jav13}.

\subsection{Reduction of matrix-AMP to abstract AMP} We now inductively map the iterates of matrix-AMP \eqref{eq:GAMP} to iterates of the abstract AMP.  At each step of the induction, there are two main reductions:
\begin{itemize}
    \item Reduction of the matrix-AMP iterates to the abstract AMP iterates;
    \item Reduction of the matrix-AMP SE parameters to the abstract AMP SE parameters.
\end{itemize}
Recalling that each matrix-AMP iterate has $L$ columns, we will use $4L$ abstract AMP iterations to represent the first iteration of matrix-AMP. We will then represent each of the subsequent matrix-AMP iterations using $2L$ abstract AMP iterations.

We will prove Theorem \ref{thm:GAMP} for a slightly modified version of the matrix-AMP algorithm, where the matrices $C^{k,m+1}$ and $F^{k,m+1}$ in \eqref{eq:CkF_k1_def} are replaced by their deterministic versions, defined in \eqref{eq:CkF_k1_det_def} below. The state evolution result remains valid for the matrix-AMP algorithm in \eqref{eq:GAMP}, by standard arguments in Remark 2.19 of \cite{Wan22}.

Starting with an initialization $\hB^0\in\mb{R}^{p\times L}$ and defining $\hR^{-1} := 0\in\mb{R}^{n\times L}$, for iteration $k \ge 0$ the matrix-AMP algorithm computes:
\begin{align}
\begin{split}
        & \Theta^k=\tX\hB^k-\sum_{m=0}^{k-1}\hR^{m}(F^{k,m+1})^\top, \quad \hR^k=g_k(\Theta^0,\dots,\Theta^k,\tY), \\ 
        & B^{k+1}=\tX^\top\hR^k-\sum_{m=0}^k\hB^m (C^{k,m+1})^\top, \quad \hB^{k+1}=f_{k+1}(B^1,\dots,B^{k+1}).
\end{split}
\label{eq:GAMP_deterministic}
\end{align}
The deterministic matrices $C^{k,m+1}, F^{k,m+1} \in \reals^{L \times L}$ are defined as
\begin{align}
    C^{k,m+1}=\E\big[\partial_{m+1}g_k(Z^0,\dots,Z^k,\bar{Y})\big], \ \  
    F^{k,m+1}=\frac{1}{\delta}\E\big[\partial_{m+1}f_{k}(\Mu_B^{1}\bar{B}+G_B^{1},\dots,\Mu_B^{k}\bar{B}+G_B^{k})\big], 
    \label{eq:CkF_k1_det_def}
\end{align}
where $\partial_{m+1} g_k, \partial_{m+1}f_{k}$ denote the Jacobians of $g_k, f_{k}$, with respect to their $(m+1)$th argument.

\textbf{Initialization of abstract AMP:} We initialize $u^1=0$, and  define the side information vectors
\begin{align}
    & d^1=B_{:,1},\dots,\,d^L=B_{:,L},\,
    d^{L+1}=\hB^0_{:,1},\dots,\,d^{2L}=\hB^0_{:,L} \, ,  \nonumber \\
   & c^1=\tPsi_{:,1},\dots,\,c^{L_\Psi}=\tPsi_{:,L_\Psi}.  \label{eq:abs_AMP_reduction0}
\end{align}

\textbf{Base case:} We consider the case $k=0$, and our goal is to reduce $\hB^0$, $\Theta^0$, $B^1$, $\hR^0$, $\hB^1$, and $\Theta^1$ to abstract AMP iterates, via careful choices of the functions $f^v_t$ and $f^u_{t+1}$. This requires $4L$ abstract AMP iterations. We give a summary of the reductions below, and then  provide detailed derivations:
\begin{itemize}
    \item For $t=1$:
    \begin{align}
    \begin{split}
        h^1=0;\quad
        v^1=\frac{1}{\sqrt{\delta}}B_{:,1};\quad
        e^1=\Theta_{:,1};\quad
        u^{2}=0
    \end{split}
    \label{eq:abs_AMP_reduction1}
    \end{align}
    \item For $t=2,\dots,L$:
    \begin{equation}
    \begin{split}
       &  h^2,\dots,h^L =0;  \quad 
        \big(v^2,\dots,v^L\big)
        =\frac{1}{\sqrt{\delta}}\big(B_{:,2},\dots,B_{:,L}); \\
        & \big(e^2,\dots,e^L)=\big(\Theta_{:,2},\dots,\Theta_{:,L}\big); \quad 
        u^{3},\dots,u^{L+1} =0.
    \end{split}
    \label{eq:abs_AMP_reduction2}
    \end{equation}
    \item For $t=L+1,\dots,2L-1$:
    \begin{align}
    \begin{split}
       &  h^{L+1},\dots,h^{2L-1} =0;  \quad 
        \big(v^{L+1},\dots,v^{2L-1}\big)=\frac{1}{\sqrt{\delta}}\Big(\hB_{:,1}^0,\dots,\hB_{:,L-1}^0\Big); \\
        & \big(e^{L+1},\dots,e^{2L-1}\big) =\Big(\Theta_{:,1}^0,\dots,\Theta_{:,L-1}^0\Big); \quad 
        u^{L+1},\dots,u^{2L}=0.
    \end{split}
    \label{eq:abs_AMP_reduction3}
    \end{align}
    \item For $t=2L$:
    \begin{align}
    \begin{split}
        h^{2L}=0;\quad
        v^{2L}=\frac{1}{\sqrt{\delta}}\hB_{:,L}^0;\quad
        e^{2L}=\Theta_{:,L}^0;\quad 
        u^{2L+1}=\frac{1}{\sqrt{\delta}}\hR_{:,1}^0.
    \end{split}
    \label{eq:abs_AMP_reduction4}
    \end{align}
    \item For $t=2L+1,\dots,3L-1$:
    \begin{align}
    \begin{split}
        &\big(h^{2L+1},\dots,h^{3L-1}\big)
        =\Big(B_{:,1}^1-\{B(\Mu_B^1)^\top\}_{:,1},\dots,B_{:,L-1}^1-\{B(\Mu_B^1)^\top\}_{:,L-1}\Big); \\ 
        & v^{2L+1},\dots,v^{3L-1} =0; \\
        & e^{2L+1},\dots,e^{3L-1} =0; \quad 
        \big(u^{2L+2},\dots,u^{3L}\big)
        =\frac{1}{\sqrt{\delta}}\Big(\hR_{:,2}^0,\dots,\hR_{:,L}^0\Big).
    \end{split}
    \label{eq:abs_AMP_reduction5}
    \end{align}
    \item For $t=3L$:
    \begin{align}
    \begin{split}
        h^{3L}=B_{:,L}^1-\{B(\Mu_B^1)^\top\}_{:,L};\quad 
        v^{3L}=0;\quad 
        e^{3L}=0;\quad 
        u^{3L+1}=0.
    \end{split}
    \label{eq:abs_AMP_reduction6}
    \end{align}
    \item For $t=3L+1,\dots,4L-1$:
    \begin{align}
    \begin{split}
        & h^{3L+1},\dots,h^{4L-1}=0; \quad 
        \big(v^{3L+1},\dots,v^{4L-1}\big)=\frac{1}{\sqrt{\delta}}\Big(\hB_{:,1}^1,\dots,\hB_{:,L-1}^{1}\Big);  \\
        & \big(e^{3L+1},\dots,e^{4L-1}\big) =\Big(\Theta_{:,1}^1,\dots,\Theta_{:,L-1}^1\Big); \quad 
        u^{3L+2},\dots,u^{4L}=0.
    \end{split}
    \label{eq:abs_AMP_reduction7}
    \end{align}
    \item For $t=4L$:
    \begin{align}
    \begin{split}
        h^{4L}&=0;\quad 
        v^{4L}=\frac{1}{\sqrt{\delta}}\hB_{:,L}^1;\quad
        e^{4L}=\Theta_{:,L}^1;\quad 
        u^{4L+1}=\frac{1}{\sqrt{\delta}}\hR_{:,1}^1.
    \end{split}
    \label{eq:abs_AMP_reduction8}
    \end{align}
\end{itemize}

We now provide the detailed derivations of \eqref{eq:abs_AMP_reduction1}-\eqref{eq:abs_AMP_reduction8}. For $t=1,\dots,L+1$,  define 
\begin{align}
    f_t^v(h^1,\dots,h^L,d^1,\dots,d^{2L})
    =\frac{1}{\sqrt{\delta}}d^t,  \ 
    \text{ and } \ 
    f_{t+1}^u(e^1,\dots,e^t,c^1,\dots,c^{L_\Psi})=0.
    \label{eq:fufv_choice1}
\end{align}
For $t=1$, using $u^1=0$, this gives
\begin{align*}
   &  h^1=0; \quad v^1=f_1^v(h^1,d^1,\dots,d^{2L}) 
    =\frac{1}{\sqrt{\delta}}d^1 
    =\frac{1}{\sqrt{\delta}}B_{:,1}; \\
   &  e^1=\sqrt{\delta}\tX v^1-a_{1}^1u^1=\tX B_{:,1} \, ; \quad 
    u^2=f_2^u(e^1,c^1,\dots,c^{L_\Psi})=0,
\end{align*}
giving us \eqref{eq:abs_AMP_reduction1}. For $t=2$, using \eqref{eq:fufv_choice1} we have
\begin{align*}
    & h^2=\sqrt{\delta}\tX^\top u^2-b_{1}^2v^1=0 \, ; \quad 
    v^2=f_2^v(h^1,h^2,d^1,\dots,d^{2L})
    =\frac{1}{\sqrt{\delta}}d^2=\frac{1}{\sqrt{\delta}}B_{:,2} \, ; \\
    & e^2=\sqrt{\delta}\tX v^2-\sum_{s=1}^2a_{s}^2u^s
    =\tX B_{:,2} \, ;  \quad 
     u^3=f_3^u(e^1,e^2,c^1,\dots,c^{L_\Psi})=0.
\end{align*}
A similar derivation follows for $t=3,\dots,L$ to give
\begin{align*}
    h^t=0;\quad v^t=\frac{1}{\sqrt{\delta}}B_{:,t};\quad
    e^t=\tX B_{:,t};\quad u^{t+1}=0.
\end{align*}
This completes the derivation of \eqref{eq:abs_AMP_reduction2}. For $t=L+1$, we have
\begin{align*}
    h^{L+1}&=\sqrt{\delta}\tX^\top u^{L+1}-\sum_{s=1}^Lb^{L+1}_sv^s
    \stackrel{(a)}{=}0 \, ;\\
    v^{L+1}&=f_{L+1}^v(h^1,\dots,h^{L+1},d^1,\dots,d^{2L})
    =\frac{1}{\sqrt{\delta}}d^{L+1}
    =\frac{1}{\sqrt{\delta}}\hB^0_{:,1} \, ;\\
    e^{L+1}&=\sqrt{\delta}\tX v^{L+1}-\sum_{s=1}^{L+1}a^{L+1}_su^s
    \stackrel{(b)}{=}\tX\hB_{:,1}^0 \, ;\\
    u^{L+2}&=f_{L+2}^u(e^1,\dots,e^{L+1},c^1,\dots,c^{L_\Psi})=0 \, ;
\end{align*}
Here (a) uses $u^{L+1}=0$ and $b_{s}^{L+1}=0$ (because $f_{L+1}^u=0$), and (b) uses $u^1,\dots,u^{L+1}=0$.

For $t=L+2,\dots,2L-1$, let us define
\begin{align}
    f_t^v(h^1,\dots,h^L,d^1,\dots,d^{2L})
    =\frac{1}{\sqrt{\delta}}d^{t-L}
    \text{ and }
    f_{t+1}^u(e^1,\dots,e^t,c^1,\dots,c^{L_\Psi})=0.
     \label{eq:fufv_choice2}
\end{align}
Then, a similar derivation follows for $t=L+2,\dots,2L-1$ to give
\begin{align*}
    h^t=0;\quad v^t=\frac{1}{\sqrt{\delta}}\hB^0_{:,t-L};\quad
    e^t=\tX\hB^0_{:,t-L};\quad u^{t+1}=0.
\end{align*}
This completes the derivation of \eqref{eq:abs_AMP_reduction3}. For $t=2L$, define
\begin{align}
\begin{split}
    & f_{2L}^v(h^1,\dots,h^{2L},d^1,\dots,d^{2L})
    =\frac{1}{\sqrt{\delta}}d^{2L},    \\
    & f_{2L+1}^u(\underbrace{e^1,\dots,e^L}_{\tX B=\Theta},\underbrace{e^{L+1},\dots,e^{2L}}_{\tX\hB^0=\Theta^0},\underbrace{c^1,\dots,c^{L_\Psi}}_{\tPsi})
    =\frac{1}{\sqrt{\delta}}\big\{\tg_{0}(\Theta,\Theta^0,\tPsi)\big\}_{:,1}.
\end{split}
      \label{eq:fufv_choice3}
\end{align}
Using this, we have
\begin{align*}
    & h^{2L} =\sqrt{\delta}\tX^\top u^{2L}-\sum_{s=1}^{2L-1}b_{s}^{2L}v^s
    \stackrel{(a)}{=}0; \\
    & v^{2L} =f_{2L}^v(h^1,\dots,h^{2L},d^1,\dots,d^{2L})
    =\frac{1}{\sqrt{\delta}}d^{2L}
    =\frac{1}{\sqrt{\delta}}\hB^0_{:,L}; \\
    & e^{2L} =\sqrt{\delta}\tX v^{2L}-\sum_{s=1}^{2L}a_{s}^{2L}u^s
    \stackrel{(b)}{=}\tX\hB^0_{:,L}; \\
    & u^{2L+1} =f_{2L+1}^u( e^1,\dots,e^L , e^{L+1},\dots,e^{2L}, c^1,\dots,c^{L_\Psi})
    =\frac{1}{\sqrt{\delta}}\big\{g_{0}(\Theta^0,q(\Theta,\tPsi))\big\}_{:,1}
    =\frac{1}{\sqrt{\delta}}\hR_{:,1}^0 \, .
\end{align*}
Here (a) uses $u^{2L}=0$ and $b_{s}^{2L}=0$ (because $f_{2L}^u=0$), and (b) holds because $u^1,\dots,u^{2L}=0$. This completes the derivation of \eqref{eq:abs_AMP_reduction4}.

For $t=2L+1, \ldots, 3L-1$, we define
\begin{align}
    \begin{split}
        & f_{t}^v(h^1,\dots,h^{t},d^1,\dots,d^{2L})=0,  \\ 
        & f_{t+1}^u (e^1,\dots,e^{t},c^1,\dots,c^{L_\Psi}) = \frac{1}{\sqrt{\delta}}\big\{g_{0}(\Theta^0,q(\Theta,\tPsi))\big\}_{:, \, t+1-2L}.
    \end{split}
    \label{eq:fufv_choice4}
\end{align}
For $t=2L+1$, using \eqref{eq:CkF_k1_det_def}  we first write the first column of the term $\{\hB^0(C^{0,1})^\top\}$   as
\begin{align}
    \{\hB^0(C^{0,1})^\top\}_{:,1}
    &=\hB^0\cdot
    \begin{bmatrix}
        \E[\tpartial_1g_{0,1}(Z^0,\bar{Y})] \\
        \vdots \\
        \E[\tpartial_Lg_{0,1}(Z^0,\bar{Y})]
    \end{bmatrix}
    =\sum_{l=1}^L\hB_{:,l}^0 \, \E[\tpartial_lg_{0,1}(Z^0,\bar{Y})],
    \label{eq:exp_form}
\end{align}
where $g_{0,1}$ refers to the first entry of the output of $g_0$, and $\tpartial_lg_{0,1}$ means we are taking derivative w.r.t.~the $l$-th entry of the input of $g_{0,1}$, e.g., $\tpartial_1 \, g_{0,1}$ is taking derivative w.r.t.~$Z^0_{1}$. We then have
\begin{align*}
    h^{2L+1}&=\sqrt{\delta}\tX^\top u^{2L+1}-\sum_{s=1}^Lb_{s}^{2L+1}v^s-\sum_{s=L+1}^{2L}b_{s}^{2L+1}v^s \\
    &=\sqrt{\delta}\tX^\top u^{2L+1}-\sum_{s=1}^L\delta\E[\partial_sf_{2L+1}^u(\be^1,\dots,\be^{2L},\bc^1,\dots,\bc^{L_\Psi})]v^s \\
    &\qquad-\sum_{s=L+1}^{2L}\delta\E\big[\partial_sf_{2L+1}^u(\be^1,\dots,\be^{2L},\bc^1,\dots,\bc^{L_\Psi})\big]v^s \\
    &\stackrel{(a)}{=}\tX\hR_{:,1}^0-\sum_{l=1}^L\E\big[\tpartial_l\tg_{0,1}(Z,Z^0,\bar{\Psi})\big]B_{:,l}-\sum_{l=1}^L\E\big[\tpartial_l\tg_{0,1}(Z^0,\bar{Y})\big]\hB_{:,l}^0 \\
    &\stackrel{(b)}{=}B_{:,1}^1-\{B(\Mu_B^k)\}_{:,1},
\end{align*}
where (a) applies $u^{2L+1}=\frac{1}{\sqrt{\delta}}\hR_{:,1}^0$, $v^s=\frac{1}{\sqrt{\delta}}B_{:,s}$ for $s\in[L]$, $v^s=\frac{1}{\sqrt{\delta}}\hB^0_{:,s}$ for $s\in[L+1:2L]$, and
\begin{align*}
    u^{2L+1}
    &=f_{2L+1}^u(e^1,\dots,e^{2L},c^1,\dots,c^{L_\Psi})
    =\frac{1}{\sqrt{\delta}}g_{0,1}(\Theta^0,q(\Theta,\tPsi))
    =\frac{1}{\sqrt{\delta}}\tg_{0,1}(\Theta,\Theta^0,\tPsi) \\
    \implies
    \bu^{2L+1}
    &=f_{2L+1}^u(\be^1,\dots,\be^{2L},\bc^1,\dots,\bc^{L_\Psi})
    =\frac{1}{\sqrt{\delta}}\tg_{0,1}(Z, Z^0,\bar{\Psi}).
\end{align*}
The equality (b) applies \eqref{eq:exp_form} and $B_{:,1}^1=\tX\hR_{:,1}^0-\{\hB^0(C^{0,1})^\top\}_{:,1}$.  Next, using \eqref{eq:fufv_choice4} we have
\begin{align*}
    v^{2L+1}&=f_{2L+1}^v(h^1,\dots,h^{2L+1},d^1,\dots,d^{2L})=0 \, ,   \\
    e^{2L+1}&=\tX v^{2L+1}-\sum_{s=1}^{2L+1}a^{2L+1}_su^s\stackrel{(a)}{=}0  \, ,\\
    u^{2L+2}&=f_{2L+2}^u(e^1,\dots,e^{2L+1},c^1,\dots,c^{L_\Psi})
    :=\frac{1}{\sqrt{\delta}}\big\{g_{0}(\Theta^0,q(\Theta,\tPsi))\big\}_{:,2}
    =\frac{1}{\sqrt{\delta}}\hR_{:,2}^0 \, ,
\end{align*}
where (a) applies $v^{2L+1}=0$ and $a^{2L+1,s}=0$ (because $f_{2L+1}^v=0$). A similar derivation follows for $t=2L+2,\dots,3L-1$ to give
\begin{align*}
    h^t=B^1_{:,t-2L}-\{B(\Mu_B^k)^\top\}_{:,t-2L};\quad
    v^t=0;\quad
    e^t=0;\quad
    u^{t+1}=\frac{1}{\sqrt{\delta}}\hR_{:,t+1-2L}^0.
\end{align*}
This completes the derivation of \eqref{eq:abs_AMP_reduction5}. For $t=3L$, we have
\begin{align*}
    h^{3L}&=\sqrt{\delta}\tX^\top u^{3L}-\sum_{s=1}^{3L-1}b^{3L}_sv^s \\
    &=\sqrt{\delta}\tX^\top u^{3L}-\sum_{s=1}^{2L}b^{3L}_sv^s-\sum_{s=2L+1}^{3L-1}b^{3L}_s\underbrace{v^s}_{=0} \\
    &=B_{:,L}^1-\{B(\Mu_B^k)^\top\}_{:,L},
\end{align*}
where the final equality applies \eqref{eq:exp_form} and $B_{:,L}^1=\tX\hR_{:,L}^0-\{\hB^0(C^{0,1})^\top\}_{:,L}$. Next, we define both $f_{3L}^v$ and $f_{3L+1}^u$ to be 0, so that
\begin{align*}
    v^{3L}&=f_{3L}^v(h^1,\dots,h^{3L},d^1,\dots,d^{2L})=0 \\
    e^{3L}&=\tX \underbrace{v^{3L}}_{=0}-\sum_{s=1}^{3L}\underbrace{a^{3L}_s}_{=0}u^s=0 \\
    u^{3L+1}&=f_{3L+1}^u(e^1,\dots,e^{3L},c^1,\dots,c^{L_\Psi})=0.
\end{align*}
This completes the derivation of \eqref{eq:abs_AMP_reduction6}. 

For $t=3L+1, \ldots, 4L-1$, we define
\begin{align}
    \begin{split}
        & f_{t}^v(h^1,\dots,h^{2L},\underbrace{h^{2L+1},\dots,h^{3L}}_{B^1-B(\Mu_B^1)^\top},h^{3L+1}, \ldots, h^t,\underbrace{d^1,\dots,d^{2L}}_{B,\hB^0})  =\frac{1}{\sqrt{\delta}}\{f_1(B^1)\}_{:,t-3L}, \\ 
        &f_{t+1}^u(e^1, \ldots, e^t, c^1,\dots,c^{L_\Psi})=0.
    \end{split}
    \label{eq:fufv_choice5}
\end{align}
In the first definition above,  we note that $\Mu_B^1$ is a constant matrix.  
Using this, for $t=3L+1$  we obtain
\begin{align*}
    h^{3L+1}&=\sqrt{\delta}\tX^\top \underbrace{u^{3L+1}}_{=0}-\sum_{s=1}^{3L}\underbrace{b^{3L+1}_s}_{=0}v^s=0 \\
    v^{3L+1}&=f_{3L+1}^v(h^1,\dots,h^{2L},\underbrace{h^{2L+1},\dots,h^{3L}}_{B^1-B(\Mu_B^1)^\top},h^{3L+1},\underbrace{d^1,\dots,d^{2L}}_{B,\hB^0})
    =\frac{1}{\sqrt{\delta}}\{f_1(B^1)\}_{:,1}
    =\frac{1}{\sqrt{\delta}}\hB^1_{:,1},
\end{align*}
At this juncture, we pause to write the first column of  $\{\hR^0(F^{1,1})^\top\}$ (recall that this term comes from the matrix-AMP algorithm \eqref{eq:GAMP_deterministic}) as follows:
\begin{align*}
    \{\hR^0(F^{1,1})^\top\}_{:,1}
    &=\hR^0\cdot\frac{1}{\delta}
    \begin{bmatrix}
        \E[\tpartial_1f_{1,1}(\Mu_B^1\bar{B}+G_B^1)] \\
        \vdots \\
        \E[\tpartial_Lf_{1,1}(\Mu_B^1+G_B^1)]
    \end{bmatrix}
    =\sum_{l=1}^L\hR_{:,l}^0\cdot\frac{1}{\delta}\E\big[\tpartial_lf_{1,1}(\Mu_B^1\bar{B}+G_B^1)\big],
\end{align*}
where $f_{1,1}$ refers to the first entry of the output of $f_1$, and $\tpartial_l \, f_{1,1}$ means we are taking derivative w.r.t.~the $l$-th entry of the input of $f_{1,1}$, for $l \in [L]$.
%i.e., $\tpartial_1f_{1,1}$ is taking derivative w.r.t.~$\big\{\Mu_B^1\bar{B}+G_B^1\big\}_1$. 
We then have
\begin{align*}
    e^{3L+1}&=\sqrt{\delta}\tX v^{3L+1}-\sum_{s=1}^{2L}a^{3L+1}_s\underbrace{u^s}_{=0}
    -\sum_{s=2L+1}^{3L}a^{3L+1}_su^s-a_{3L+1}^{3L+1}\underbrace{u_{3L+1}}_{=0} \\
    &=\sqrt{\delta}\tX v^{3L+1}-\sum_{s=2L+1}^{3L}\E[\partial_sf_{3L+1}^v(\bh^1,\dots,\bh^{3L+1},\bd^1,\dots,\bd^{2L})]u^s \\
    &\stackrel{(a)}{=}\tX\hB^1_{:,1}-\sum_{l=1}^L\frac{1}{\delta}\E[\partial_lf_{1,1}(\Mu_B^1\bar{B}+G_B^1)]\hR_{:,l}^0 \\
    &=\tX\hB_{:,1}^1-\{\hR^0(F^{1,1})^\top\}_{:,1}
    =\Theta_{:,1}^1,
\end{align*}
where the switch from $\partial_s$ to $\partial_l$ in (a) is via the chain rule of differentiation. From \eqref{eq:fufv_choice5}, we also have
\begin{align*}
    u^{3L+2}&=f_{3L+2}^u(e^1,\dots,e^{3L},c^1,\dots,c^{L_\Psi})=0.
\end{align*}
A similar derivation follows for $t=3L+2,\dots,4L-1$ to give
\begin{align*}
    h^t=0;\quad 
    v^t=\frac{1}{\sqrt{\delta}}\hB_{:,t-3L}^1;\quad
    e^t=\Theta_{:,t-3L}^1;\quad
    u^{t+1}=0.
\end{align*}
This completes the derivation of \eqref{eq:abs_AMP_reduction7}. For $t=4L$, we define
\begin{align}
    \begin{split}
        & f_{4L}^v(h^1,\dots,h^{2L},\underbrace{h^{2L+1},\dots,h^{3L}}_{B^1-B(\Mu_B^1)^\top},h^{3L+1}, \ldots, h^{4L},\underbrace{d^1,\dots,d^{2L}}_{B,\hB^0})  =\frac{1}{\sqrt{\delta}}\{f_1(B^1)\}_{:,L}, \\ 
        &f_{4L+1}^u(\underbrace{e^1,\dots,e^L}_{\Theta},e^{L+1},\dots,e^{3L},\underbrace{e^{3L+1},\dots,e^{4L}}_{\Theta^1},\underbrace{c^1,\dots,c^{L_\Psi}}_{\tPsi}) =\frac{1}{\sqrt{\delta}}g_{1,1}(\Theta^1,q(\Theta,\tPsi)).
    \end{split}
    \label{eq:fufv_choice6}
\end{align}
Using this, we obtain
\begin{align*}
    & h^{4L} =\sqrt{\delta}\tX^\top\underbrace{u^{4L}}_{=0}-\sum_{s=1}^{4L-1}\underbrace{b^{4L}_s}_{=0}v^s=0 \, , \\
    & v^{4L} =f_{4L}^v(h^1,\dots,h^{4L},d^1,\dots,d^{2L})
    = \frac{1}{\sqrt{\delta}}\hB_{:,L}^1 \, , \\
    & e^{4L} =\sqrt{\delta}\tX v^{4L}-\sum_{s=1}^{4L}a^{4L}_su^s=\Theta_{:,L}^1 \, , \\
    & u^{4L+1}
    =f_{4L+1}^u(e^1, \ldots, e^{4L}, c^1,\dots,c^{L_\Psi} ) =\frac{1}{\sqrt{\delta}}g_{1,1}(\Theta^1,q(\Theta,\tPsi))
    =\frac{1}{\sqrt{\delta}}\hR_{:,1}^1 \, ,
\end{align*}
giving us \eqref{eq:abs_AMP_reduction8}. This concludes the reduction of the matrix-AMP iterates to the abstract AMP iterates for the case of $k=0$. 

We now  reduce the matrix-AMP SE parameters to the abstract AMP SE parameters. Recall from Theorem \ref{thm:abs_AMP} that $(e^1,\dots,e^{4L})  \stackrel{W_2}{\rightarrow}(\be^1,\dots,\be^{4L})\sim\normal(0,\Gamma^{4L})$, where $\Gamma^{4L}=\big(\E[\bv^r\bv^s]\big)_{r,s=1}^{4L}$.  
Since we have shown in \eqref{eq:abs_AMP_reduction1}-\eqref{eq:abs_AMP_reduction3}  that $(v^1,\dots,v^L)=\frac{1}{\sqrt{\delta}}B$ and $(v^{L+1},\dots,\\v^{2L})=\frac{1}{\sqrt{\delta}}\hB^0$, from Assumption \textbf{(A1)} of Theorem \ref{thm:GAMP}, we must have that $(\bv^1,\dots,\bv^L)=\frac{1}{\sqrt{\delta}}\bar{B}$ and $(\bv^{L+1},\dots,\bv^{2L})=\frac{1}{\sqrt{\delta}}\bar{B}^0$.
Furthermore, Assumption \textbf{(A1)} guarantees that the top left $2L \times 2L$ submatrix of $\Gamma^{4L}$ equals $\Sigma^0$, i.e., \[ \Gamma^{4L}_{[2L],[2L]} = \big(\E[\bv^r\bv^s]\big)_{r,s=1}^{2L} = \Sigma^0.
\]
Letting $Z \equiv (\be^1,\dots,\be^L)$ and $Z^0 \equiv (\be^{L+1},\dots,\be^{2L})$, we have that $(Z,Z^0)\sim\normal(0,\Sigma^0)$. Since we have shown in \eqref{eq:abs_AMP_reduction0}-\eqref{eq:abs_AMP_reduction4} that $(c^1, \ldots, c^{L_\Psi}) = \tPsi$,
$(e^1,\dots,e^L)=\Theta$ and $(e^{L+1},\dots,e^{2L})=\Theta^0$, by applying Theorem \ref{thm:abs_AMP} to this collection of vectors we obtain 
\begin{align}
\label{eq:Theta0_proved}
    (\tPsi, \Theta, \Theta^0) \stackrel{W_2}{ \to} (\bar{\Psi}, Z, Z^0) 
\stackrel{d}{=}(\bar{\Psi}, Z,\Mu^{0}_{\Theta}Z+G^0_\Theta),
\end{align}
where the equality in distribution follows from \eqref{eq:ZZk_joint}. 

Next, recall from Theorem \ref{thm:abs_AMP} that the joint empirical distribution of $(h^1,\dots,h^{4L})$ converges to the law of $(\bh^1,\dots,\bh^{4L})\sim\normal(0,\Omega^{4L})$, where $\Omega^{4L} = \big(\delta \E[\bu^r\bu^s]\big)_{r,s=1}^{4L+1}$. Since we have shown in \eqref{eq:abs_AMP_reduction5} that
\begin{align*}
    (u^{2L+1},\dots,u^{3L})
    &=\frac{1}{\sqrt{\delta}}\hR^0
    =\frac{1}{\sqrt{\delta}}g_0\big(\Theta^0,q(\Theta,\tPsi)\big)
    =\frac{1}{\sqrt{\delta}}\tg_0\big(\Theta,\Theta^0,\tPsi\big),
\end{align*}
and also that $(\tPsi, \Theta, \Theta^0) \stackrel{W_2}{ \to} (\bar{\Psi}, Z, Z^0)$, we  have $(\bu^{2L+1},\dots,\bu^{3L})=\frac{1}{\sqrt{\delta}}\tg_0(Z,Z^0,\bar{\Psi})$. Using this,  we obtain
\begin{align*}
    \Omega^{4L}_{[2L+1:3L],[2L+1:3L]} = \big(\delta \E[\bu^r\bu^s]\big)_{r,s=2L+1}^{3L}=\Tau_B^{1},
\end{align*}
where the last equality follows from \eqref{eq:SE_Tk1B}. Let $G_B^1 \equiv (\bh^{2L+1},\dots,\bh^{3L}) \sim \normal(0, \Tau_B^{1})$, and recall that $h^{2L+1},\dots,h^{3L}=B^1-B(\Mu_B^1)^\top$, \, $(d^1, \ldots, d^L) = B$, and $ (v^{L+1}, \ldots, v^{2L})= \hB^0$. Applying Theorem \ref{thm:abs_AMP} to this collection of vectors, it follows that 
\begin{align}
    (B, \hB^0, B^1) \stackrel{W_2}{ \to} (\bar{B}, \bar{B}^0, \Mu_B^1 \bar{B} + G_B^1).
    \label{eq:B1_proved}
\end{align}
Using $B^1\stackrel{W_2}{\rightarrow}\Mu_B^1\bar{B}+G_B^1$ (from above) and     
$\hB^1
    =f_1(B^1)$ with $f_1$ satisfying the polynomial growth condition in \eqref{eq:poly_growth_cond}, we have
\begin{align}
    (B, \hB^1)
    \stackrel{W_2}{\rightarrow}
    (\bar{B}, f_1(\Mu_B^1\bar{B}+G_B^1)),
    \label{eq:B_hat_1_convergence}
\end{align}
via \eqref{eq:sum_to_exp}. 
Next, we define
\begin{align*}
    \mathcal{I}_1
    =[L]\cup[3L+1:4L].
\end{align*}
Recall from Theorem \ref{thm:abs_AMP} that $(e^1,\dots,e^{4L})   \stackrel{W_2}{\rightarrow} (\be^1,\dots,\be^{4L})\sim\normal(0,\Gamma^{4L})$, where $\Gamma^{4L}=\big(\E[\bv^r\bv^s]\big)_{r,s=1}^{4L}$.  
Since we have shown in \eqref{eq:abs_AMP_reduction1}-\eqref{eq:abs_AMP_reduction2}  that $(v^1,\dots,v^L)=\frac{1}{\sqrt{\delta}}B$, and in \eqref{eq:abs_AMP_reduction7}-\eqref{eq:abs_AMP_reduction8} that $(v^{3L+1},\dots,v^{4L})=\frac{1}{\sqrt{\delta}}\hB^1$, from \eqref{eq:B_hat_1_convergence} and Theorem \ref{thm:abs_AMP}, we must have that $(\bv^1,\dots,\bv^L)=\frac{1}{\sqrt{\delta}}\bar{B}$ and $(\bv^{3L+1},\dots,\bv^{4L})=\frac{1}{\sqrt{\delta}}f_1(\Mu_B^1\bar{B}+G_B^1)$. This guarantees that
$$
\Gamma_{\mathcal{I}_1,\mathcal{I}_1}^{4L}
=\big(\E[\bv^r\bv^s]\big)_{r,s\in\mathcal{I}_1}
=\Sigma^1.
$$
Letting $Z^1 \equiv (\be^{3L+1},\dots,\be^{4L})$, we have that $(Z,Z^1)\sim\normal(0,\Sigma^1)$. Since we have shown in \eqref{eq:abs_AMP_reduction0}-\eqref{eq:abs_AMP_reduction2} that $(c^1, \ldots, c^{L_\Psi}) = \tPsi$,
$(e^1,\dots,e^L)=\Theta$, and in \eqref{eq:abs_AMP_reduction7}-\eqref{eq:abs_AMP_reduction8} that $(e^{3L+1},\dots,e^{4L})=\Theta^1$, by applying Theorem \ref{thm:abs_AMP} to this collection of vectors we obtain 
\begin{align}
(\tPsi, \Theta, \Theta^1) \stackrel{W_2}{ \to} (\bar{\Psi}, Z, Z^1) 
\stackrel{d}{=}(\bar{\Psi}, Z,\Mu^{1}_{\Theta}Z+G^1_\Theta),
      \label{eq:Theta1_proved}
\end{align} 
where the equality in distribution follows from \eqref{eq:ZZk_joint}. Equation \eqref{eq:B1_proved} proves the first claim of Theorem \ref{thm:GAMP} for $k=0$, and  \eqref{eq:Theta0_proved} and \eqref{eq:Theta1_proved} prove the second claim for $k=0,1$.

\textbf{Inductive hypothesis:} Assume that we have a reduction for the matrix-AMP up to iteration $(k-1)$, for $k > 1$. In words, this means that we can reduce $B^k$, $\hR^{k-1}$, $\hB^k$, and $\Theta^k$ to abstract AMP iterates in iterations $t=2kL+1,\dots,(2k+2)L$. In formulas, this means that we have the following:
\begin{itemize}
    \item For $t=2kL+1,\dots,(2k+1)L-1$: 
    \begin{align*}
        & \Big(h^{2kL+1},\dots,h^{(2k+1)L-1}\Big)
        =\Big(B^k_{:,1}-\{B(\Mu_B^k)^\top\}_{:,1},\dots,B^k_{:,L-1}-\{B(\Mu_B^k)^\top\}_{:,L-1}\Big) \, ;  \\
        & v^{2kL+1},\dots,v^{(2k+1)L-1}
        =0 \, ;\\
        & e^{2kL+1},\dots,e^{(2k+1)L-1}
        =0 \, ; \quad 
        \Big(u^{2kL+2},\dots,u^{(2k+1)L}\Big)
        =\frac{1}{\sqrt{\delta}}\Big(\hR^{k-1}_{:,2},\dots,\hR^{k-1}_{:,L}\Big).
    \end{align*}
    \item For $t=(2k+1)L$: 
    \begin{align*}
        h^{(2k+1)L}=B^k_{:,L}-\{B(\Mu_B^k)^\top\}_{:,L};\quad
        v^{(2k+1)L}=0; \quad
        e^{(2k+1)L}=0;\quad
        u^{(2k+1)L+1}=0.
    \end{align*}
    \item For $t=(2k+1)L+1,\dots,(2k+2)L-1$: 
    \begin{align*}
        & h^{(2k+1)L+1},\dots,h^{(2k+2)L-1} =0 ;  \ \ 
        \Big(v^{(2k+1)L+1},\dots,v^{(2k+2)L-1}\Big) = \frac{1}{\sqrt{\delta}}\Big(\hB^k_{:,1}\dots,\hB^k_{:,L-1}\Big)\\
        & \Big(e^{(2k+1)L+1},\dots,e^{(2k+2)L-1}\Big)=\Big(\Theta^k_{:,1},\dots,\Theta^k_{:,L-1}\Big) \, ; \quad 
        u^{(2k+1)L+2},\dots,u^{(2k+2)L}=0.
    \end{align*}
    \item For $t=(2k+2)L$: 
    \begin{align*}
        h^{(2k+2)L}=0;\quad
        v^{(2k+2)L}=\frac{1}{\sqrt{\delta}}\hB^k_{:,L};\quad
        e^{(2k+2)L}=\Theta^k_{:,L};\quad
        u^{(2k+2)L+1}=\frac{1}{\sqrt{\delta}}\hR^k_{:,1}.
    \end{align*}
\end{itemize}
We now reduce the matrix-AMP SE parameters to the abstract AMP SE parameters. Define the index sets 
\[ \mathcal{I}_{k} = [L]\cup\left\{\bigcup_{r=0}^k[(2r+1)L+1:(2r+2)L]\right\}, 
\quad 
\mathcal{J}_{k} = \bigcup_{r=1}^k[2rL+1:(2r+1)L]. \]
By Theorem \ref{thm:abs_AMP}, we have 
\begin{align*}
    \big(h^{2rL+1},\dots,h^{(2r+1)L}\big)_{r=1}^k
    \stackrel{W_2}{\rightarrow}
    \big(\bar{h}^{2rL+1},\dots,\bar{h}^{(2r+1)L}\big)_{r=1}^k
    \sim\normal(0,\Omega_{\mathcal{J}_k,\mathcal{J}_k}^{(2k+2)L}) \, ,
\end{align*}
where $\Omega_{\mathcal{J}_k,\mathcal{J}_k}^{(2k+2)L}=\{\delta\E[\bar{u}^r\bar{u}^s]\}_{r,s\in\mathcal{J}_k}$. By the inductive hypothesis and Theorem \ref{thm:abs_AMP}, we have
$$
\{u^r\}_{r\in\mathcal{J}_k}
=\frac{1}{\sqrt{\delta}}\Big(\tilde{g}_{0}(\Theta,\Theta^{0},\tPsi),\dots,\tilde{g}_{k-1}(\Theta,\Theta^0,\dots,\Theta^{k-1},\tPsi)\Big)
$$
converging to 
$$
\{\bu^r\}_{r\in\mathcal{J}_k}
=\frac{1}{\sqrt{\delta}}\Big(\tilde{g}_{0}(Z,Z^{0},\bar{\Psi}),\dots,\tilde{g}_{k-1}(Z,Z^0,\dots,Z^{k-1},\bar{\Psi})\Big).
$$
This implies that $\Omega_{\mathcal{J}_k,\mathcal{J}_k}^{(2k+2)L}=\Tau_B^k$. Letting $G_B^k\equiv\big(\bar{h}^{2rL+1},\dots,\bar{h}^{(2r+1)L}\big)_{r=1}^k\sim\normal(0,\Tau_B^k)$. From the induction hypothesis we have $\big(h^{2rL+1},\dots,h^{(2r+1)L}\big)_{r=1}^k=B^1 - B(\Mu_B^1)^\top, \dots, B^k - B(\Mu_B^k)^\top $, and we have already shown that  $(d^1, \ldots, d^L) = B$, and $ (v^{L+1}, \ldots, v^{2L})= \hB^0$. Applying Theorem \ref{thm:abs_AMP} to this collection of vectors, we obtain 
\begin{align}
(B, \hB^0, B^1, \dots, B^k) \stackrel{W_2}{\to} (\bar{B}, \bar{B}^0, \Mu_B^1 \bar{B} + G_B^1, \dots, \Mu_B^k \bar{B} + G_B^k). 
    \label{eq:Bk_proved}
\end{align} 
Using $B^k\stackrel{W_2}{\rightarrow}\Mu_B^k\bar{B}+G_B^k$ (from above) and $\hB^k
    =f_k(B^k)$ with  $f_k$ satisfying the polynomial growth condition in \eqref{eq:poly_growth_cond}, we have
\begin{align}
    (B, \hB^1, \dots, \hB^k)
    \stackrel{W_2}{\rightarrow}
    \Big( \bar{B}, f_1(\Mu_B^1\bar{B}+G_B^1), \dots, f_k(\Mu_B^1\bar{B}+G_B^1,\dots,\Mu_B^k\bar{B}+G_B^k)\Big),
    \label{eq:B_hat_k_convergence}
\end{align}
via \eqref{eq:sum_to_exp}. 
%This implies that $(v^{(2k+1)L+1},\dots,v^{(2k+2)L})=\frac{1}{\sqrt{\delta}}\hB^k$ converges to 
%$$
%(\bv^{(2k+1)L+1},\dots,\bv^{(2k+2)L})=\frac{1}{\sqrt{\delta}}f_k(\Mu_B^k\bar{B}+G_B^k).
%$$
By Theorem \ref{thm:abs_AMP}, we have
\begin{align*}
  \big(e^t\big)_{t\in\mathcal{I}_k}
  \stackrel{W_2}{\rightarrow}
  \big(\be^t\big)_{t\in\mathcal{I}_k}
   \sim\normal\big(0,\Gamma_{\mathcal{I}_k,\mathcal{I}_k}^{(2k+2)L}\big)
\end{align*}
where $\Gamma_{\mathcal{I}_k,\mathcal{I}_k}^{(2k+2)L}=(\E[\bar{v}^r\bar{v}^s])_{r,s\in\mathcal{I}_k}$.
Since we have shown that $(v^1,\dots,v^L)=\frac{1}{\sqrt{\delta}}B$ and by the inductive hypothesis, we have $(v^{(2k+1)L+1},\dots,v^{(2k+2)L})=\frac{1}{\sqrt{\delta}}\hB^k$, from  \eqref{eq:B_hat_k_convergence} and Theorem \ref{thm:abs_AMP}, we must have that $(\bv^1,\dots,\bv^L)=\frac{1}{\sqrt{\delta}}\bar{B}$ and
$$
\big(\bv^{(2r+1)L+1},\dots,\bv^{(2r+2)L}\big)_{r=1}^k
=\frac{1}{\sqrt{\delta}}\Big(f_1\big(\Mu_B^1\bar{B}+G_B^1\big),\dots,f_k\big(\Mu_B^1\bar{B}+G_B^1,\dots,\Mu_B^k\bar{B}+G_B^k\big)\Big).
$$
This guarantees that
$$
\Gamma_{\mathcal{I}_k,\mathcal{I}_k}^{(2k+2)L}
=\big(\E[\bv^r\bv^s]\big)_{r,s\in\mathcal{I}_k}
=\Sigma^k.
$$
Letting $Z^k \equiv (\be^{(2k+1)L+1},\dots,\be^{(2k+2)L})$, we have that $(Z,Z^0,\dots,Z^k)\sim\normal(0,\Sigma^k)$. Since we have shown that $(c^1, \ldots, c^{L_\Psi}) = \tPsi$,
$(e^1,\dots,e^L)=\Theta$, and have
$$
\big(e^{(2r+1)L+1},\dots,e^{(2r+2)L}\big)_{r=1}^k
=\big(\Theta^1,\dots,\Theta^k\big)
$$
via the inductive hypothesis, by applying Theorem \ref{thm:abs_AMP} to this collection of vectors we obtain 
\begin{align}
    (\tPsi, \Theta, \Theta^0, \dots, \Theta^k) 
    \stackrel{W_2}{ \to} (\bar{\Psi}, Z, Z^0, \dots, Z^k) 
    \stackrel{d}{=}(\bar{\Psi}, Z, \Mu^{0}_{\Theta}Z+G^0_\Theta, \dots, \Mu^{k}_{\Theta}Z+G^k_\Theta).
    \label{eq:Thetak_proved}
\end{align} 
%where the equality in distribution follows from \eqref{eq:ZZk_joint}. 
Thus, under the inductive hypothesis for the reduction, by \eqref{eq:Bk_proved}  the first claim of Theorem \ref{thm:GAMP} holds with $k$ replaced by $(k-1)$, and  by \eqref{eq:Thetak_proved}   the second claim holds for $k$.

\textbf{Induction step:} Here need to show that for iteration $k$ of the matrix-AMP, where $k>1$, $B^{k+1}$, $\hR^k$, $\hB^{k+1}$, $\Theta^{k+1}$ can be reduced to abstract AMP iterates using steps $t=(2k+2)L+1,\dots,(2k+4)L$, and that the corresponding SE matrices match. We provide the summary of the steps below before giving the full derivations:
\begin{itemize}
    \item For $t=(2k+2)L+1,\dots,(2k+3)L-1$: We have
    \begin{align}
    \begin{split}
        & \Big(h^{(2k+2)L+1},\dots,h^{(2k+3)L-1}\Big) \\
        &\qquad\qquad=\Big(B^{k+1}_{:,1}-\{B(\Mu_B^{k+1})^\top\}_{:,1},\dots,B^{k+1}_{:,L-1}-\{B(\Mu_B^{k+1})^\top\}_{:,L-1}\Big) \, ;  \\
        & v^{(2k+2)L+1},\dots,v^{(2k+3)L-1}
        =0 \, ;\\
        & e^{(2k+2)L+1},\dots,e^{(2k+3)L-1}
        =0 \, ; \quad 
        \Big(u^{(2k+2)L+2},\dots,u^{(2k+3)L}\Big)
        =\frac{1}{\sqrt{\delta}}\Big(\hR^{k}_{:,2},\dots,\hR^{k}_{:,L}\Big).
    \end{split}
    \label{eq:abs_AMP_ind_reduction1}
    \end{align}
    \item For $t=(2k+3)L$: We have
    \begin{align}
    \begin{split}
        h^{(2k+3)L}=B^{k+1}_{:,L}-\{B(\Mu_B^{k+1})^\top\}_{:,L};\quad
        v^{(2k+3)L}=0; \quad
        e^{(2k+3)L}=0;\quad
        u^{(2k+3)L+1}=0.
    \label{eq:abs_AMP_ind_reduction2}
    \end{split}
    \end{align}
    \item For $t=(2k+3)L+1,\dots,(2k+4)L-1$, we have
    \begin{align}
    \begin{split}
        & h^{(2k+3)L+1},\dots,h^{(2k+4)L-1} =0 ; \\
        &\Big(v^{(2k+3)L+1},\dots,v^{(2k+4)L-1}\Big) = \frac{1}{\sqrt{\delta}}\Big(\hB^{k+1}_{:,1}\dots,\hB^{k+1}_{:,L-1}\Big) \\
        & \Big(e^{(2k+3)L+1},\dots,e^{(2k+4)L-1}\Big)=\Big(\Theta^{k+1}_{:,1},\dots,\Theta^{k+1}_{:,L-1}\Big) \, ; \quad 
        u^{(2k+3)L+2},\dots,u^{(2k+4)L}=0.
    \end{split}
    \label{eq:abs_AMP_ind_reduction3}  
    \end{align}
    \item For $t=(2k+4)L$, we have
    \begin{align}
    \begin{split}
        h^{(2k+4)L}=0;\quad
        v^{(2k+4)L}=\frac{1}{\sqrt{\delta}}\hB^{k+1}_{:,L};\quad
        e^{(2k+4)L}=\Theta^{k+1}_{:,L};\quad
        u^{(2k+4)L+1}=\frac{1}{\sqrt{\delta}}\hR^{k+1}_{:,1}.
    \end{split}
    \label{eq:abs_AMP_ind_reduction4}  
    \end{align}
\end{itemize}
We proceed to provide the derivations of \eqref{eq:abs_AMP_ind_reduction1}--\eqref{eq:abs_AMP_ind_reduction4}. For $t=(2k+2)L+1,\dots,(2k+3)L-1$, choose
\begin{align*}
    f_t^v(h^1,\dots,h^t,d^1,\dots,d^{2L})
    &=0 \ \text{ and } \\
    f_{t+1}^u(e^1,\dots,e^t,c^1,\dots,c^{L_\Psi})
    &=\frac{1}{\sqrt{\delta}}\Big\{g_{k}\big(\Theta^0,\dots,\Theta^{k},q(\Theta,\tPsi)\big)\Big\}_{:,t+1-(2k+2)L} \\
    &=\frac{1}{\sqrt{\delta}}\Big\{\tg_{k}\big(\Theta,\Theta^0,\dots,\Theta^{k},\tPsi\big)\Big\}_{:,t+1-(2k+2)L}.
\end{align*} 
We first look at $t=(2k+2)L+1$. From the definition of $C^{k,m+1}$ in \eqref{eq:CkF_k1_det_def}, we have the following identity:
\begin{align}
    \left\{\sum_{m=0}^k\hB^m\big(C^{k,m+1}\big)^\top\right\}_{:,1}
    &=\sum_{m=0}^k\left\{\hB^m\big(C^{k,m+1}\big)^\top\right\}_{:,1}
    =\sum_{m=0}^k\hB^m
    \begin{bmatrix}
        \E[\tpartial_{mL+1} \, g_{k,1}(Z^0,\dots,Z^k,\bar{Y})] \\
        \vdots \\
        \E[\tpartial_{(m+1)L} \, g_{k,1}(Z^0,\dots,Z^k,\bar{Y})]
    \end{bmatrix} \nonumber \\
    &=\sum_{m=0}^k\sum_{l=1}^L\hB^m_{:,l}\cdot\E\Big[\tpartial_{mL+l} \, g_{k,1}(Z^0,\dots,Z^k,\bar{Y})\Big],
    \label{eq:identity1_k}
\end{align}
where $\tpartial_{mL+l} \, g_{k,1}$ denotes the derivative w.r.t.~the $(mL+l)$-th entry of the input of $g_{k,1}$. From the definition of  $f^u_{(2k+2)L+1}$, we have  $u^{(2k+2)L+1}=\frac{1}{\sqrt{\delta}}\hR_{:,1}^{k}$.  We therefore have:
\begin{equation}
    \begin{split}
    & h^{(2k+2)L+1}
    =\tX^\top\hR_{:,1}^k-\sum_{s=1}^{(2k+2)L}\delta\E[\partial_sf_{(2k+2)L+1}^u]v^s  \\
    &\stackrel{(a)}{=}\tX^\top\hR_{:,1}^k-\sum_{s=1}^{L}\delta\E[\partial_sf_{(2k+2)L+1}^u]v^s-\sum_{r=0}^k\sum_{s\in[(2r+1)L+1:(2r+2)L]}\delta\E[\partial_sf_{(2k+2)L+1}^u]v^s  \\
    &=\tX^\top\hR_{:,1}^k-\sum_{l=1}^L\E[\tpartial_l\tg_{k,1}(Z,Z^0,\dots,Z^k,\bar{\Psi})]B_{:,l}
    -\sum_{m=0}^k\sum_{l=1}^L\hB^m_{:,l}\E\Big[\tpartial_{mL+l} \, g_{k,1}(Z^0,\dots,Z^k,\bar{Y})\Big]  \\
    &\stackrel{(b)}{=}B_{:,1}^{k+1}-\big\{B(\Mu_B^{k+1})^\top\big\}_{:,1},
    \label{eq:ind_h_deriv}
    \end{split}
\end{equation}
where (a) uses the fact that the other terms of the summation equals to zero since $f_s^v=0$ for those $s$, and (b) applies the definition in \eqref{eq:SE_Mk1B} and the identity in \eqref{eq:identity1_k}. Next, we have $v^{(2k+2)L+1}=0$ from the choice of $f_{(2k+2)L+1}^v=0$, and 
\begin{align*}
    e^{(2k+2)L+1}
    =\sqrt{\delta}\tX\underbrace{v^{(2k+2)L+1}}_{=0}-\sum_{s=1}^t\underbrace{a_s^{(2k+2)L+1}}_{=0}u^s
    =0.
\end{align*}
Next, we have $u^{(2k+2)L+2}=\frac{1}{\sqrt{\delta}}\hR^k_{:,2}$ from the choice of $f_{(2k+2)L+1}^u$. A similar set of steps will give us the reduction for $t=(2k+2)L+2,\dots,(2k+3)L-1$. This proves \eqref{eq:abs_AMP_ind_reduction1}.

For $t=(2k+3)L$, choose $f_t^v$ and $f_{t+1}^u$ to be zero. Using the steps provided in \eqref{eq:ind_h_deriv} we obtain $h^t$ in \eqref{eq:abs_AMP_ind_reduction2}. The remaining reductions can be obtained by direct substitution of $f_t^v=0$ and $f_{t+1}^u=0$. This proves \eqref{eq:abs_AMP_ind_reduction2}.

For $t=(2k+3)L+1,\dots,(2k+4)L-1$, choose
\begin{align*}
    f_t^v(h^1,\dots,h^t,d^1,\dots,d^{2L})
    &=\frac{1}{\sqrt{\delta}}\big\{f_{k+1}(B^{k+1})\big\}_{:,t-(2k+3)L} \, \text{ and } \\
    f_{t+1}^u(e^1,\dots,e^t,c^1,\dots,c^{L_\Psi})
    &=0.
\end{align*}
For $t=(2k+3)L+1$, we have
\begin{align*}
    h^{(2k+3)L+1}
    =\sqrt{\delta}\tX^\top \underbrace{u^{(2k+3)L+1}}_{=0}-\sum_{s=1}^{(2k+3)L}\underbrace{b_s^{(2k+3)L+1}}_{=0}v^s
    =0.
\end{align*}
 Let $\bar{B}^m:=\Mu_B^m\bar{B}+G_B^m$.  From the definition of $F^{k+1,m+1}$ in 
\eqref{eq:CkF_k1_det_def}, we have the following identity:
\begin{align}
    \left\{\sum_{m=0}^k\hR^m\big(F^{k+1,m+1}\big)^\top\right\}_{:,1}
    & %=\sum_{m=0}^k\big\{\hR^m\big(F^{k+1,m+1}\big)^\top\big\}_{:,1}
    =\sum_{m=0}^k\hR^m\frac{1}{\delta}
    \begin{bmatrix}
        \E[\tpartial_{mL+1} \, f_{k+1,1}(\bar{B}^1,\dots,\bar{B}^{k+1})] \\
        \vdots \\
        \E[\tpartial_{(m+1)L} \, f_{k+1,1}(\bar{B}^1,\dots,\bar{B}^{k+1})] \\
    \end{bmatrix} \nonumber \\
    &=\sum_{m=0}^k\sum_{l=1}^L\hR_{:,l}^m\cdot
    \frac{1}{\delta}\E[\tpartial_{mL+l} \, f_{k+1,1}(\bar{B}^1,\dots,\bar{B}^k)],
    \label{eq:identity2_k}
\end{align}
where $f_{k+1,1}$ denotes the first entry of $f_{k+1}$, and $\tpartial_{mL+l} \, f_{k+1,1}$ denotes the derivative w.r.t.~the $(mL+l)$-th entry of the input of $f_{k+1,1}$. From the definition of $f_{(2k+3)L+1}^v$, we have that $v^{(2k+3)L+1}=\frac{1}{\sqrt{\delta}}\hB_{:,1}^{k+1}$.  We therefore have:
\begin{align}
    e^{(2k+3)L+1}
    &=\tX\hB_{:,1}^{k+1}-\sum_{s=1}^{(2k+3)L+1}\E[\partial_sf_{(2k+3)L+1}^v]u^s \nonumber \\
    &\stackrel{(a)}{=}\tX\hB_{:,1}^{k+1}-\sum_{r=1}^{k+1}\sum_{s\in[2rL+1:(2r+1)L]}\E[\partial_sf_{(2k+3)L+1}^v]u^s \nonumber \\
    &=\tX\hB_{:,1}^{k+1}-\sum_{m=0}^k\sum_{l=1}^L\hR_{:,l}^m\frac{1}{\delta}\E[\tpartial_{mL+l}f_{k+1,1}(\bar{B}^1,\dots,\bar{B}^k)] \nonumber \\
    &\stackrel{(b)}{=}\tX\hB_{:,1}^{k+1}-\left\{\sum_{m=0}^k\hR^m\big(F^{k+1,m+1}\big)^\top\right\}_{:,1}
    =\Theta_{:,1}^{k+1},
    \label{eq:Theta_k+1_reduction}
\end{align}
where (a) uses the fact that the other terms of the summation equals to zero since $f_s^v=0$ for those $s$, and (b) applies the identity in \eqref{eq:identity2_k}. A similar set of steps will give us the reduction for $t=(2k+3)L+2,\dots,(2k+4)L-1$. This proves \eqref{eq:abs_AMP_ind_reduction3}.

For $t=(2k+4)L$, choose
\begin{align*}
    f_t^v(h^1,\dots,h^t,d^1,\dots,d^{2L})
    &=\frac{1}{\sqrt{\delta}}\big\{f_{k+1}(B^1,\dots,B^{k+1})\big\}_{:,L}  \ \text{ and } \\
    f_{t+1}^u(e^1,\dots,e^t,c^1,\dots,c^{L_\Psi})
    &=\frac{1}{\sqrt{\delta}}\big\{g_{k+1}(\Theta^0,\dots,\Theta^{k+1},q(\Theta,\tPsi))\big\}_{:,1}.
\end{align*}
By similar techniques displayed above in \eqref{eq:Theta_k+1_reduction} and by direct substitution of $f_t^v$ and $f_{t+1}^u$, we get
\begin{align*}
    h^{(2k+4)L}=0;\quad
    v^{(2k+4)L}=\frac{1}{\sqrt{\delta}}\hB_{:,L}^{k+1};\quad
    e^{(2k+4)L}=\Theta_{:,L}^{k+1}; \quad
    u^{(2k+4)L+1}=\frac{1}{\sqrt{\delta}}\hR_{:,1}^{k+1}.
\end{align*}
This proves \eqref{eq:abs_AMP_ind_reduction4}.

We now reduce the matrix-AMP SE parameters to those of the abstract AMP. Define the index sets 
\[
\mathcal{I}_{k+1} = [L]\cup\left\{\bigcup_{r=0}^{k+1}[(2r+1)L+1:(2r+2)L]\right\}, \quad 
\mathcal{J}_{k+1} = \bigcup_{r=1}^{k+1}[2rL+1:(2r+1)L]. 
\]
By Theorem \ref{thm:abs_AMP}, we have 
\begin{align*}
  (h^{(2rL+1},\dots,h^{(2k+1)L})_{r=1}^{k+1} \stackrel{W_2}{\to}  (\bar{h}^{(2rL+1},\dots,\bar{h}^{(2r+1)L})_{r=1}^k
  \sim\normal\big(0,\Omega_{\mathcal{J}_{k+1},\mathcal{J}_{k+1}}^{(2k+4)L}\big) \, ,
\end{align*}
where $\Omega_{\mathcal{J}_{k+1},\mathcal{J}_{k+1}}^{(2k+4)L}=\{\delta\E[\bar{u}^r\bar{u}^s]\}_{r,s\in\mathcal{J}_{k+1}}$. By the inductive hypothesis and Theorem \ref{thm:abs_AMP}, we have 
$$
\{u^r\}_{r\in\mathcal{J}_{k+1}}
=\frac{1}{\sqrt{\delta}}\Big(\tg(\Theta,\Theta^0,\tPsi),\dots,\tg_{k}(\Theta,\Theta^0,\dots,\Theta^{k},\tPsi)\Big)
$$ 
converging to 
$$
\{\bu^r\}_{r\in\mathcal{J}_{k+1}}
=\frac{1}{\sqrt{\delta}}\Big(\tg_0(Z,Z^0,\bar{\Psi}),\dots,\tg_{k}(Z,Z^0,\dots,Z^{k},\bar{\Psi})\Big).
$$
This implies that $\Omega_{\mathcal{J}_{k+1},\mathcal{J}_{k+1}}^{(2k+2)L}=\Tau_B^{k+1}$. Letting $G_B^{k+1}\equiv\big(\bar{h}^{2rL+1},\dots,\bar{h}^{(2r+1)L}\big)_{r=1}^{k+1}\sim\normal(0,\Tau_B^{k+1})$. From the induction hypothesis we have $\big(h^{2rL+1},\dots,h^{(2r+1)L}\big)_{r=1}^{k+1}=B^1-B(\Mu_B^1)^\top,\dots,B^{k+1} - B(\Mu_B^{k+1})^\top $, and we have already shown that  $(d^1, \ldots, d^L) = B$, and $ (v^{L+1}, \ldots, v^{2L})= \hB^0$. Applying Theorem \ref{thm:abs_AMP} to this collection of vectors, we obtain 
\begin{align}
(B, \hB^0, B^1, \dots, B^{k+1}) \stackrel{W_2}{\to} (\bar{B}, \bar{B}^0, \Mu_B^{1} \bar{B} + G_B^{1}, \dots, \Mu_B^{k+1} \bar{B} + G_B^{k+1}).
\label{eq:Bk1_proved}
\end{align} 
Using $B^{k+1}\stackrel{W_2}{\rightarrow}\Mu_B^{k+1}\bar{B}+G_B^{k+1}$ (from above) and $\hB^{k+1}=f_{k+1}(B^1,\dots,B^{k+1})$ with  $f_{k+1}$ satisfying the polynomial growth condition in \eqref{eq:poly_growth_cond}, we have
\begin{align}
    (B,\hB^1,\dots,\hB^{k+1})
    \stackrel{W_2}{\rightarrow}
    \Big(\bar{B},f_{1}(\Mu_B^{1}\bar{B}+G_B^{1}),\dots,f_{k+1}(\Mu_B^{1}\bar{B}+G_B^{1},\dots,\Mu_B^{k+1}\bar{B}+G_B^{k+1})\Big),
    \label{eq:B_hat_kplus1_convergence}
\end{align}
via \eqref{eq:sum_to_exp}. 
%This implies that $(v^{(2k+3)L+1},\dots,v^{(2k+4)L})=\frac{1}{\sqrt{\delta}}\hB^{k+1}$ converges to 
%$$
%(\bv^{(2k+3)L+1},\dots,\bv^{(2k+4)L})=\frac{1}{\sqrt{\delta}}f_{k+1}(\Mu_B^{k+1}\bar{B}+G_B^{k+1}).
%$$
By Theorem \ref{thm:abs_AMP}, we have
\begin{align*}
   \big(e^t\big)_{t\in\mathcal{I}_{k+1}}
   \stackrel{W_2}{\rightarrow}
   \big(\be^t\big)_{t\in\mathcal{I}_{k+1}}
   \sim\normal\big(0,\Gamma_{\mathcal{I}_{k+1},\mathcal{I}_{k+1}}^{(2k+4)L}\big),
\end{align*}
where $\Gamma_{\mathcal{I}_{k+1},\mathcal{I}_{k+1}}^{(2k+4)L}=(\E[\bar{v}^r\bar{v}^s])_{r,s\in\mathcal{I}_{k+1}}$.
Since we have 
$$
(v^1,\dots,v^L)=\frac{1}{\sqrt{\delta}}B
\text{ and }
\big(v^{(2r+1)L+1},\dots,v^{(2r+2)L}\big)_{r=1}^{k+1}=\frac{1}{\sqrt{\delta}}\big(\hB^1,\dots,\hB^{k+1}\big),
$$
from \eqref{eq:B_hat_kplus1_convergence} and Theorem \ref{thm:abs_AMP}, we must have that
$$
(\bv^1,\dots,\bv^L)=\frac{1}{\sqrt{\delta}}\bar{B}
\text{ and }
\big(\bv^{(2r+1)L+1},\dots,\bv^{(2r+2)L}\big)_{r=1}^{k+1}=\frac{1}{\sqrt{\delta}}f_{k+1}\big(\Mu_B^{1}\bar{B}+G_B^{1},\dots,\Mu_B^{k+1}\bar{B}+G_B^{k+1}\big).
$$
This guarantees that
$$
\Gamma_{\mathcal{I}_{k+1},\mathcal{I}_{k+1}}^{(2k+4)L}
=\big(\E[\bv^r\bv^s]\big)_{r,s\in\mathcal{I}_{k+1}}
=\Sigma^{k+1}.
$$
Letting $Z^{k+1} \equiv (\be^{(2k+3)L+1},\dots,\be^{(2k+4)L})$, we have that $(Z,Z^0,\dots,Z^{k+1})\sim\normal(0,\Sigma^{k+1})$. Since $(c^1, \ldots, c^{L_\Psi}) = \tPsi$,
$(e^1,\dots,e^L)=\Theta$, and we have $(e^{(2k+3)L+1},\dots,e^{(2k+2)L})=\Theta^{k+1}$ via the inductive hypothesis, by applying Theorem \ref{thm:abs_AMP} to this collection of vectors we obtain 
\begin{align}
    (\tPsi, \Theta, \Theta^0, \dots, \Theta^{k+1}) \stackrel{W_2}{ \to} (\bar{\Psi}, Z, Z^0, \dots, Z^{k+1}) 
    \stackrel{d}{=}(\bar{\Psi}, Z,\Mu^{0}_{\Theta}Z+G^{0}_\Theta,\dots,\Mu^{k+1}_{\Theta}Z+G^{k+1}_\Theta),
\label{eq:Thetak1_proved}
\end{align} 
%where the equality in distribution follows from \eqref{eq:ZZk_joint}. 
Thus, \eqref{eq:Bk1_proved} shows that the first claim of Theorem \ref{thm:GAMP} holds, and \eqref{eq:Thetak1_proved}  shows that the second claim holds for $(k+1)$. This completes the proof of the induction step.

We conclude the proof by observing that the assumptions in Theorem \ref{thm:abs_AMP} are satisfied by those of the matrix-AMP algorithm:
\begin{itemize}
    \item Assumption \ref{ass:universality_result} is satisfied by the model assumptions in Section \ref{sec:prelim} and assumption \textbf{(A1)} of Theorem \ref{thm:GAMP}.
    \item Assumptions 1, 2, and 3 in Theorem \ref{thm:abs_AMP} are directly satisfied by the matrix-AMP algorithm via assumptions \textbf{(A2)} and \textbf{(A3)} of Theorem \ref{thm:GAMP} respectively.
\end{itemize}

\section{Proof of Proposition \ref{prop:eqv_of_AMPs}} \label{sec:eqv_of_AMP_and_SE}

Our goal is to show that a  matrix-AMP algorithm in \eqref{eq:memoryless_GAMP} with memoryless Bayes-optimal denoisers is equivalent to the AMP algorithm proposed by El Alaoui et al. in \cite{Ala18}. The Bayes-optimal denoisers $g_k:\mb{R}^L\times\mb{R}^L\rightarrow\mb{R}^L$ and $f_{k+1}:\mb{R}^L\rightarrow\mb{R}^L$ for\eqref{eq:memoryless_GAMP}  act row-wise on their matrix inputs and are defined as
\begin{align}
\begin{split}
    f_k(s^k)
    &=\E[\bar{B}|\Mu_B^k\bar{B}+G_B^k=s^k] \\
    g_k(z^k,y)
    &=\Cov[Z \mid Z^k=u]^{-1}\left(\E[Z|Z^k=z^k,\bar{Y}=y]-\E[Z|Z^k=z^k]\right),
\end{split}
\label{eq:memoryless_denoisers}    
\end{align}
where we recall from the discussion below \eqref{eq:gk_opt_def} that left multiplication of a positive definite matrix to $(\E[Z|Z^k,\bar{Y}]-\E[Z|Z^k])$ does not affect the optimality of $g_k$.

With some abuse of notation, we reuse the notations of the original state evolution parameters defined in \eqref{eq:SE_Mk1B}-\eqref{eq:Sigma_22_def} for a simplified version described below. For the matrix-AMP in \eqref{eq:memoryless_GAMP} with memoryless denoisers, we use the simplified state evolution parameter $\Sigma^k\in\mb{R}^{2L\times 2L}$, where $(Z,Z^k)^\top\sim\normal(0,\Sigma^k)$ is independent of $\bar{\Psi}\sim P_{\bar{\Psi}}$. Using $\Sigma^k$, state evolution computes:
\begin{align}
    \Mu_B^{k+1}
    &=\E[\partial_1\tg_k(Z,Z^k,\bar{\Psi})]  \quad \in \, \reals^{L \times L},
    \label{eq:memoryless_Mu}   
    \\
    \Tau_B^{k+1}
    &=\E[\tg_k(Z,Z^k,\bar{\Psi})\tg_k(Z,Z^k,\bar{\Psi})^\top]  \quad \in \, \reals^{L \times L}, 
    \label{eq:memoryless_Tau}   
    \\
    \Sigma^{k+1}
    &=\begin{bmatrix}
        \Sigma_{(11)}^{k+1} & \Sigma_{(12)}^{k+1} \\
        \Sigma_{(21)}^{k+1} & \Sigma_{(22)}^{k+1}
    \end{bmatrix}  \quad \in \,  \reals^{2L \times 2L}, 
    \label{eq:memoryless_Sigma}   
\end{align}
where the four $L\times L$ matrices constituting $\Sigma^{k+1}$ are given by:
\begin{align}
    \Sigma_{(11)}^\top
    &=\frac{1}{\delta}\E[\bar{B}\bar{B}^\top] 
    \label{eq:memoryless_Sigma11}   
    \\
    \Sigma_{(12)}^{k+1}
    &=\left(\Sigma_{(21)}^{k+1}\right)^\top
    =\frac{1}{\delta}\E\left[\bar{B}f_{k+1}(\Mu_B^{k+1}\bar{B}+G_B^{k+1})^\top\right] 
    \label{eq:memoryless_Sigma12}   
    \\
    \Sigma_{(22)}^{k+1}
    &=\frac{1}{\delta}\E\left[f_{k+1}(\Mu_B^{k+1}\bar{B}+G_B^{k+1})f_{k+1}(\Mu_B^{k+1}\bar{B}+G_B^{k+1})^\top\right].
    \label{eq:memoryless_Sigma22}   
\end{align}
Under the same assumptions as Theorem \ref{thm:GAMP}, we have Theorem \ref{thm:GAMP} specializing to give us
\begin{align}
    (B,\hB^0,B^{k+1})
    \stackrel{W_2}{\rightarrow}
    (\bar{B},\bar{B}^0,\Mu_B^{k+1}\bar{B}+G_B^{k+1}) 
    \label{eq:memoryless_Bk_conv}   
    \\
    (\tPsi,\Theta,\Theta^k)
    \stackrel{W_2}{\rightarrow}
    (\bar{\Psi},Z,\Mu_\Theta^kZ+G_\Theta^k)
    \label{eq:memoryless_Thetak_conv}   
\end{align}
almost surely as $n,p\rightarrow\infty$ with $n/p\rightarrow\delta$, where $G_B^{k+1}\sim\normal(0,\Tau_B^{k+1})$ is independent of $\bar{B}$, and $G_\Theta^k\sim\normal(0,\Tau_\Theta^k)$ is independent of $(Z,\bar{\Psi})$.

We first review the AMP algorithm and state evolution proposed by El Alaoui et al.~\cite{Ala18} in Section \ref{subsec:ElA_AMP}. In Section \ref{subsec:bayes_denoiser_SE}, we specify the Bayes-optimal denoisers and the resulting state evolution for our matrix-AMP. Then, in Section \ref{subsec:SE_equiv}, we show the equivalence of the two state evolution recursions. Finally, in Section \ref{sec:eqv_of_AMPs} we show how the AMP in \eqref{eq:GAMP} can be obtained from the one in ~\cite{Ala18} via a change of variables and large-sample approximations. 

\subsection{AMP and state evolution of El Alaoui et al.~\cite{Ala18}} \label{subsec:ElA_AMP}

For clarity, we refer to the AMP algorithm of \cite{Ala18} as \AlaouiAMP, and refer to the  AMP studied in this paper (in \eqref{eq:memoryless_GAMP}) as matrix-AMP. 
\AlaouiAMP is written in terms of rescaled matrices $\bX, \bY$ defined as:  
\begin{align}
\begin{split}
    \bX&=\sqrt{\alpha(1-\alpha)\delta}\,\tX
   \  \text{ and } \ 
    \bY=\sqrt{\alpha(1-\alpha)\delta}\, \tY,
\end{split}
\label{eq:bar_X_and_Y}
\end{align}
where $\tX, \tY$ are given by \eqref{eq:X_ij_decomp} and \eqref{eq:tYtPsi_def}.
%We begin by presenting \AlaouiAMP: 
For $k\geq0$, \AlaouiAMP iteratively produces estimates $\hz^{k+1}$ of the signal $B\in\mb{R}^{p\times L}$, via iterates $w^{k+1} \in \reals^{n \times L}$ and $z^{k+1} \in \reals^{p \times L}$. The rows of these iterates, indexed by $i \in [n]$ and $j \in [p]$, are computed as follows:
\begin{align}
\begin{split}
& \hz_{j,:}^{k+1}=\eta(z_{j,:}^{k},\Gamma_j^{k}); \\
    &w_{i,:}^{k+1}
    =\sum_{j=1}^{p}\bX_{ij}\hz^{k+1}_{j,:}-\Xi_i^{k+1}(\Xi_i^{k})^{-1}(\bY_{i,:}-w_{i,:}^{k}); \quad
    \Xi_i^{k+1}
    =\sum_{j=1}^p\bX_{ij}^2\Big(\diag(\hz_{j,:}^{k+1})-\hz_{j,:}^{k+1}(\hz_{j,:}^{k+1})^\top\Big) \\
    &z_{j,:}^{k+1}
    =\hz_{j,:}^{k+1}+\Gamma_j^{k+1}\sum_{i=1}^n\bX_{ij}(\Xi_i^{k+1})^{-1}(\bY_{i,:}-w_{i,:}^{k+1}); \quad
    \Gamma_j^{k+1}
    =\left(\sum_{i=1}^n\bX_{ij}^2\big(\Xi_i^{k+1}\big)^{-1}\right)^{-1},
\end{split}
\label{eq:Alaoui_AMP}
\end{align}
where 
\begin{align}
    \eta(z_{j,:}^{k},\Gamma_j^{k})
    :=\frac{\sum_{l=1}^L\pi_le_l\exp[-\frac{1}{2}(z_{j,:}^{k}-e_l)^\top(\Gamma_j^{k})^{-1}(z_{j,:}^{k}-e_l)]}{\sum_{l=1}^L\pi_l\exp[-\frac{1}{2}(z_{j,:}^{k}-e_l)^\top(\Gamma_j^{k})^{-1}(z_{j,:}^{k}-e_l)]}
    \in\mb{R}^L.
    \label{eq:Ala_AMP_fk}
\end{align}
The algorithm is initialized with $\hz_{j,:}^{1} = \pi$, for $j \in [p]$.

\paragraph{State evolution for \AlaouiAMP }
For $k\geq0$, El-Alaoui et al.~\cite{Ala18} introduced the following $L \times L$ matrices to measure the performance of \AlaouiAMP:
\begin{align}
    M_p^{k+1}
    =\frac{1}{p}\sum_{j=1}^p\hz_{j,:}^{k+1}B_{j,:}^\top\in\mb{R}^{L\times L}, \ 
    \text{ and } \
    Q_p^{k+1}
    =\frac{1}{p}\sum_{j=1}^p\hz_{j,:}^{k+1}\big(\hz_{j,:}^{k+1}\big)^\top\in\mb{R}^{L\times L}.
    \label{eq:Mpk1_Qpk1}
\end{align}
Here $M_p^{k+1}$ is the average alignment (overlap) between the true signal and the estimate at iteration $(k+1)$, and $Q_p^{k+1}$ is the empirical covariance of the estimate. In \cite{Ala18}, the matrices $M_p^{k+1}$ and $Q_p^{k+1}$ are assumed to   converge to well-defined limits $M^{k+1}$ and $Q^{k+1}$,  respectively, defined by the following state evolution recursion. Starting with $M^1= Q^1 = \pi \pi^{\top}  \in \reals^{L \times L}$, for $k \ge 0$ we recursively compute:
\begin{align}
\begin{split}
    A^{k+1}
    &=\frac{1}{\delta}\Big(\diag(\pi)-M^{k+1}-(M^{k+1})^\top+Q^{k+1}\Big) \,  \\
    R^{k+1}
    &=\diag(Q^{k+1}\cdot1_L)-Q^{k+1}  \, , \\
    M^{k+2}
    &=\sum_{l=1}^L\pi_l\E_G\Big[\eta\Big(e_l+(A^{k+1})^{1/2}G, \, \delta^{-1}R^{k+1}\Big)\Big]\cdot e_l^\top \,  ,\\
    Q^{k+2}
    &=\sum_{l=1}^L\pi_l\E_G\Big[\eta\Big(e_l+(A^{k+1})^{1/2}G, \, \delta^{-1}R^{k+1}\Big)\eta\Big(e_l+(A^{k+1})^{1/2}G, \, \delta^{-1}R^{k+1}\Big)^\top\Big] \, ,
\end{split}
\label{eq:Alaoui_SE}
\end{align}
where $G \sim \normal(0, I_L)$.
\begin{remark} \label{rmk:non_rig_hypo_Alaoui}
    The main hypothesis behind the derivations and results of \cite{Ala18} -- which they did not verify rigorously -- is that the empirical distribution  $\big(z_{j,:}^{k+1}\big)_{j \in [p]}$ converges to the law of $\bar{B}+\big(A^{k+1}\big)^{1/2}G$.
\end{remark}

\subsection{Bayes-optimal denoisers and state evolution for matrix-AMP} \label{subsec:bayes_denoiser_SE}
We compute the Bayes-optimal AMP denoisers for the noiseless pooled data setting. Using \eqref{eq:memoryless_denoisers} we have
\begin{align}
    f_k(s)
    &=\E[\bar{B} \mid \Mu^k_B\bar{B}+G_B^k = s]
    =\sum_{l=1}^Le_l\mb{P}[\bar{B}=e_l|s]
    =\sum_{l=1}^Le_l\frac{\mb{P}[\bar{B}=e_l]\mb{P}[s|\bar{B}=e_l]}{\mb{P}[s]} \nonumber \\
    &=\frac{\sum_{l=1}^L\pi_le_l\mb{P}[s|\bar{B}=e_l]}{\sum_{l=1}^L\pi_l\mb{P}[s|\bar{B}=e_l]}
    =\frac{\sum_{l=1}^L\pi_l e_l \, \normal(s;\Mu_B^ke_l,\Mu_B^k)}{\sum_{l=1}^L\pi_l \, \normal(s;\Mu_B^ke_l,\Mu_B^k)},
    \label{eq:Bayes_AMP_fk}
\end{align}
where $\normal(s;\Mu_B^ke_l,\Mu_B^k)$ denotes the multivariate Gaussian probability density with mean $\Mu_B^ke_l$ and covariance $\Mu_B^k$. (Here, we used the fact that $\Tau_B^k=\Mu_B^k$, which is shown below in \eqref{eq:deriv_gk_w.r.t._Z}.)  
Comparing \eqref{eq:Ala_AMP_fk} and \eqref{eq:Bayes_AMP_fk}, we observe that
\begin{align}
    f_{k}(s)=\eta\Big((\Mu_B^{k})^{-1}s,(\Mu_B^{k})^{-1}\Big).
    \label{eq:fk_eta_eqv}
\end{align}

Next, for $k\geq1$, using  \eqref{eq:memoryless_denoisers}, we have
\begin{align}
    g_k(u,\,y)
    &=\Cov[Z \mid Z^k=u]^{-1}\big(\E[Z \mid Z^k=u, \bar{Y}=y]-\,\E[Z \mid Z^k=u] \big) \nonumber \\
    &\stackrel{(a)}{=}\Big(\Sigma^k_{(11)}-\Sigma^k_{(12)}(\Sigma^k_{(22)})^{-1}\Sigma^k_{(21)}\Big)^{-1}\Big(y-\Sigma^k_{(12)}(\Sigma_{(22)}^k)^{-1}u\Big) \nonumber \\
    &\stackrel{(b)}{=}\big(\Sigma_{(11)}^k-\Sigma_{(21)}^k\big)^{-1}(y-u),
    \label{eq:Bayes_AMP_gk}
\end{align}
where (a) applies $\Cov[Z|Z^k]=\Sigma^k_{(11)}-\Sigma^k_{(12)}(\Sigma^k_{(22)})^{-1}\Sigma^k_{(21)}$, $\E[Z|Z^k=u]=\Sigma^k_{(12)}(\Sigma_{(22)}^k)^{-1}u$ (recall that $[Z,Z^k]^\top\sim\normal(0,\Sigma^k)$, and $\E[Z|Z^k,\Bar{Y}=y]=y$ (since $\bar{Y}=Z$ for the noiseless pooled data problem), and (b) applies $\Sigma_{(12)}^k=\Sigma_{(21)}^k=\Sigma_{(22)}^k$ since for $f_k=\E[\bar{B}|\Mu_k^B\bar{B}+G_B^k]$, we have
\begin{align}
\begin{split}
    \Sigma_{(12)}^k
    &=\left(\Sigma_{(21)}^k\right)^\top
    \stackrel{(a)}{=}\frac{1}{\delta}\E[\bar{B}f_k(\Mu_B^k\bar{B}+G_B^k)^\top] \\
    &\stackrel{(b)}{=}\frac{1}{\delta}\E\Big[\bar{B}\, \E[\bar{B}^\top|\Mu_B^k\bar{B}+G_B^k]\Big] \\
    &\stackrel{(c)}{=}\E\Big[\E\Big[\bar{B} \, \E[\bar{B}^\top|\Mu_B^k\bar{B}+G_B^k]\Big|\Mu_B^k\bar{B}+G_B^k\Big]\Big] \\
    &=\E\Big[\E[\bar{B}|\Mu_B^k\bar{B}+G_B^k] \, \E[\bar{B}|\Mu_B^k\bar{B}+G_B^k]^\top\Big] \\
    &=\E[f_kf_k^\top]
    =\Sigma_{(22)}^k,
\end{split}
\label{eq:eqv_of_Sigma_12_21_22}
\end{align}
where (a) uses the definition in \eqref{eq:memoryless_Sigma12}, (b) uses the definition of Bayes-optimal $f_k$  in \eqref{eq:memoryless_denoisers}, and (c) uses the tower property of expectation.

\paragraph{State evolution of matrix-GAMP with Bayes-optimal denoisers} To match the initialization of \AlaouiAMP, we set
$\hB_{j,:} = \pi$, for $j \in [p]$, which using  \eqref{eq:Sig0_def}, gives:
\begin{align}
     \Sigma^0 =   \frac{1}{\delta} \begin{bmatrix}
        \text{Diag}(\pi) &   \pi \pi^{\top} \\
        \pi \pi^{\top} & \pi \pi^{\top}
    \end{bmatrix}.
    \label{eq:Sig0_bayes}
\end{align}
With the Bayes-optimal choice of denoisers, the state evolution recursion in \eqref{eq:memoryless_Mu}-\eqref{eq:memoryless_Sigma} takes the following simpler form. For $k \ge 0$, given $\Mu^k_B$, $\Sigma^k$, and $(Z,Z^k)^\top \sim \normal(0,\Sigma^k)$ independent of $\bar{\Psi}$, we have:
\begin{align}
\begin{split}
        & \Mu_B^{k+1} = \Tau_B^{k+1} =  \E[g_k(Z^k,\bar{Y})g_k(Z^k,\bar{Y})^\top], \qquad 
     \Sigma^{k+1} =
    \begin{bmatrix}
 \Sigma_{(11)}^{k+1} & \Sigma_{(12)}^{k+1} \\ \Sigma_{(21)}^{k+1} & \Sigma_{(22)}^{k+1}
    \end{bmatrix},
\end{split}
\label{eq:MkB_TkB_eqv} 
\end{align}
where 
\begin{align}
    &  \Sigma_{(11)}^{k+1} = \frac{1}{\delta} \E[\bar{B}\bar{B}^\top],  \nonumber \\
    & \Sigma_{(12)}^{k+1} = \Sigma_{(21)}^{k+1} = \Sigma_{(22)}^{k+1}
    = \frac{1}{\delta}  \E[f_{k+1}(\Mu^{k+1}_{B}\bar{B}+G^{k+1}_B)f_{k+1}(\Mu^{k+1}_{B}\bar{B}+G^{k+1}_B)^\top]. \label{eq:Sig_k1_Bayes}
\end{align}
We have already shown \eqref{eq:Sig_k1_Bayes} in \eqref{eq:eqv_of_Sigma_12_21_22}. To show the first equality in \eqref{eq:MkB_TkB_eqv}, we need the multivariate version of Stein's lemma:
    \begin{lemma}\textup{\cite{Tan23c}} \label{lem:gen_steins}
    Let $x=(x_1,\dots,x_L) \in \reals^L$, and $g:\mb{R}^L\rightarrow\mb{R}^L$ be such that for $j, l \in [L]$, the function $x_j\rightarrow g_l(x_1,\dots,x_L)$ 
    is absolutely continuous for Lebesgue almost every $(x_i: \, i\neq j)\in\mb{R}^{L-1}$, with weak derivative $\partial_{x_j} g_l:\mb{R}^L\rightarrow\mb{R}$ satisfying $\E[|\partial_{x_j}g_l(x)|]<\infty$. Let $\nabla g(x)=(\nabla g_1(x),\dots,\nabla g_L(x))^\top\in\mb{R}^{L\times L}$ where $\nabla g_l(x)=\big(\partial_{x_1}g_l(x),\dots,\partial_{x_L}g_l(x)\big)^\top$ for $x\in\mb{R}^L$. If $X\sim \normal(\mu,\Sigma)$ with $\Sigma$ positive definite, then
    \begin{align*}
        \E[\nabla g(X)]=\Big(\Sigma^{-1}\E\big[(X-\mu)g(X)^\top\big]\Big)^\top = \, \E\big[ g(X) (X-\mu)^\top  \big] \Sigma^{-1}.
    \end{align*}
\end{lemma}

Using the tower property of expectation, we can write  $\Mu^{k+1}_B $ in \eqref{eq:memoryless_Mu} as
\begin{align}
    \Mu^{k+1}_B 
    &= \E\Big[\E[\partial_{Z}\tg_k(Z,Z^k,\bar{\Psi})|Z^k]\Big] \nonumber \\
    &\stackrel{(a)}{=}\E\Big[\Cov[Z|Z^k]^{-1}\E\big[(Z-\E[Z|Z^k])\tg_k(Z,Z^k,\bar{\Psi})^\top\big\lvert \,  Z^k \big]\Big]^\top \nonumber \\
    & = \E[\Cov[Z|Z^k]^{-1}(Z-\E[Z|Z^k])\tg_k(Z,Z^k,\bar{\Psi})^\top]^\top \nonumber \\
    & = \E\Big[\E\big[\Cov[Z|Z^k]^{-1}(Z-\E[Z|Z^k])g_k(Z^k,\bar{Y})^\top \big\lvert \, Z^k,\bar{Y}\big]\Big]^\top \nonumber \\
    & \stackrel{(b)}{=} \E\big[g_k(Z^k,\bar{Y})g_k(Z^k,\bar{Y})^\top\big] 
    \,  =\Tau_B^{k+1}, \label{eq:deriv_gk_w.r.t._Z}
\end{align}
where (a) applies Lemma \ref{lem:gen_steins}, (b) follows from the definition of Bayes-optimal $g_k$ in \eqref{eq:memoryless_denoisers}.

\subsection{Equivalence of state evolution recursions}
\label{subsec:SE_equiv}

In this section, we prove the following result, showing the equivalence between the state evolution recursion of \AlaouiAMP in \eqref{eq:Alaoui_SE} and matrix-AMP in \eqref{eq:MkB_TkB_eqv}-\eqref{eq:Sig_k1_Bayes} for noiseless pooled data. 
\begin{proposition}
For $k \ge 0$,  we have
\begin{align}
\begin{split}
&A^{k+1} = \delta^{-1}R^{k+1} = \Mu_B^{k+1,} \\
& M^{k+2} =  \E[f_k(\Mu_B^{k+1}\bar{B}+G_B^{k+1})\bar{B}^\top], \\
&Q^{k+2} =  \E[f_k(\Mu_B^{k+1}\bar{B}+G_B^{k+1})f_k(\Mu_B^{k+1}\bar{B}+G_B^{k+1})^\top].
\end{split}
\label{eq:Ala_SE_v_our_SE_eqv}
\end{align}
\label{prop:SE_equiv_app}
\end{proposition}

\begin{remark}
    Theorem \ref{thm:GAMP} implies that for $k \ge 0$, the overlap between the signal $B$ and the matrix-AMP estimate $\hB^{k+1}$ converges as:
    \[  \frac{1}{p}\sum_{j=1}^p\hB_{j,:}^{k+1}B_{j,:}^\top \,  \stackrel{a.s.}{\to} 
    \, \E[f_k(\Mu_B^{k+1}\bar{B}+G_B^{k+1})\bar{B}^\top] =  M^{k+2},
    \]
    where the equality follows from Proposition \ref{prop:SE_equiv_app}. This is a rigorous version of the claim in \cite{Ala18} that $M^{k+2}_p :=  \frac{1}{p}\sum_{j=1}^p\hz_{j,:}^{k+1}B_{j,:}^\top  \to M^{k+2}$.  
Theorem \ref{thm:GAMP},  together with Proposition \ref{prop:SE_equiv_app},  also provides a rigorous version of the claim in Remark \ref{rmk:non_rig_hypo_Alaoui}, showing 
the empirical distribution of $\big( (\Mu_B^{k+1})^{-1} B^{k+1}_{j, :} \big)_{j \in [p]}$ converges to the law of $\bar{B}+\big(A^{k+1}\big)^{1/2}G$.
\end{remark}

We will need the following elementary result in the proof of Proposition \ref{prop:SE_equiv_app}.
\begin{lemma} \label{lem:E_BB_transpose}
For a random one-hot vector $\bar{B}\in\{0,1\}^L$ with the position of the entry one following the distribution $\textup{Categorical}(\pi)$ where $\pi\in\mb{R}^L$ is a probability vector, we have $\E[\bar{B}\bar{B}^\top]=\diag(\pi)$.
\end{lemma}
\begin{proof}
By the definition of a one-hot vector, for  $l,l'\in[L]$ such that $l\neq l'$, we have $\E[\bar{B}_l\bar{B}_{l'}]=0$. Furthermore,
$ \E[\bar{B}_l\bar{B}_l]  =\mb{P}[\bar{B}_l\bar{B}_l=1]  =\mb{P}[\bar{B}_l =1 ]
    =\pi_l$. The lemma follows.
\end{proof}
We will also use the following identity, which holds in the noiseless pooled data setting with Bayes-optimal matrix-AMP denoisers:
\begin{align}
(\Mu_B^{k+1})^{-1} =     \Sigma_{(11)}^k-\Sigma_{(21)}^k, \quad k \ge 0.
\label{eq:inv_M_formula}
\end{align}
To show \eqref{eq:inv_M_formula}, starting from the formula for 
$\Mu_B^{k+1}$ in \eqref{eq:memoryless_Mu}, we have
\begin{align*}
    (\Mu_B^{k+1})^{-1}
    &=\E[\partial_1\tg_k(Z,Z^k,\bar{\Psi})]^{-1}
    \stackrel{(a)}{=}\left(\frac{\partial}{\partial Z}\Cov[Z|Z^k]^{-1}\big(\E[Z|Z^k,\bar{Y}]-\E[Z|Z^k]\big)\right)^{-1} \\
    &\stackrel{(b)}{=}\left(\frac{\partial}{\partial Z}\big(\Sigma_{(11)}^k-\Sigma_{(21)}^k\big)^{-1}(Z-Z^k)\right)^{-1}
    =\Sigma_{(11)}^k-\Sigma_{(21)}^k,
\end{align*}
where (a) applies \eqref{eq:memoryless_denoisers} and (b) uses the fact that $(Z,Z^k)^\top\sim\normal(0_{2L},\Sigma^k)$ and $\bar{Y}=Z$ for the noiseless pooled data problem.

\begin{proof}[Proof of Proposition \ref{prop:SE_equiv_app}]
We prove the result by induction.
For $k=0$, we have
\begin{align*}
    A^1
    &\stackrel{(a)}{=}\frac{1}{\delta}\Big(\diag(\pi)-M^1-(M^1)^\top+Q^1\Big) \\
    &\stackrel{(b)}{=}\frac{1}{\delta}\bigg(\diag(\pi)-\pi\pi^\top-(\pi\pi^\top)^\top+\pi\pi^\top\bigg) \\
    &\stackrel{(c)}{=}\frac{1}{\delta}\Big(\E\big[\bar{B}\bar{B}^\top\big] - \E\big[\pi\pi^\top\big]\Big) \\
    &\stackrel{(d)}{=}\Sigma_{(11)}^0-\Sigma_{(22)}^0
    \stackrel{(e)}{=}\big(\Mu_B^{1}\big)^{-1},
\end{align*}
where (a) uses \eqref{eq:Alaoui_SE},  (b) uses the initialization $M^1$ specified above \eqref{eq:Alaoui_SE}, (c) applies Lemma \ref{lem:E_BB_transpose}, (d) applies \eqref{eq:Sig0_bayes}, and (e) uses \eqref{eq:inv_M_formula} along with $\Sigma_{(21)}^0=\Sigma_{(22)}^0$. Similarly, using \eqref{eq:Alaoui_SE} and the initialization $Q^1= \pi \pi^{\top}$, we have
\begin{align*}
    \delta^{-1}R^1
    = \frac{1}{\delta}\Big(\diag\big(Q^1\cdot 1_L\big)-Q^1\Big) 
    = \frac{1}{\delta}\Big(\diag(\pi)- \pi\pi^\top \Big) 
     = \Sigma_{(11)}^0-\Sigma_{(22)}^0 = \big(\Mu_B^1\big)^{-1}.
\end{align*}

\textbf{Induction step.} Assume that  the identities \eqref{eq:Ala_SE_v_our_SE_eqv} hold for $A^k$, $M^{k+1}$, and $Q^{k+1}$. Now we want to show that they hold for $A^{k+1}$, $M^{k+2}$, $Q^{k+2}$.  We begin by deriving alternative expressions for 
$\E[f_k(\Mu_B^k\bar{B}+G_B^k)\bar{B}^\top]$ and $\E\Big[f_k(\Mu_B^k\bar{B}+G_B^k) f_k(\Mu_B^k\bar{B}+G_B^k)^\top\Big]$ in terms of $\eta(\cdot, \cdot)$. For $k \ge 1$, we have
\begin{align}
    \E[f_k(\Mu_B^k\bar{B}+G_B^k)\bar{B}^\top]
    &=\E_{G_B^k}\Big[\E_{\bar{B}}\Big[f_k(\Mu_B^k\bar{B}+G_B^k)\bar{B}^\top\Big|G_B^k\Big]\Big] \nonumber \\
    &\stackrel{(a)}{=}\E_{G_B^k}\Big[\sum_{l=1}^L\pi_l f_k(\Mu_B^ke_l+G_B^k)e_l^\top\Big] \nonumber \\
    &\stackrel{(b)}{=}\sum_{l=1}^L\pi_l\E_G\Big[f_k(\Mu_B^ke_l+(\Mu_B^k)^{1/2}G)\Big]\cdot e_l^\top  \nonumber \\
    &\stackrel{(c)}{=}\sum_{l=1}^L\pi_l\E_G\Big[\eta\Big(e_l+(\Mu_B^k)^{-1/2}G, \, (\Mu_B^k)^{-1}\Big)\Big]\cdot e_l^\top,
    \label{eq:M^k_ours}
\end{align}
where $G \sim \normal(0, I_L)$. Here (a) uses the prior $\pi$ on the location of the one  in $\bar{B}$, (b) uses $G_B^{k+1}\sim\normal(0_L,\Tau_B^{k+1})$ and $\Tau_B^{k+1}=\Mu_B^{k+1}$ from \eqref{eq:MkB_TkB_eqv}, and (c) uses the equivalence in \eqref{eq:fk_eta_eqv}.  By a similar argument,
\begin{align}
    \E\Big[f_k(\Mu_B^k\bar{B}+G_B^k)&f_k(\Mu_B^k\bar{B}+G_B^k)^\top\Big]
    = \E_{G_B^k}\Big[\E_{\bar{B}}\Big[f_k(\Mu_B^k\bar{B}+G_B^k)f_k(\Mu_B^k\bar{B}+G_B^k)^\top\Big|G_B^k\Big]\Big] \nonumber \\
    &= \E_{G_B^k}\Big[\sum_{l=1}^L\pi_lf_k(\Mu_B^ke_l+G_B^k)f_k(\Mu_B^ke_l+G_B^k)^\top\Big] \nonumber \\
    &\stackrel{(a)}{=}\sum_{l=1}^L\pi_l\E_G\Big[f_k(\Mu_B^k\bar{B}+(\Mu_B^k)^{1/2}G)f_k(\Mu_B^k\bar{B}+(\Mu_B^k)^{1/2}G)^\top\Big] \nonumber \\
    &\stackrel{(b)}{=}\sum_{l=1}^L\pi_l\E_G\Big[\eta\Big(e_l+(\Mu_B^k)^{-1/2}G,(\Mu_B^k)^{-1}\Big)\eta\Big(e_l+(\Mu_B^k)^{-1/2}G,(\Mu_B^k)^{-1}\Big)^\top\Big],
    \label{eq:Q^k_ours}
\end{align}
Here (a) uses $G_B^{k+1}\sim\normal(0_L,\Tau_B^{k+1})$ and $\Tau_B^{k+1}=\Mu_B^{k+1}$, and (b) uses the equivalence in \eqref{eq:fk_eta_eqv}.

Let us use the shorthand $f_k:=f_k(\Mu_B^k\bar{B}+G_B^k)$.
Starting from \eqref{eq:Alaoui_SE}, we write $ A^{k+1}$ as
\begin{align*}
    A^{k+1}
    &=\frac{1}{\delta}\Big(\diag(\pi)-M^{k+1}-(M^{k+1})^\top+Q^{k+1}\Big) \\
    &\stackrel{(a)}{=}\frac{1}{\delta}\Big(\E[\bar{B}\bar{B}^\top]-\E[f_k\bar{B}^\top]-\E[f_k\bar{B}^\top]^\top+\E[f_kf_k^\top]\Big) \\
    &\stackrel{(b)}{=}\frac{1}{\delta}\Big(\E[\bar{B}\bar{B}^\top]-\E[f_kf_k^\top]\Big) \\
    &\stackrel{(c)}{=}\Sigma_{(11)}^k-\Sigma_{(22)}^k
    \stackrel{(d)}{=}\big(\Mu_B^{k+1}\big)^{-1},
\end{align*}
where (a) applies Lemma \ref{lem:E_BB_transpose}, and uses \eqref{eq:M^k_ours} and the inductive hypothesis to get $M^{k+1}=\E[f_k\bar{B}^\top]$, as well as \eqref{eq:Q^k_ours} and the inductive hypothesis to get $Q^{k+1}=\E[f_kf_k^\top]$, (b) uses the fact that $\E[f_k\bar{B}^\top]=\E[f_kf_k^\top]$ for Bayes-optimal $f_k$, (c) uses the formulas in \eqref{eq:Sig_k1_Bayes}, and (d) uses \eqref{eq:inv_M_formula} with $\Sigma_{(21)}^k=\Sigma_{(22)}^k$ by \eqref{eq:Sig_k1_Bayes}. Next, we have
\begin{align*}
    \delta^{-1}R^{k+1}
    &=\frac{1}{\delta}\Big(\diag(Q^{k+1}\cdot1_L)-Q^{k+1}\Big) \\
    &\stackrel{(a)}{=}\frac{1}{\delta}\Big(\diag\Big(\E[f_kf_k^\top]\cdot1_L\Big)-\E[f_kf_k^\top]\Big) \\
    &\stackrel{(b)}{=}\Sigma_{(11)}^k-\Sigma_{(22)}^k
    \stackrel{(c)}{=}\big(\Mu_B^{k+1}\big)^{-1},
\end{align*}
where (a) uses \eqref{eq:Q^k_ours} and the inductive hypothesis to get $Q^{k+1}=\E[f_kf_k^\top]$,  and (b) is obtained as follows:
\begin{align*}
    \diag\Big(\E[f_kf_k^\top]\cdot1_L\Big)
    &\stackrel{\rm (i)}{=}\diag\Big(\E[f_k\bar{B}^\top1_L]\Big)
    \stackrel{\rm (ii)}{=}\diag\Big(\E\big[\E\big[\bar{B}\big|\Mu_B^k\bar{B}+G_B^k\big]\big]\Big) \\
    &\stackrel{\rm (iii)}{=}\diag\big(\E[\bar{B}]\big)
    =\diag\big(\pi\big)
    \stackrel{\rm (iv)}{=}\E[\bar{B}\bar{B}^\top],
\end{align*}
where (i) applies $\E[f_kf_k^\top]=\E[f_k\bar{B}^\top]$ for Bayes-optimal $f_k$, (ii) uses the definition of $f_k$ in \eqref{eq:memoryless_denoisers} and the fact that $\bar{B}^\top1_L=1$ since $\bar{B}$ is a one-hot vector, (iii) uses the law of total expectation, and (iv) applies Lemma \ref{lem:E_BB_transpose}. 

Substituting $A^{k+1}=\delta^{-1}R^{k+1}=\big(\Mu_B^{k+1}\big)^{-1}$ into the definitions of $M^{k+2}$ and $Q^{k+2}$ in \eqref{eq:Alaoui_SE} gives us the desired formulas of $M^{k+2}$ and $Q^{k+2}$ in \eqref{eq:Ala_SE_v_our_SE_eqv}. This completes the proof of the induction step, and the proposition follows.
\end{proof}

\subsection{Deriving matrix-AMP from \AlaouiAMP} \label{sec:eqv_of_AMPs}
We start with the matrix-AMP and rewrite it in terms of $\bX$ 
defined in \eqref{eq:bar_X_and_Y}. Using \eqref{eq:bar_X_and_Y} in the matrix-AMP equations in \eqref{eq:memoryless_GAMP}, 
and rearranging, we obtain:
\begin{align}
\begin{split}
    & \sqrt{\alpha(1-\alpha)\delta}\cdot\Theta^k
     =\bX\hB^k-\hR^{k-1}(F^k)^\top\sqrt{\alpha(1-\alpha)\delta},  \qquad \hR^k=g_k(\Theta^k,\tY),  \\
    & \sqrt{\alpha(1-\alpha)\delta}\cdot B^{k+1}
    =\bX^\top\hR^k-\hB^k(C^k)^\top\sqrt{\alpha(1-\alpha)\delta}, \qquad 
    \hB^{k+1} =f_{k+1}(B^{k+1}).
\end{split}
\label{eq:modified_AMP1}
\end{align}
Here the functions $g_k, f_k$ are as defined in \eqref{eq:Bayes_AMP_fk}-\eqref{eq:Bayes_AMP_gk}, and act row-wise on their matrix inputs. We recall that the matrices $C^k, F^{k+1} \in \reals^{L \times L}$ are defined as 
\begin{align*}
    C^k= \frac{1}{n}\sum_{i=1}^ng_k'(\Theta_{i,:}^k,\tY_{i,:}), \   \ 
    F^{k+1}=\frac{1}{n}\sum_{j=1}^pf_{k+1}'(B_{j,:}^{k+1}),
\end{align*}
where $g_k', f_{k+1}'$ are the Jacobians with respect to their first arguments.

To match \AlaouiAMP, we define $\tw^{k+1}\stackrel{\Delta}{=}\sqrt{\alpha(1-\alpha)\delta}\cdot\Theta^k$. This gives
\begin{align}
    g_k(\Theta_{i,:}^k,\tY_{i,:})
    &=\big(\Sigma_{(11)}^k-\Sigma_{(21)}^k\big)^{-1}(\tY_{i,:}-\Theta_{i,:}^k) \nonumber \\
    &=\big(\Sigma_{(11)}^k-\Sigma_{(21)}^k\big)^{-1}\frac{\bY_{i,:}-\sqrt{\alpha(1-\alpha)\delta}\cdot\Theta_{i,:}^k}{\sqrt{\alpha(1-\alpha)\delta}} \nonumber \\
    &=\big(\Sigma_{(11)}^k-\Sigma_{(21)}^k\big)^{-1}\frac{\bY_{i,:}-\tw_{i,:}^{k+1}}{\sqrt{\alpha(1-\alpha)\delta}} \nonumber \\
    &=:\bar{g}_k\big(\tw_{i,:}^{k+1},\bY_{i,:}\big)
    =\hR_{i,:}^k, \label{eq:R_hat_eqv}
\end{align}
and
\begin{align}
    C^k
    &=\frac{1}{n}\sum_{i=1}^ng_k'\big(\Theta_{i,:}^k,\tY_{i,:}\big)
    =\frac{1}{n}\sum_{i=1}^n\bar{g}_k'\big(\tw_{i,:}^{k+1},\bY_{i,:}\big)\sqrt{\alpha(1-\alpha)\delta} \nonumber \\
    &=\frac{1}{n}\sum_{i=1}^n\bigg(-\frac{(\Sigma_{(11)}^k-\Sigma_{(21)}^k)^{-1}}{\sqrt{\alpha(1-\alpha)\delta}}\bigg)\sqrt{\alpha(1-\alpha)\delta} \nonumber \\
    &=-\big(\Sigma_{(11)}^k-\Sigma_{(21)}^k\big)^{-1}.
    \label{eq:Ck_eqv}
\end{align}
Furthermore, we have
\begin{align}
    F^{k+1}
    &=\frac{1}{n}\sum_{j=1}^pf_{k+1}'(B_{j,:}^{k+1})
    =\frac{1}{n}\sum_{j=1}^p\eta'\Big((\Mu_B^{k+1})^{-1}B_{j,:}^{k+1},(\Mu_B^{k+1})^{-1}\Big)(\Mu_B^{k+1})^{-1},
    \label{eq:Fk_eqv}
\end{align}
where the last equality is by the chain rule of differentiation. Note that $\eta'(\cdot,\cdot)$ is the Jacobian with respect to the first argument. Substituting \eqref{eq:R_hat_eqv}, \eqref{eq:Ck_eqv}, and \eqref{eq:Fk_eqv} into \eqref{eq:modified_AMP1} and rearranging gives
\begin{align}
    \tw^{k+1}
    &=\bX\hB^k-(\bY-\tw^{k})\big(\Sigma_{(11)}^{k-1}-\Sigma_{(21)}^{k-1}\big)^{-1}\big(\Mu_B^k\big)^{-1}\cdot\frac{1}{n}\sum_{j=1}^p\eta'\Big((\Mu_B^k)^{-1}B_{j,:}^k,(\Mu_B^k)^{-1}\Big)^\top
    \label{eq:tw_k1}
\end{align}
and 
\begin{align}
    & \sqrt{\alpha(1-\alpha)\delta}\cdot B^{k+1}
    = \Bigg(\frac{\bX^\top(\bY-\tw^{k+1})}{\sqrt{\alpha(1-\alpha)\delta}} + \hB^k  \sqrt{\alpha(1-\alpha)\delta} \Bigg) \big[\big(\Sigma_{(11)}^k-\Sigma_{(21)}^k\big)^{-1}\big]^\top \nonumber \\
    & \implies
    B^{k+1}\big(\Sigma_{(11)}^k-\Sigma_{(21)}^k\big)^{\top}
    = \, \hB^k \, + \, \frac{\bX^\top(\bY-\tw^{k+1})}{\alpha(1-\alpha)\delta}.
    \label{eq:Bk1_adj}
\end{align}
Let us set $\tilz^{k+1}\triangleq B^{k+1}\big(\Sigma_{(11)}^k-\Sigma_{(21)}^k\big)^{\top}=B^{k+1}\big( \Mu^{k+1}_B \big)^{-1}$, where the second quality follows from \eqref{eq:inv_M_formula}, and noting from \eqref{eq:deriv_gk_w.r.t._Z} that $\Mu^{k+1}_B $ is symmetric. For notation simplicity, we also denote
\begin{align}
    V^k=\alpha(1-\alpha)\delta (\Mu_B^{k+1})^{-1}
  \,   \text{ and } \,
    U^k=\frac{\alpha(1-\alpha)\delta}{n}\sum_{j=1}^p \eta'\Big((\Mu_B^k)^{-1}B_{j,:}^k,(\Mu_B^k)^{-1}\Big) (\Mu_B^k)^{-1}
    \label{eq:UkVk_equivalences}
\end{align}
Looking at the matrix-AMP equations \eqref{eq:tw_k1} and \eqref{eq:Bk1_adj} row by row, and using \eqref{eq:fk_eta_eqv}, we obtain the following for $j \in [p]$ and $i \in [n]$:
\begin{align}
\begin{split}
   & \hB_{j,:}^k = \eta(\tilz_{j,:}^k, \, (\Mu_B^k)^{-1}), \\ 
   &  \tw_{i,:}^{k+1} =\sum_{j=1}^p\bX_{ij}\hB_{j,:}^k-U^k(V^{k-1})^{-1}(\bY_{i,:}-\tw_{i,:}^{k}),  \\
    & \tilz_{j,:}^{k+1}
    =\hB_{j,:}^k+\frac{V^k}{\alpha(1-\alpha)\delta}\sum_{i=1}^n\bX_{ij}(V^k)^{-1}(\bY_{i,:}-\tw_{i,:}^{k+1}).
\end{split}
\label{eq:modified_AMP2}
\end{align}

Comparing  the equations for matrix-AMP in \eqref{eq:modified_AMP2}  and \AlaouiAMP in \eqref{eq:Alaoui_AMP}, we observe that \AlaouiAMP can be reduced to matrix-AMP, i.e. $( \hz^{k+1}, w^{k+1}, z^{k+1}) = (\hB^k, \tw^{k+1}, \tilz^{k+1})$ for $k \ge 0$, by making the following substitutions in \eqref{eq:Alaoui_AMP}:
\begin{align}
    \Gamma_j^{k+1} \to   (\Mu_B^{k+1})^{-1} \  \text{ for } j \in [p], \  \  \text{ and } \quad 
      \Xi_i^{k+1} \to     U^{k},  \quad 
      \Xi_i^{k} \to  V^{k-1} \  \text{ for } i \in [n].
      \label{eq:eqv_needed}
\end{align}
Both matrix-AMP and  \AlaouiAMP are intialized in the same way, i.e.,  $\hB^0_{j,:} =
\hz^{1}_{j,:}=\pi$.

To justify the substitutions in \eqref{eq:eqv_needed}, we first show that $U^k$ can be approximated by $V^{k}$, for $k \ge 1$. Recalling the definition of $U^k$ from \eqref{eq:UkVk_equivalences}, we approximate the empirical average of $\eta'$ using the expectation as follows:
\begin{align}
\begin{split}
    & \frac{1}{p}\sum_{j=1}^p\eta'\big((\Mu_B^k)^{-1}B_{j,:}^k \, , \, (\Mu_B^k)^{-1} \big)  (\Mu_B^k)^{-1}
    \approx \E[\eta'\big(\bar{B}+(\Mu_B^k)^{-1}G_B^k \, , \, (\Mu_B^k)^{-1} \big)] (\Mu_B^k)^{-1} \\
    &\stackrel{(a)}{=}\E\Big[\E[\eta'\big(\bar{B}+(\Mu_B^k)^{-1}G_B^k,(\Mu_B^k)^{-1}\big)|\bar{B}]\Big](\Mu_B^k)^{-1} \\
    &\stackrel{(b)}{=} \E\left[\eta(\bar{B}+(\Mu_B^k)^{-1}G_B^k \, , \, (\Mu_B^k)^{-1}) (G_B^k)^\top (\Mu_B^k)^{-1}  \right] 
    \big((\Mu_B^k)^{-1}\big)^{-1} (\Mu_B^k)^{-1} \\
    &\stackrel{(c)}{=}\E\left[f_k(\Mu_B^k\bar{B}+G_B^k)\, (\bar{B}^\top\Mu_B^k+(G_B^k)^\top-\bar{B}^\top\Mu_B^k)\right](\Mu_B^k)^{-1} \\
    &=\delta\left\{\frac{1}{\delta}\E[f_k(\Mu_B^k\bar{B}+G_B^k)(\bar{B}^\top+(G_B^k)^\top(\Mu_B^k)^{-1})]-\frac{1}{\delta}\E[f_k(\Mu_B^k\bar{B}+G_B^k)\, \bar{B}^\top]\right\} \\
    &\stackrel{(d)}{=}\delta\big(\Sigma_{(11)}^k-\Sigma_{(21)}^k\big).
    \end{split}
    \label{eq:eta_pr_approx}
\end{align}
Here  (a) uses the tower property of expectation, (b) applies Lemma \ref{lem:gen_steins} (recall  $\eta'(\cdot)$ is the Jacobian with respect to the first argument), (c) uses \eqref{eq:fk_eta_eqv}, and (d) uses $\frac{1}{\delta}\E[f_k\bar{B}^\top]=\Sigma_{(21)}^k$ (see \eqref{eq:eqv_of_Sigma_12_21_22}) and the following:
\begin{align*}
    \frac{1}{\delta}\E\Big[f_k(\cdot)\Big(\bar{B}^\top+(G_B^k)^\top(\Mu_B^k)^{-1}\Big)\Big]
    &=\frac{1}{\delta}\E\Big[\E[\bar{B}|\Mu_B^k\bar{B}+G_B^k]\cdot(\bar{B}^\top(\Mu_B^k)+(G_B^k)^\top)\Big](\Mu_B^k)^{-1} \\
    &=\frac{1}{\delta}\E\Big[\E\Big[\bar{B}\Big(\bar{B}^\top(\Mu_B^k)+(G_B^k)^\top\Big)\Big|\Mu_B^k\bar{B}+G_B^k\Big]\Big](\Mu_B^k)^{-1} \\
    &\stackrel{(a)}{=}\frac{1}{\delta}\E\left[\bar{B}\bar{B}^\top(\Mu_B^k)+\bar{B}(G_B^k)^\top\right](\Mu_B^k)^{-1} \\
    &\stackrel{(b)}{=}\frac{1}{\delta}\E[\bar{B}\bar{B}^\top]+\frac{1}{\delta}\E[\bar{B}]\cdot\E[G_B^k]^\top(\Mu_B^k)^{-1}
     \stackrel{(c)}{=}\Sigma_{(11)}^k,
\end{align*}
where (a) uses the tower property of expectation, (b) uses the independence of $\bar{B}$ and $G_B^k$ (mentioned below \eqref{eq:memoryless_Thetak_conv}), and (c) uses $\Sigma_{(11)}^k=\frac{1}{\delta}\E[\bar{B}\bar{B}^\top]$ (defined in \eqref{eq:memoryless_Sigma11}) and $\E[G_B^k]=0$ (distribution of $G_B^k$ is mentioned below \eqref{eq:memoryless_Thetak_conv}). Using \eqref{eq:eta_pr_approx} in the definition of $V^k$ in \eqref{eq:eqv_needed} and recalling the identity in \eqref{eq:inv_M_formula}, we have that $V^k$ can be approximated by $U^k$  for $k \ge 1$.

We inductively justify the first two substitutions in \eqref{eq:eqv_needed}, assuming  that they have been justified up to step $(k-1)$ for some $k \ge 1$. 
By the induction hypothesis, we have $\Gamma_j^{k} \approx (\Mu_B^{k})^{-1} $ and $z^k_{j,:} \approx \tilz^k_{j,:}$, for $j \in [p]$. Using this in the first lines of \eqref{eq:Alaoui_AMP} and \eqref{eq:modified_AMP2} gives that $\hz^{k+1}_{j, :} \approx \hB^{k}_{j, :}$, for $j \in [p]$. We also recall that the entries of $\bX$ are i.i.d.~with
\begin{align*}
    \bX_{ij}
    =
    \begin{cases}
        \frac{1-\alpha}{\sqrt{p}} &\text{with probability $\alpha$} \\
        -\frac{\alpha}{\sqrt{p}} &\text{with probability $1-\alpha$}.
    \end{cases}
\end{align*}
From the definition of $\Xi_i^{k+1}$ in \eqref{eq:Alaoui_AMP},  we have, for $i \in [n]$:
\begin{align}
\begin{split}
    \Xi_i^{k+1}
    &=\sum_{j=1}^p\bX_{ij}^2\Big(\diag(\hz_{j,:}^{k+1})-\hz_{j,:}^{k+1}(\hz_{j,:}^{k+1})^\top\Big) \\
    &\stackrel{(a)}{\approx}\sum_{j=1}^p\frac{\alpha(1-\alpha)}{p}\Big(\diag(\hB_{j,:}^k)-\hB_{j,:}^k(\hB_{j,:}^k)^\top\Big) \\
    &\stackrel{(b)}{\approx}\alpha(1-\alpha)\E[\diag(f_k)-f_kf_k^\top] \\
    &=\alpha(1-\alpha)\delta\left(\frac{1}{\delta}\E[\diag(f_k)]-\frac{1}{\delta}\E[f_kf_k^\top]\right) \\
    &\stackrel{(c)}{=}\alpha(1-\alpha)\delta(\Sigma_{(11)}^k-\Sigma_{(21)}^k) = \alpha(1-\alpha)\delta(\Mu_B^{k+1})^{-1}
    =V^k.
    \label{eq:Xi1_equiv}
\end{split}
\end{align}
Here (a) uses $\bX_{ij}^2\approx\E[\bX_{ij}^2]=\frac{\alpha(1-\alpha)}{p}$ and $\hz_{j,:}^{k+1}\approx \hB_{j,:}^k$, (b) uses replaces the empirical average of the AMP iterates $\hB_{j,:}^k$ by the state evolution limit (with $f_k$ shorthand for $f_k(\Mu_B^k\bar{B}+G_B^k)$), and (c) uses $\Sigma_{(22)}^k=\Sigma_{(21)}^k =\frac{1}{\delta}\E[f_kf_k^\top]$ and
\begin{align*}
    \frac{1}{\delta}\E[\diag(f_k)]
    &= \frac{1}{\delta}\diag\Big(\E\Big[\E[\bar{B}|\Mu_B^k\bar{B}+G_B^k]\Big]\Big)
    = \frac{1}{\delta}\diag(\E[\bar{B}]) \\
    &=\frac{1}{\delta}\diag(\pi)
    \stackrel{(a)}{=}\frac{1}{\delta}\E[\bar{B}\bar{B}^\top]
    =\Sigma_{(11)}^k,
\end{align*}
where (a) follows from Lemma \ref{lem:E_BB_transpose}.

Next, using the definition of $ \Gamma_j^{k+1}$ in \eqref{eq:Alaoui_AMP} we have 
\begin{align}
    \Gamma_j^{k+1}
    &=\left(\sum_{i=1}^n \bX_{ij}^2(\Xi_i^{k+1})^{-1}\right)^{-1}
    \stackrel{(a)}{\approx}\left(\sum_{i=1}^n \bX_{ij}^2(V^{k})^{-1}\right)^{-1}
     \stackrel{(b)}{\approx} \frac{V^{k}}{\alpha(1-\alpha)\delta} = (\Mu_B^{k+1})^{-1},
    \label{eq:Gam_k1}
\end{align}
where in step (a) we  have used the induction hypothesis together with $U^k \approx V^k$, and step (b) uses the law of large numbers.
 This completes the justification of the substitutions in \eqref{eq:eqv_needed} to obtain  matrix-AMP from \AlaouiAMP.

\section{Proof of Corollary \ref{cor:FPR_FNR}} \label{sec:FPR_FNR_proof}

Observe that any indicator function satisfies the polynomial growth condition in \eqref{eq:poly_growth_cond} with order $r=2$. %(the 2 here comes from the order in \eqref{eq:GAMP_vec1} 
%and \eqref{eq:GAMP_vec2}), since 
Indeed, taking $\phi$ in \eqref{eq:poly_growth_cond}  to be an indicator function, the left hand side takes values in $\{0,1\}$ while the right hand side of \eqref{eq:poly_growth_cond} is always greater than one for $C=1$. Hence, using Theorem \ref{thm:GAMP_vec}
and recalling from \eqref{eq:sum_to_exp} that $W_2$ convergence implies the convergence of the empirical average of $\phi$, the numerator of the FPR in \eqref{eq:FPR_and_FNR} satisfies
\begin{align*}
    \frac{1}{p}\sum_{j=1}^p\mathds{1}\big\{\beta_j=0\cap j\in\widehat{\mathcal{S}}\big\}  
    &  =   \frac{1}{p}\sum_{j=1}^p\mathds{1}\big\{\beta_j=0\cap f_K(\beta_j^K)>\zeta\big\} \\
    &\stackrel{a.s.}{\rightarrow}\E\left[\mathds{1}\left\{\bar{\beta}=0\cap f_K(\mu_\beta^K \bar{\beta} +  G_\beta^K)>\zeta\right\}\right] \\
    & = \mb{P}[\bar{\beta}=0]\cdot\mb{P}[f_K(G_\beta^K)>\zeta] \\
    &= (1-\pi)\cdot\mb{P}[f_K(G_\beta^K)>\zeta].
\end{align*}
For the denominator of FPR in \eqref{eq:FPR_and_FNR}, we note that $\frac{1}{p}\sum_{j=1}^p\beta_j\stackrel{a.s.}{\rightarrow}\E[\bar{\beta}]=\pi$.  Putting the above together gives the convergence result for the FPR.  Following a similar series of steps for the FNR, we get
\begin{align*}
    \text{FNR}
    \stackrel{a.s.}{\rightarrow}
    \mb{P}\big[f_K(\mu_\beta^K+G_\beta^K)\leq\zeta\big]. 
\end{align*}

\section{Implementation of AMP denoisers and their derivatives} \label{sec:imp_details}

\subsection{Pooled Data} \label{sec:imp_details_pooled_data}

The implementation details for the AMP pooled data simulations in Section \ref{sec:sim_pool_data} are largely the same as those presented in \cite[Appendix A.1]{Tan23c} for the mixed linear regression problem, with the only difference being the implementation details of the denoisers $f_k^*$ and $g_k^*$. Hence, we will only focus on them.

The optimal denoiser $f_k^*$ depends only on the signal prior and was derived in Section \ref{sec:eqv_of_AMPs}: see \eqref{eq:Bayes_AMP_fk}. For $g_k^*$, we have:
\begin{align}
   g_k^*(u,\,y)
    &=\E[Z \mid Z^k=u, \bar{Y}=y]-\,\E[Z \mid Z^k=u]
    = \E[Z \mid Z^k=u, \bar{Y}=y] \, - \, u,
\end{align}
where the equality is obtained from \eqref{eq:gk_opt_def}. We now evaluate $\E[Z \mid Z^k=u, \bar{Y}=y]$. For the noiseless setting we have $\E[Z|Z^k=u,\bar{Y}=y]=y$ giving us
\begin{align}
    g_k^*(u,y)
    =y-u.
    \label{eq:gk_noiseless_pooled_data}
\end{align}
From \eqref{eq:SE_Mk1B}, this choice of $g_k^*$ will result in $\Mu_B^k = I_L$.
In the noisy setting,  only the evaluation of $\E[Z|Z^k=u,\Bar{Y}=y]$ changes from the noiseless case. If we have  $\Psi_i \stackrel{\iid}{\sim} \normal(0, p\sigma^2I_L)$, after scaling according to \eqref{eq:tYtPsi_def} this corresponds to the state evolution random vector $\bar{\Psi} \sim \normal\big(0,\frac{\sigma^2}{\delta\alpha(1-\alpha)} I_L \big)$. Since $\bar{Y} = Z + \bar{\Psi}$, the random vectors $(Z, Z^k, \bar{Y})$ are jointly Gaussian.
We denote the covariance of $(Z,Z^k,\Bar{Y})$ by 
\begin{align*}
    \Sigma^{k,+}=
    \begin{bmatrix}
        \Sigma^{k,+}_{(11)} & \Sigma^{k,+}_{(12)} & \Sigma^{k,+}_{(13)} \\
        \Sigma^{k,+}_{(21)} & \Sigma^{k,+}_{(22)} & \Sigma^{k,+}_{(23)} \\
        \Sigma^{k,+}_{(31)} & \Sigma^{k,+}_{(32)} & \Sigma^{k,+}_{(33)}
    \end{bmatrix}
    \in\mb{R}^{3L\times 3L},
\end{align*}
where each sub-matrix above is in $\mb{R}^{L\times L}$. Since the covariance of $(Z, Z^k)$ equals $\Sigma^k$ (using the simplified state evolution notation in \eqref{eq:memoryless_Sigma}), we have
\begin{align*}
    \begin{bmatrix}
        \Sigma^{k,+}_{(11)} & \Sigma^{k,+}_{(12)} \\
        \Sigma^{k,+}_{(21)} & \Sigma^{k,+}_{(22)}
    \end{bmatrix}
    =\Sigma^k.
\end{align*}
Next,
\begin{align*}
    \Sigma^{k,+}_{(13)}
    &=\Sigma^{k,+}_{(31)}
    =\Cov[Z,\Bar{Y}]
    =\Cov[Z,Z+\Bar{\Psi}]
    =\Cov[Z,Z]+\Cov[Z,\Bar{\Psi}]
    \stackrel{(a)}{=}\Sigma^{k}_{(11)} \\
    \Sigma^{k,+}_{(23)}
    &=\Sigma^{k,+}_{(32)}
    =\Cov[Z^k,\Bar{Y}]
    =\Cov[Z^k,Z+\Bar{\Psi}]
    =\Cov[Z^k,Z]+\Cov[Z^k,\Bar{\Psi}]
    \stackrel{(b)}{=}\Sigma_{(12)}^k,
\end{align*}
where (a) and (b) use the fact that $(Z,Z^k)$ is independent of $\Psi$. Using the independence of $Z$ and $\Psi$, we also obtain
\begin{align*}
    \Sigma^{k,+}_{(33)}
    =\Cov[\Bar{Y}]
    =\Cov[Z+\Bar{\Psi}]
    \stackrel{(a)}{=}\Cov[Z]+\Cov[\Psi]
    =\Sigma_{(11)}^k+\frac{\sigma^2}{\delta\alpha(1-\alpha)}I_L.
\end{align*}
%where (a) uses the fact that $Z$ is independent to $\Psi$. 
Finally, using a standard property of jointly Gaussian vectors, we can compute
\begin{align*}
    \E[Z \mid Z^k=u,\Bar{Y}=y]
    &=\Sigma_{[L],[L+1,3L]}^{k,+}
    \left(\Sigma_{[L+1:3L],[L+1:3L]}^{k,+}\right)^{-1}
    \begin{bmatrix}
        u \\
        y
    \end{bmatrix} \\
    &\stackrel{(a)}{=}
    \begin{bmatrix}
        \Sigma_{(22)}^k & \Sigma_{(11)}^k
    \end{bmatrix}
    \begin{bmatrix}
        \Sigma_{(22)}^k & \Sigma_{(22)} \\
        \Sigma_{(22)}^k & \Sigma_{(11)}^k+\frac{\sigma^2}{\delta\alpha(1-\alpha)I_L}
    \end{bmatrix}^{-1}
    \begin{bmatrix}
        u \\
        y
    \end{bmatrix},
\end{align*}
where (a) applies $\Sigma_{(12)}^k=\Sigma_{(21)}^k=\Sigma_{(22)}^k$ since we are using Bayes-optimal denoisers. Using the block inversion formula for the inverse on the RHS, we obtain:
\begin{align}
    g_k^*(u,y)
    &=\bigg( \Sigma_{(11)}^k -\Sigma_{(22)}^k \bigg)
    \Big(\Sigma_{(11)}^k +\frac{\sigma^2}{\delta\alpha(1-\alpha)}I_L-\Sigma_{(22)}^k\Big)^{-1}(y-u).
    \label{eq:gk_noisy_pooled_data_w_factor}
\end{align}
Note that $(\Sigma_{(11)}^k -\Sigma_{(22)}^k) = \frac{1}{\delta} \E[ (\bar{B} - f_k^*(\Mu_B^{k}\bar{B}+G_B^{k}))(\bar{B} - f_k^*(\Mu_B^{k}\bar{B}+G_B^{k}))^\top]$  is positive definite, and hence its inverse is as well. Recalling from the discussion below \eqref{eq:gk_opt_def} that left multiplying $\E[Z|Z^k,\bar{Y}]-\E[Z|Z^k]$ by any positive definite matrix does not change the optimality of $g_k$, we can therefore take
\begin{align}
    g_k^*(u,y)
    &=y-u,
    \label{eq:gk_noisy_pooled_data}
\end{align}
which  is the same denoiser as in the noiseless setting; this also implies $\Mu_B^k = I_L$ -- see \eqref{eq:gk_noiseless_pooled_data}. Therefore, the derivative term $C^k$ in \eqref{eq:memoryless_GAMP}
 is $C^k =-I_L$,
for both the noiseless and noisy setting. Recall from \eqref{eq:memoryless_GAMP} that 
$$
F^{k}
=\frac{1}{n}\sum_{j=1}^p f_{k}'(B_{j,:}^{k})
\approx\frac{1}{\delta}\E\big[f_k'(\Mu_B^k\bar{B}+G_B^k)\big].
$$
Denoting $f_{k,l}^*$ as the $l$th entry of the output of $f_k^*$, we have
\begin{align}
    f_k'(s^k)
    &=\begin{bmatrix}
        \frac{\partial f_{k,1}}{\partial s_1^{k}} & \dots & \frac{\partial f_{k,1}}{\partial s_L^{k}} \\
        \vdots & & \vdots \\
        \frac{\partial f_{k,L}}{\partial s_1^{k}} & \dots & \frac{\partial f_{k,L}}{\partial s_L^{k}}
    \end{bmatrix}
    =\begin{bmatrix}
        (\nabla_{s^{k}}f_{k,1})^\top \\
        \vdots \\
        (\nabla_{s^{k}}f_{k,L})^\top
    \end{bmatrix}.
    \label{eq:fk_jacobian}
\end{align}
Denote
$$
f_{k,l}^*
=\frac{\pi_l\normal(s^k;\Mu_B^ke_l,\Tau_B^k)}{\sum_{l^*=1}^L\pi_{l^*}\normal(s^k;\Mu_B^ke_{l^*},\Tau_B^k)}
=:\frac{\text{num}_l}{\text{denom}}.
$$
By Quotient rule, we have
\begin{align}
    \nabla_{s^{k}}f_{k,l}^*(s^k)
    &=\frac{(\nabla_{s^{k}}\text{num}_l)(\text{denom})-(\text{num}_l)(\nabla_{s^{k}}\text{denom})}{\text{denom}^2}.\label{eq:quotient}
\end{align}
Using vector calculus, we know that the derivative with relation to $s^k$ is
\begin{align}
    \nabla_{s^k}\normal\left(
    s^k;\Mu_B^ke_l,\Tau_B^k
    \right)
    &=\big(\Tau_B^k\big)^{-1}
    \left(\Mu_B^ke_l-s^k\right)
    \normal\left(
    s^k;\Mu_B^ke_l,\Tau_B^k
    \right).
    \label{eq:pdf_deriv_wrt_all}
\end{align}
Applying it to $\text{num}_l$ and $\text{denom}$, we get
\begin{align*}
    \nabla_{s^{k}}\text{num}_l
    &=\pi_l\cdot\big(\Tau_B^k\big)^{-1}
    \left(\Mu_B^ke_l-s^k\right)
    \normal\left(
    s^k;\Mu_B^ke_l,\Tau_B^k
    \right) \\
    \nabla_{s^{k}}\text{denom}
    &=\sum_{l^*=1}^L\pi_{l^*}\cdot\big(\Tau_B^k\big)^{-1}
    \left(\Mu_B^ke_l-s^k\right)
    \normal\left(
    s^k;\Mu_B^ke_l,\Tau_B^k
    \right).
\end{align*}
substituting the above into \eqref{eq:fk_jacobian} completes the implementation of $f_k'(s^k)$. We use a Monte Carlo approximation for $F^k$.

\subsection{Quantitative Group Testing} \label{sec:imp_details_GT}

For QGT, the AMP algorithm in \eqref{eq:GAMP} has vector iterates, and the implementation is straightforward once $f_k^*$ and $g_k^*$ are specified. As above, $f_k^*$ only depends on the signal distribution, and can be computed using Bayes theorem:
\begin{align}
    f_k^*\big(s\big)
    &=\E\big[\bar{\beta}\big|\mu_\beta^k\bar{\beta}+\sigma_\beta^k G_k=s\big] \nonumber \\
    & = \frac{\mb{P}[\bar{\beta}=1]\cdot\mb{P}[\mu_\beta^k\bar{\beta}+\sigma_\beta^k G_k=s|\bar{\beta}=1]}{\sum_{\bar{\beta}\in\{0,1\}}\mb{P}[\bar{\beta}]\cdot\mb{P}[\mu_\beta^k\bar{\beta}+\sigma_\beta^k G_k=s|\bar{\beta}]} \nonumber \\
    & = \frac{\pi\phi\big( (s-\mu_\beta^k)/\sigma_{k}^\beta\big)}{\pi\phi\big( (s-\mu_\beta^k)/\sigma_{k}^\beta\big)+(1-\pi)\phi(s/\sigma_{k}^\beta)}, \label{eq:f_k_bayes}
\end{align}
where in the last line we have used $\mb{P}[\bar{\beta}=1]=\pi$, and $\phi(x)=\frac{1}{\sqrt{2\pi}}\exp\big(-\tfrac{x^2}{2}\big)$ is the standard normal density. For QGT, following from \eqref{eq:memoryless_GAMP}, $F^{k}=\frac{1}{n}\sum_{j=1}^p f_{k}'(\beta_{j}^{k})$,
where the derivative can be obtained from \eqref{eq:f_k_bayes} by applying the Quotient rule, as done in \eqref{eq:quotient}-\eqref{eq:pdf_deriv_wrt_all}.

For $g_k^*$, noting that $\Sigma^k \in \reals^{2 \times 2}$, we have from \eqref{eq:gk_opt_def} that
\begin{align}
    g_k^*\big(Z^k,\bar{Y}\big)
    =\E[Z|Z^k,\bar{Y}]-\E[Z|Z^k]
    &  \stackrel{(a)}{=} \E[Z|Z^k,\bar{Y}]-(\Sigma^k_{21}/\Sigma^k_{22})Z^k
    \stackrel{(b)}{=} \E[Z|Z^k,\bar{Y}]- Z^k, \label{eq:gk_gen_exp}
\end{align}
where in (a) we used $\E[Z|Z^k]=\frac{\Sigma_{21}^k}{\Sigma_{22}^k}Z^k$, and (b) follows the arguments in \eqref{eq:eqv_of_Sigma_12_21_22} since the AMP uses the Bayes-optimal choice $f_k^*$. In the noiseless setting since $Y=Z$, we have $\E[Z|Z^k,\bar{Y}] = Y$ and $g_k^*\big(Z^k,\bar{Y}\big) = \bar{Y} - Z^k$. From \eqref{eq:memoryless_GAMP}, $C^k$ is given by
\begin{align}
    C^k=\frac{1}{n}\sum_{i=1}^n g_k'(\Theta_{i}^k,\tY_{i}), \label{eq:C^k} 
\end{align}
where $g_{k}'(u, y)$ is the derivative of $g_k$ with respect to its first argument $u$. For the noiseless model, the derivative $g_k^{'*}\big(\Theta_{i}^k,\tY_{i}\big)=-1$ for all $\Theta_{i}^k$ and therefore $C^k=-1$.

\paragraph{$g_k^*$ for  Uniform Noise Distribution} For $g_k^*$ in the noisy case,  from \eqref{eq:gk_gen_exp} we need to  compute $\E[Z|Z^k,\bar{Y}]$
Using $\mb{P}[\cdot]$ to denote the relevant density function, we can write 
\begin{align}
\begin{split}
       \E[Z|Z^k,\bar{Y}]
    &=\E\Big[\E[Z|Z^k,\bar{Y},\Bar{\Psi}]\Big] \\
    &=\int_{\bar{\Psi}}\E[Z|Z^k,\bar{Y},\bar{\Psi}]\cdot \mb{P}[\bar{\Psi}|Z^k,\bar{Y}] \, d\bar{\Psi} \\
    &=\int_{\bar{\Psi}}(\bar{Y}-\bar{\Psi})\cdot\frac{\mb{P}[\bar{\Psi}] \cdot 
    \mb{P}[Z^k,\bar{Y}|\bar{\Psi}]}{\mb{P}[Z^k,\bar{Y}]} \, d\bar{\Psi}. 
\end{split}
\label{eq:Z_Zk_Y_exp1}
\end{align}
Since $\bar{Y} = Z + \bar{\Psi}$, the conditional density
$\mb{P}[Z^k,\bar{Y}|\bar{\Psi}]$ is jointly Gaussian, with $(\bar{Y},Z^k|\bar{\Psi})\sim\normal_2\big([\bar{\Psi},0]^\top,\Sigma^k\big)$. For the simulations in 
Section \ref{sec:QGT_sims}, we used $\Psi_i \stackrel{\iid}{\sim} \text{Uniform}[-\lambda\sqrt{p},\lambda\sqrt{p}]$, which after scaling according to \eqref{eq:QGT_scaling} gives $\mb{P}[\bar{\Psi}]=\Big(2\lambda\sqrt{\frac{1}{\delta\alpha(1-\alpha)}} \Big)^{-1}$.

Under this uniform noise distribution, defining $U:= \lambda\sqrt{\frac{1}{\delta\alpha(1-\alpha)}}$, from \eqref{eq:Z_Zk_Y_exp1} we have
\begin{align}
    \E[Z|Z^k,\bar{Y}]
    &=\int_{-U}^U(\bar{Y}-\bar{\Psi})\cdot \frac{1}{2\lambda/\sqrt{{\delta\alpha(1-\alpha)}}} \cdot\frac{ 
    \mb{P}[Z^k,\bar{Y}|\bar{\Psi}]}{\mb{P}[Z^k,\bar{Y}]} \, d\bar{\Psi}\nonumber\\
    &\stackrel{(a)}{=} \int_{-U}^U(\bar{Y}-\bar{\Psi})\cdot \frac{1}{2\lambda/\sqrt{{\delta\alpha(1-\alpha)}}} \cdot\frac{ 
    \mb{P}[Z^k|Z=\bar{Y}-\bar{\Psi}]\cdot\mb{P}[Z=\bar{Y}-\bar{\Psi}]}{\mb{P}[Z^k,\bar{Y}]} \, d\bar{\Psi} \nonumber \\
    &\stackrel{(b)}{=}\int_{\bar{Y}-U}^{\bar{Y}+U}Z\cdot \frac{1}{2\lambda/\sqrt{{\delta\alpha(1-\alpha)}}} \cdot\frac{ 
    \mb{P}[Z^k|Z]\cdot\mb{P}[Z]}{\mb{P}[Z^k,\bar{Y}]} \, dZ\nonumber\\
    &{=}\int_{\bar{Y}-U}^{\bar{Y}+U}Z\cdot \frac{1}{2\lambda/\sqrt{{\delta\alpha(1-\alpha)}}} \cdot\frac{ 
    \mb{P}[Z|Z^k]}{\mb{P}[\bar{Y}|Z^k]} \, dZ, \label{eq:E_Z_Z^k_Y}
\end{align}
where (b) uses the change of variables $Z=\bar{Y}-\bar{\Psi}$, and (a) follows from 
\begin{align*}
\mb{P}[Z^k,\bar{Y}|\bar{\Psi}]
&= \mb{P}[\bar{Y}|\bar{\Psi}] \cdot \mb{P}[Z^k|\bar{Y},\bar{\Psi}] \\
&= \mb{P}[Z=\bar{Y}-\bar{\Psi}] \cdot \mb{P}[Z^k|Z=\bar{Y}-\bar{\Psi},\bPsi] \\
&= \mb{P}[Z=\bar{Y}-\bar{\Psi}] \cdot \mb{P}[Z^k|Z=\bar{Y}-\bar{\Psi}],
\end{align*}
where the last inequality holds because $(Z,Z^k)$ is independent of $\bPsi$ (as stated below \eqref{eq:g_tilde_k_def}).

Similar to \eqref{eq:Bayes_AMP_gk}, with $[Z,Z^k]^\top\sim\normal(0,\Sigma^k)$, we observe that $\E[Z|Z^k]=\Sigma_{(12)}^k/\Sigma_{(22)}^k$, which gives $\E[Z|Z^k]=Z^k$ for the optimal denoiser $f_k^*$. Consequently, we have $(Z, Z^k)\stackrel{d}{=}(Z^k + \sigma_{Z}^k \widetilde{G}^k,\, Z^k)$, where
\begin{align*}
    (\sigma_{Z}^k)^2 = \Var(Z|Z^k) = \Sigma_{(11)}^k - \frac{(\Sigma_{(12)}^k)^2}{\Sigma_{(22)}^k}\stackrel{(a)}{=}\Sigma_{(11)}^k - \Sigma_{(12)}^k, \quad \widetilde{G}^k \sim\normal(0,1).
\end{align*}
Equality (a) follows from \eqref{eq:eqv_of_Sigma_12_21_22}, since we use the Bayes-optimal denoiser $f_k^*$. Under this distribution of $(Z, Z^k)$, the conditional density  is given by
\begin{align}
    \mb{P}[Z=z|Z^k]
    =\frac{1}{\sqrt{2\pi(\sigma_{Z}^k)^2}}\exp\left(-\frac{(z- Z^k)^2}{2(\sigma_{Z}^k)^2}\right).\label{eq:pZ_Z^k}
\end{align}
The conditional density   $\mb{P}[\bar{Y}|Z^k]$ can be calculated as
\begin{align}
   \mb{P}[\bar{Y}|Z^k] 
   &= \int  \mb{P}[Z, \bar{Y}|Z^k]\, dZ 
   \stackrel{(a)}{=} \int  \mb{P}[\bar{Y}|Z]\cdot \mb{P}[Z|Z^k]\, dZ \nonumber \\
   &= \int_{\bar{Y}-U}^{\bar{Y}+U}\frac{1}{2\lambda/\sqrt{{\delta\alpha(1-\alpha)}}} \cdot  \frac{1}{\sqrt{2\pi(\sigma_{Z}^k)^2}}\exp\left(-\frac{(Z- Z^k)^2}{2(\sigma_{Z}^k)^2}\right) \, dZ, \label{eq:P_Y_Z^k}
\end{align}
 where (a) holds because $\bar{Y}=Z+\bPsi$ is conditionally independent of  $Z^k$, given $Z$. Substituting the expressions in \eqref{eq:pZ_Z^k} and \eqref{eq:P_Y_Z^k} into \eqref{eq:E_Z_Z^k_Y}, we obtain
\begin{align}
    g_k^*\big(Z^k,\bar{Y}\big)
    &=\E[Z|Z^k,\bar{Y}]-Z^k 
    = \frac{\int_{\bar{Y}-U}^{\bar{Y}+U} Z \cdot  \phi\left(\frac{Z-Z^k}{\sigma_{Z}^k}\right) \, dZ}{\int_{\bar{Y}-U}^{\bar{Y}+U} \phi\left(\frac{Z-Z^k}{\sigma_{Z}^k}\right) \, dZ} -Z^k. \label{eq:g_k*Unif}
\end{align}

To compute the term $C^k$ in \eqref{eq:C^k}, the partial derivative 
$\partial_{1}g_k^{*}$ can  computed by applying the Quotient rule on the expression in \eqref{eq:g_k*Unif} and using
\begin{align*}
    &\frac{\partial}{\partial u}\phi\left(\frac{z - u}{\sigma_{Z}^k}\right)
 = \left(\frac{z - u}{(\sigma_{Z}^k)^2}\right)\phi\left(\frac{z - u}{\sigma_{Z}^k}\right),\nonumber\\
 &\frac{\partial}{\partial u}\left(\int_{\bar{Y}-U}^{\bar{Y}+U}   \phi\left(\frac{z-u}{\sigma_{Z}^k}\right) \, dz\right)= \frac{1}{(\sigma_{Z}^k)^2}\left[\int_{\bar{Y}-U}^{\bar{Y}+U} z \phi\left(\frac{z-u}{\sigma_{Z}^k}\right) \, dz - u\int_{\bar{Y}-U}^{\bar{Y}+U} \phi\left(\frac{z-u}{\sigma_{Z}^k}\right) \, dz\right], \nonumber\\
 &\frac{\partial}{\partial u}\left(\int_{\bar{Y}-U}^{\bar{Y}+U} z   \phi\left(\frac{z-u}{\sigma_{Z}^k}\right) \, dz\right)= \frac{1}{(\sigma_{Z}^k)^2}\left[\int_{\bar{Y}-U}^{\bar{Y}+U} z^2 \phi\left(\frac{z-u}{\sigma_{Z}^k}\right) \, dz - u\int_{\bar{Y}-U}^{\bar{Y}+U} z  \phi\left(\frac{z-u}{\sigma_{Z}^k}\right) \, dz\right],\nonumber
\end{align*}
resulting in
\begin{align}
    \partial_1g_k^{*}\big(Z^k,\bar{Y}\big) &= \frac{\partial}{\partial Z^k}\left(\E[Z|Z^k,\bar{Y}]-Z^k\right)\nonumber\\
    &= \frac{1}{(\sigma_{Z}^k)^2}\left[\frac{\int_{\bar{Y}-U}^{\bar{Y}+U} Z^2 \cdot  \phi\left(\frac{Z-Z^k}{\sigma_{Z}^k}\right) \, dZ}{\int_{\bar{Y}-U}^{\bar{Y}+U} \phi\left(\frac{Z-Z^k}{\sigma_{Z}^k}\right) \, dZ} - {\left(\frac{\int_{\bar{Y}-U}^{\bar{Y}+U} Z \cdot  \phi\left(\frac{Z-Z^k}{\sigma_{Z}^k}\right) \, dZ}{\int_{\bar{Y}-U}^{\bar{Y}+U} \phi\left(\frac{Z-Z^k}{\sigma_{Z}^k}\right) \, dZ}\right)}^2\, \right]\, - 1. \label{eq:deriv_g_k*Unif}
\end{align}
The integrals in \eqref{eq:g_k*Unif}, \eqref{eq:deriv_g_k*Unif} can be expressed using the standard normal cumulative distribution function:
$\Phi(x)=\int_{-\infty}^x \phi(x)\, dx$ 
\begin{align*}
    &I_1 = \int_{\bar{Y}-U}^{\bar{Y}+U} \phi\left(\frac{Z-Z^k}{\sigma_{Z}^k}\right) \, dZ = \Phi\left(\frac{\bar{Y}+ U - Z^k}{\sigma_{Z}^k}\right) - \Phi\left(\frac{\bar{Y}- U - Z^k}{\sigma_{Z}^k}\right),\\
    &I_2 = \int_{\bar{Y}-U}^{\bar{Y}+U} Z\cdot\phi\left(\frac{Z-Z^k}{\sigma_{Z}^k}\right) \, dZ = Z^k I_1 + \sigma_{Z}^k\left[\phi\left(\frac{\bar{Y}- U - Z^k}{\sigma_{Z}^k}\right) - \phi\left(\frac{\bar{Y}+ U - Z^k}{\sigma_{Z}^k}\right)\right], \\
    &I_3 = \int_{\bar{Y}-U}^{\bar{Y}+U} Z^2\cdot\phi\left(\frac{Z-Z^k}{\sigma_{Z}^k}\right) \, dZ \\
    &\hspace{1em}=\left[(Z^k)^2 + (\sigma_{Z}^k)^2\right] I_1  + \left[2Z^k\sigma_{Z}^k + (\sigma_{Z}^k)^2\left(\frac{\bar{Y}- U - Z^k}{\sigma_{Z}^k}\right)\right]\phi\left(\frac{\bar{Y}- U - Z^k}{\sigma_{Z}^k}\right) \\
    & \hspace{6em} - \left[2Z^k\sigma_{Z}^k + (\sigma_{Z}^k)^2\left(\frac{\bar{Y}+ U - Z^k}{\sigma_{Z}^k}\right)\right]\phi\left(\frac{\bar{Y}+ U - Z^k}{\sigma_{Z}^k}\right).
\end{align*}

\section{Derivation of Optimization Methods in Section \ref{sec:sim_pool_data}} \label{sec:optm_derivations}

For the linear program (designed for the noiseless setting), we start with maximizing the log-likelihood of the signal:
\begin{align*}
    \text{maximize}&\quad \log \mb{P}[B]\quad(\text{w.r.t.~ $B_\opt$}) \\
    \text{subject to}
    &\quad B_{\opt}\in\{0,1\}^{pL}, 
    \quad Y_{\opt}=X_{\opt}B_{\opt},
    \quad C_{\opt}B_{\opt}=1_p,
\end{align*}
where $1_p$ is the $p$-length vector of ones, and the other matrices are defined in \eqref{eq:opt_defs}. Note that while this objective is optimal for exact recovery, it is computationally challenging due to the integer constraints and the large dimension $p$. Hence, this suggests the need for relaxation. Let us denote $\mathcal{I}_l$ as the set of items in category $l$, $p_l$ as the number of items in the $l$th category, and $I_{pL}^{(pl)}$ as the sub-matrix of $I_{pL}$ obtained by taking the $p(l-1)$-th to $pl$-th rows of $I_{pL}$. Recalling that each item belongs to category $l \in [L]$ with probability $\pi_l$, we  simplify the objective function as follows:
\begin{align*}
    \log\mb{P}[B]
    &\stackrel{(a)}{=}\sum_{j=1}^p\log\mb{P}[B_{j,:}]
    =\sum_{j\in\mathcal{I}_1}\log\mb{P}[B_{j,:}]+\dots+\sum_{j\in\mathcal{I}_L}\log\mb{P}[B_{j,:}] \\
    &=\sum_{l=1}^L\sum_{j\in\mathcal{I}_l}\log\pi_l
    =\sum_{l=1}^Lp_l\log\pi_l,
\end{align*}
where (a) follows from the independence in distribution between items. We then relax the integer program to a linear program. This is done by relaxing $p_l$ from a count to the sum of the $l$th column of $B$, which gives
\begin{align}
    \sum_{l=1}^Lp_l\log\pi_l
   \,  \stackrel{\text{relaxed to}}{\implies} \, 
    \sum_{l=1}^L(\log\pi_l)1_p^\top I_{pL}^{(pl)}B_\opt,
    \label{eq:P[B]_relax}
\end{align}
and changing the integer constraints, which gives
\begin{align}
    B_{\opt}\in\{0,1\}^{pL} \, 
    \stackrel{\text{relaxed to}}{\implies} \, 
    0_p\leq B_{\opt}\leq 1_p.
    \label{eq:B_opt_relax}
\end{align}
This gives us the simplified linear program:
\begin{align*}
    \text{minimize}&\quad -\sum_{l=1}^L(\log\pi_l)1_p^\top I_{pL}^{(pl)}B_\opt\quad(\text{w.r.t.~ $B_\opt$}) \\
    \text{subject to}
    &\quad 0_p\leq B_{\opt}\leq 1_p, 
    \quad Y_{\opt}=X_{\opt}B_{\opt}, 
    \quad C_{\opt}B_{\opt}=1_p.
\end{align*}
For the convex program (designed for Gaussian additive noise), using the MAP rule for estimating the signal $B$ from the observations $Y$ gives
\begin{align*}
    \text{maximize}&\quad \log ( \mb{P}[B]\cdot\mb{P}[Y|B])
    \quad(\text{w.r.t.~ $B_\opt$}) \\
    \text{subject to}&
    \quad B_{\opt}\in\{0,1\}^{pL}, \quad C_{\opt}B_{\opt}=1_p.
\end{align*}
The objective can be further simplified as
\begin{align*}
    \log\mb{P}[B]+\log\mb{P}[Y|B]
    &\stackrel{(a)}{=}\log\mb{P}[B]+\sum_{i,l}\log\mb{P}[Y_{i,l}|B] \\
    &=\log\mb{P}[B]+\sum_{i,l}\log\left\{\frac{1}{\sigma\sqrt{2p\pi}}\exp\left[-\frac{1}{2}\bigg(\frac{Y_{i,l}-X_{i,:}^\top B_{:,l}}{\sigma\sqrt{p}}\bigg)^2\right]\right\} \\
    &=\log\mb{P}[B]-\frac{1}{2p\sigma^2}\|Y_{\opt}-X_{\opt}B_{\opt}\|_2^2+\text{constant},
\end{align*}
where (a) is because given $B$, each entries of $Y$ are independent following the distribution $Y_{il}|B\sim \normal(X_{i,:}^\top B_{:,l},p\sigma^2)$. Using the simplification and relaxation of $\mb{P}[B]$ in \eqref{eq:P[B]_relax} and the relaxation in \eqref{eq:B_opt_relax}, gives us the convex program:
\begin{align*}
    \text{minimize}&\quad \frac{1}{2p\sigma^2}\|Y_{\opt}-X_{\opt}B_{\opt}\|_2^2-\sum_{l=1}^L(\log\pi_l)1_p^\top I_{pL}^{(pl)}B_\opt
    \quad(\text{w.r.t.~ $B_\opt$}) \\
    \text{subject to}&\quad 0\leq B_\opt \leq 1,  \quad C_{\opt}B_{\opt}=1_p.
\end{align*}

\section{The scaling of $\tY$ for pooled data} \label{sec:scaling_deriv}

Consider the noiseless pooled data problem where $Y =XB \in \reals^{n \times L}$. The entries of $\tY$, defined in \eqref{eq:tYtPsi_def}, are given by $\tY_{il} = \frac{Y_{il}-\alpha p\hat{\pi}_l}{\sqrt{n\alpha(1-\alpha)}} $, for $i \in [n], l \in [L]$. We show that $\tY_{il}=\Theta(1)$ with high probability. Using the independence of $X$ and $B$ we have
\begin{align}
    \E[Y_{il}-\alpha p \hat{\pi}_l]
   % &=\E[Y_{il}]-\alpha p\E[\hat{\pi}_l] \nonumber \\
    &=\E\Big[\sum_{j=1}^pX_{ij}B_{jl}\Big]-\alpha p\E\Big[\frac{1}{p}\sum_{j=1}^pB_{jl}\Big] \nonumber \\
    &=\sum_{j=1}^p\E[X_{ij}]\E[B_{jl}]-\alpha p\bigg(\frac{1}{p}\sum_{j=1}^p\E[B_{jl}]\bigg) \nonumber \\
    &=\sum_{j=1}^p\alpha\E[\bar{B}_{l}]-\alpha p\E[\bar{B}_l]
    =0. \label{eq:tY_num_mean}
\end{align}
Next, we compute the variance of $Y_{il}$:
\begin{align}
    \Var[Y_{il}-\alpha p\hat{\pi}_l]
    &=\Var\Big[\sum_{j=1}^pX_{ij}B_{jl}-\alpha\sum_{j=1}^pB_{jl}\Big]
    =\Var\Big[\sum_{j=1}^pB_{jl}(X_{ij}-\alpha)\Big] \nonumber \\
    &=\sum_{j=1}^p\Var[B_{jl}(X_{ij}-\alpha)]+2\sum_{1\leq j'<j\leq p}\Cov[B_{jl}(X_{ij}-\alpha),B_{j'l}(X_{ij'}-\alpha)] \nonumber \\
    &\stackrel{(\rm a)}{=}\sum_{j=1}^p\Var[B_{jl}X_{ij}-B_{jl}\alpha] \nonumber \\
    &=\sum_{j=1}^p\bigg(\Var[B_{jl}X_{ij}]+\Var[B_{jl}\alpha]-2\Cov[B_{jl}X_{ij},B_{jl}\alpha]\bigg) \nonumber \\
    &=\sum_{j=1}^p\bigg(\E[B_{jl}^2]\E[X_{ij}^2]-(\E[B_{jl}])^2(\E[X_{ij}])^2+\alpha^2\Var[B_{jl}] \nonumber \\
    &\qquad-2\alpha\Big(\E[B_{jl}X_{ij}B_{jl}]-\E[B_{jl}X_{ij}]\E[B_{jl}]\Big)\bigg) \nonumber \\
    &\stackrel{\rm (b)}{=}\sum_{j=1}^p\bigg(\pi_l\alpha-(\pi_l)^2\alpha^2+\alpha^2\pi_l(1-\pi_l)-2\alpha(\pi_l\alpha-(\pi_l)^2\alpha)\bigg) \nonumber \\
    &=p\pi_l\alpha(1-\alpha)
    =\Theta(p), \label{eq:tY_num_var}
\end{align}
where in (a), the second term is zero because $B_{jl}$, $B_{j'l}$, $X_{ij}$, and $X_{ij'}$ are all generated independently from each other, and in (b) we substituted $\E[B_{jl}^2]=\pi_l$ and $\E[X_{ij}^2]=\alpha$. Since $\tY_{i,:}=\frac{Y_{i,:}-\alpha p \hat{\pi}}{\sqrt{n\alpha(1-\alpha)}}$ and $n/p \to \delta$, \eqref{eq:tY_num_mean} and \eqref{eq:tY_num_var} imply that $\E[\tY_{il}]=0$ and $\Var[\tY_{il}]=\Theta(1)$, for each $i \in [n], l \in [L]$.
%This further implies that $\tY=\Theta(1)$ with high probability.

Next, we show that a constant shift in the entries of $\hat{\pi}$ leads to the scaling of $\tY_{il}$ to jump from $\Theta(1)$ to $\Theta(\sqrt{p})$. Let us consider the case where we replace $\hat{\pi}_l$ with $\hat{\pi}_l+\epsilon$ for some positive constant $\epsilon$, so that  $\tY_{il} = \frac{Y_{il}-\alpha p(\hat{\pi}_l + \epsilon)}{\sqrt{n\alpha(1-\alpha)}} $. We  then have
\begin{align*}
    \E[Y_{il}-\alpha p (\hat{\pi}_l+\epsilon)]=-\alpha p\epsilon=-\Theta(p),
\end{align*}
and 
\begin{align*}
    \Var[Y_{il}-\alpha p (\hat{\pi}_l+\epsilon)]
    =\Var[Y_{il}-\alpha p \hat{\pi}_l]
    =\Theta(p).
\end{align*}
This implies that $\E[\tY_{il}]=\Theta(\sqrt{p})$ and $\Var[\tY_{il}]=\Theta(1)$, which implies that $\tY_{il}=\Theta(\sqrt{p})$ with high probability.

\end{document}